%% file: tri-count-approx.tex
\documentclass[11pt]{article}
\usepackage{etex}

\usepackage[margin=2.54cm]{geometry}
\usepackage[dvipsnames,usenames]{color}
\usepackage{amsfonts,amsmath,amssymb,amsthm}

\usepackage[margin=10pt]{subfig}

\usepackage{paralist}
\usepackage{algorithmic,algorithm}
\usepackage{bm}
\usepackage{xspace}
\usepackage{xspace,prettyref}
\usepackage{color}
\usepackage{dsfont}

\def\focs{0}
\usepackage[pagebackref,letterpaper=true,colorlinks=true,pdfpagemode=none,urlcolor=blue,linkcolor=blue,citecolor=BrickRed,pdfstartview=FitH]{hyperref}

\newcommand{\ignore}[1]{}

\newtheorem{theorem}{Theorem}

\newtheorem{lemma}[theorem]{Lemma}
\newtheorem{claim}[theorem]{Claim}

\newtheorem{proposition}[theorem]{Proposition}
\newtheorem{corollary}[theorem]{Corollary}

\newtheorem{definition}{Definition}
\newtheorem{assumption}{Assumption}
\theoremstyle{definition}

 \newcommand{\EX}{\hbox{\bf E}}
\newcommand{\pr}{{\rm Pr}}
 \newcommand{\var}{\hbox{\bf Var}}
\renewcommand{\th}{^\textrm{th}}

\newcommand{\wht}{\widehat{t}}

\def\eqdef{~\triangleq~}
\def\eps{\epsilon}
\def\bar{\overline}

\def\ceil#1{\lceil {#1} \rceil}

\newenvironment{proofof}[1]{\smallskip\noindent{\bf Proof of #1:}}%
        {\hspace*{\fill}$\Box$\par}

\newcommand{\om}{\overline{m}}
\newcommand{\ot}{\overline{t}}
\newcommand{\lt}{\widetilde{t}}

\newcommand{\ce}{c_H}

\newcommand{\tr}{{t}}
\newcommand{\Tr}{{T}}

\newcommand{\cE}{{\cal E}}

\newcommand{\Sec}[1]{\hyperref[sec:#1]{\S\ref*{sec:#1}}} %section
\newcommand{\Eqn}[1]{\hyperref[eq:#1]{(\ref*{eq:#1})}} %equation
\newcommand{\Fig}[1]{\hyperref[fig:#1]{Fig.\,\ref*{fig:#1}}} %figure
\newcommand{\Tab}[1]{\hyperref[tab:#1]{Tab.\,\ref*{tab:#1}}} %table
\newcommand{\Thm}[1]{\hyperref[thm:#1]{Theorem\,\ref*{thm:#1}}} %theorem
\newcommand{\Lem}[1]{\hyperref[lem:#1]{Lemma\,\ref*{lem:#1}}} %lemma
\newcommand{\Prop}[1]{\hyperref[prop:#1]{Prop.~\ref*{prop:#1}}} %property
\newcommand{\Cor}[1]{\hyperref[cor:#1]{Corollary~\ref*{cor:#1}}} %corollary
\newcommand{\Def}[1]{\hyperref[def:#1]{Definition~\ref*{def:#1}}} %definition
\newcommand{\Alg}[1]{\hyperref[alg:#1]{Alg.~\ref*{alg:#1}}} %algorithm
\newcommand{\Ex}[1]{\hyperref[ex:#1]{Ex.~\ref*{ex:#1}}} %example
\newcommand{\Clm}[1]{\hyperref[clm:#1]{Claim~\ref*{clm:#1}}} %claim
\newcommand{\Assum}[1]{\hyperref[assump:#1]{Assumption~\ref*{assump:#1}}} %assumption

\def\poly{{\rm poly}}
\newcommand{\polylog}{\poly(\log n, 1/\eps)}

\newcommand{\wt}{\hbox{wt}}

\newcommand{\heav}{{\tt Heavy}}
\newcommand{\triest}{{\tt Estimate-with-advice}}
\newcommand{\triestt}{{\tt Estimate}}

\newcommand\numberthis{\addtocounter{equation}{1}\tag{\theequation}}

\newcommand{\mA}{\mathcal{A}}
\newcommand{\mG}{\mathcal{G}}
\newcommand{\mD}{\mathcal{D}}

\newcommand{\oW}{\overline{W}}
\newcommand{\tmD}{\widetilde{\mathcal{D}}}
\newcommand{\tP}{{\widetilde{P}}}

\newcommand{\ora}[1]{\overrightarrow{#1}}

\newcommand{\aux}{\mathcal{A}_{\pi,(u,v)}}

\newcommand{\lf}{\left\lfloor}
\newcommand{\rf}{\right\rfloor}
\newcommand{\sq}{\langle u,v,u',v'\rangle}
\newcommand{\altSq}{\langle x,y,x',y'\rangle}

\newcommand{\sect}[2][]{\section#1{#2}}
\newcommand{\subsect}[2][]{\subsection#1{#2}}
\newcommand{\subsubsect}[2][]{\subsubsection#1{#2}}

\renewcommand{\Pr}{\mathrm{Pr}}
\def\nofigures{0}

\ifnum\nofigures=0
\usepackage{pstricks,tikz}
\usetikzlibrary{decorations.markings}
\usetikzlibrary{patterns,snakes}
\pgfdeclarelayer{background}
\pgfsetlayers{background,main}
\input{Drawings}

\fi

\def\withcolors{0}
\ifnum\withcolors=1
\newcommand{\dchange}[1]{{{#1}}}
\newcommand{\ddchange}[1]{{\color{cyan}{#1}}}
\newcommand{\tchange}[1]{{\color{purple}{#1}}}

\newcommand{\dddchange}[1]{{\color{cyan}{#1}}}
\else
\newcommand{\dchange}[1]{{{#1}}}
\newcommand{\ddchange}[1]{{{#1}}}
\newcommand{\tchange}[1]{{{#1}}}

\newcommand{\dddchange}[1]{{#1}}
\fi

	\title{Approximately Counting Triangles in Sublinear Time
	}
	\author{Talya Eden\footnote{School of Computer Science, Tel Aviv University, {\tt talyaa01@gmail.com}} \and Amit Levi\footnote{School of Electrical Engineering, Tel Aviv University, {\tt amitlev3@post.tau.ac.il}}\and Dana Ron\footnote{School of Electrical Engineering, Tel Aviv University, {\tt danaron@tau.ac.il}.
This research was partially supported by the Israel Science Foundation grant No.~671/13 and
by a grant from the Blavatnik fund.} \and  C. Seshadhri \footnote{University of California, Santa Cruz, {\tt scomandu@ucsc.edu}}}

\begin{document}

\begin{titlepage}

\maketitle
\begin{abstract}
	We consider the problem of estimating the number of triangles in a graph. This problem has been extensively studied in both theory
    and practice, but all existing algorithms read the entire graph. In this work we design a {\em sublinear-time\/} algorithm for approximating the number of triangles in a graph, where the algorithm is given query access to the graph. The allowed queries are degree queries, vertex-pair queries and neighbor queries.
	
	We show that for any given approximation parameter $0<\epsilon<1$, the algorithm provides an estimate $\widehat{\tr}$ such that with high constant probability, $(1-\epsilon)\cdot \tr< \widehat{\tr}<(1+\epsilon)\cdot \tr$, where $t$ is the number of triangles in the graph $G$. The expected query complexity of the algorithm is $\!\left(\frac{n}{\tr^{1/3}} + \min\left\{m, \frac{m^{3/2}}{\tr}\right\}\right)\cdot {\rm poly}(\log n, 1/\epsilon)$, where $n$ is the number of vertices in the graph and $m$ is the number of edges,
	and the expected running time is
	$\!\left(\frac{n}{\tr^{1/3}} + \frac{m^{3/2}}{\tr}\right)\cdot {\rm poly}(\log n, 1/\epsilon)$. We also prove that
	$\Omega\!\left(\frac{n}{\tr^{1/3}} + \min\left\{m, \frac{m^{3/2}}{\tr}\right\}\right)$ queries are necessary,
	thus  establishing that the query complexity of this algorithm is optimal up to polylogarithmic factors in $n$ (and the dependence on $1/\epsilon$).
\end{abstract}

%\title{A simpler sublinear algorithm for approximating the triangle count}
%\date{}

%\author{C. Seshadhri
% \\ {\tt scomandu@ucsc.edu}\\
%University of California, Santa Cruz
%}
%\begin{abstract} A recent result of Eden, Levi, and Ron (ECCC 2015) provides a sublinear time
%algorithm to estimate the number of triangles in a graph. Given an undirected graph $G$,
%one can query the degree of a vertex, the existence of an edge between vertices,
%and the $i$th neighbor of a vertex. Suppose the graph has $n$ vertices, $m$ edges, and $t$
%triangles. In this model, Eden et al provided a $O(\poly(\eps^{-1}\log n)(n/\tr^{1/3} + m^{3/2}/\tr))$
%time algorithm to get a $(1+\eps)$-multiplicative approximation for $t$, the triangle count.
%This paper provides a simpler algorithm with the same running time (up to differences
%in the $\poly(\eps^{-1}\log n)$ factor) that has a substantially simpler analysis.
%\end{abstract}

\end{titlepage}

\tableofcontents \thispagestyle{empty}

\newpage
\input{Introduction}
\input{UpperBound}

\input{LowerBound}

		\bibliographystyle{alpha}
%		\bibliography{triangles,Triapprox,TrianglesAreImportant}
		\bibliography{triangles_bib}
		
	\end{document}

%% file: Drawings.tex
\usepackage{graphicx,amsfonts,amsmath,amssymb,epsfig,color,multicol,pstricks,tikz, pgf}
\usetikzlibrary{decorations.markings}
\usetikzlibrary{patterns}
\usetikzlibrary{calc, intersections}
\usepackage{pgf,tikz}
\usetikzlibrary{calc}
\usepackage{tkz-euclide}
\usepackage{pgfplots}
\usetkzobj{all}

\pgfdeclarelayer{background}
\pgfsetlayers{background,main}

\title{Counting Triangles in Sub-linear Time}
\date{}

\tikzset{dotted pattern/.style args={#1}{
		postaction=decorate,
		decoration={
			markings,
			mark=between positions 0.25 and 0.75 step 0.25 with {
				\fill[radius=#1] (0,0) circle;
			}
		}
	},
	dotted pattern/.default={1pt},
}

	\tikzstyle{client}=[draw, circle, minimum size=2ex, inner sep=0,fill=black] 
	\tikzstyle{server}=[draw, circle, minimum size=2ex, inner sep=0,fill=black]
	\tikzstyle{cycle}=[draw,circle,inner sep=0,minimum size=2ex,fill=black]
	\tikzstyle{small cycle}=[draw,circle,inner sep=0,minimum size=.5ex,fill=black]
	\newcommand\DrawCompleteBipartite{
		\foreach \i in {1,...,3} {
			\node[client] (l\i) at (0,\i) {} ;
			\node[server] (r\i) at (2,\i) {} ;
		}
		\foreach \i in {5,...,10} {
			\node[client] (l\i) at (0,\i) {} ;
			\node[server] (r\i) at (2,\i) {} ;
		}
		
		\path[dotted pattern=1.0pt] (l3) -- (l5);
		\path[dotted pattern=1.0pt] (r3) -- (r5);
		
		\foreach \i in {1,2,3,5,6,7,8,9,10} {
			\foreach \j in {1,2,3,5,6,7,8,9,10} {
				\ifthenelse {\i = \j}	{ } {
					\draw[ultra thin] (l\i) -- (r\j);
				};
			}
		}
		
		\ifthenelse{\matching=1}
		{ %then
			\foreach \i/\j in {l1/l2, l5/l6, l7/l8} {
				\draw[ultra thick,blue] (\i) to[bend left=40] (\j); 
			}
			\foreach \i/\j in {r1/r2, r5/r6, r7/r8} {
				\draw[ultra thick,blue] (\i) to[bend right=40] (\j); 
			}
			
			\foreach \i in {1,2,5,6,7,8} {
				\draw[ultra thick,red] (l\i) -- (r\i);
	
			}
			\foreach \i in {3,9,10} {
				\draw[thick] (l\i) -- (r\i);
			}          
			
		} %end then
		{ %else
			\foreach \i in {1,2,3,5,6,7,8,9,10} {
				\draw[thin] (l\i) -- (r\i);
			}          
		}; %end else
	}

\newcommand{\DrawTwoFamilies}{
	\begin{tikzpicture}[x=1cm,scale=0.35]
\tikzstyle{cycle}=[draw,circle,inner sep=0.1cm, fill=black]
\tikzstyle{client}=[draw,circle,inner sep=0.1cm, fill=black]
\tikzstyle{server}=[draw,circle,inner sep=0.1cm, fill=black]
\begin{scope}[shift={(-1.5,0)}]
	\begin{scope}[shift={(2.5,-1.5)},scale=1.5]
			\DrawIndependentSet
			\draw (3.3,4.5) -- (3.3,-2.5);
	\end{scope}
	\begin{scope}[shift={(1.5,-0.5)},scale=1.5]
			\DrawIndependentSet
	\end{scope}
	\begin{scope}[shift={(-0.5,-1.5)},scale=1.5]
			\DrawIndependentSet
	\end{scope}
	\draw [thick, decoration={ brace, raise=0.5cm,amplitude=10pt}, decorate] (-3.5,-4) -- (-3.5,4.7) node [midway, left=0.5cm+10pt] {$n$} ;
\end{scope}		
				
	\begin{scope}[shift={(16,0)},scale=0.8] %Draw bipartite comonent
		\draw [thick, decoration={ brace, raise=0.5cm,amplitude=10pt}, decorate] 
		(-5,-6) -- (-5,6) node [midway, left=0.5cm+10pt] {$\Delta^{1/3}$} ;
			\foreach \x in {1,...,8}{
				\pgfmathparse{(\x-1)*45+floor(\x/9)*22.5}
				\node[draw,circle,inner sep=0.1cm, fill=black] (N-\x) at (\pgfmathresult:5.4cm) {};
			}
			
			\foreach \x [count=\xi from 1] in {1,...,8}{
				\foreach \y in {\x,...,8}{
					\path (N-\xi) edge[-] (N-\y);
				}
			}

	\end{scope}
	\begin{scope}[shift={(25.5,-1.5)},scale=1.5]
		\DrawIndependentSet
			\begin{scope}[shift={(0.8,-0.5)}]
				\DrawIndependentSet
			\end{scope}
		\draw [thick, decoration={ brace, raise=0.5cm,amplitude=10pt}, decorate] (2.3,3.5) -- (2.3,-2) node [midway, right=0.5cm+10pt] {$n-\Delta^{1/3}$} ;
	\end{scope}

	\end{tikzpicture}
}
	\tikzstyle{witness}=[draw, circle, minimum size=1.ex, inner sep=0,fill=black] 			
	\tikzstyle{non_witness}=[draw, circle, minimum size=1.ex, inner sep=0,fill=black] 			
	\newcommand{\goodbadgraph}[3] % {x_origin}{y_origin}{is_right}{label}
	{
		\node [witness, label=left:{}] (origin) at #1 {};
		\node [non_witness, label=right:{}] (upper) at ($#1+(#2*1.5,1.5)$) {};
		\node [non_witness, label=right:{}] (mid) at ($#1+(#2*1.5,1)$) {};
		\node [non_witness, label=right:{}] (lower) at ($#1+(#2*1.5,-0.7)$) {};
		
		\draw (mid)--(origin);
		\draw (upper)--(origin);
		\draw (lower)--(origin);
					
		\ifthenelse{#2=1}
		{
			\tkzMarkAngle[size=0.5, dotted ](lower,origin,mid);
			\coordinate [label=right:$d_{#3}$] (label) at ($#1+(0.5,0)$);
			\coordinate [] (additional) at ($#1+(#2*1.5,1.5)$);
			\draw (additional)--(origin);
		}
		{
			\tkzMarkAngle[size=0.5, dotted,thick](mid,origin,lower);
			\coordinate [label=right:$d_{#3}$] (label) at ($#1+(-2.,0)$);
		};
}
	
	\newcommand{\DrawBadAndGoodBipartite} {
		\begin{scope}[]
		\draw [thick, decoration={ brace, raise=0.5cm,amplitude=10pt}, decorate] (0,-4) -- (0,.5) node [midway, left, xshift=-0.5cm-10pt,align=center] {witness \\ graphs} ;
		\goodbadgraph{(0,0)}{1}{w};
	    \goodbadgraph{(0,-3.4)}{1}{w};
		\path[dotted pattern] (0,0)--(0,-3.2);

		\draw [thick, decoration={ brace, raise=0.5cm,amplitude=10pt}, decorate] (1.5,1.5) -- (1.5,-4.55) node [midway, right=0.5cm+10pt, align=center] {non-witness\\ graphs} ;
		\end{scope}
	}

	\tikzstyle{client}=[draw, circle, minimum size=1ex, inner sep=0,fill=black] 
	\tikzstyle{server}=[draw, circle, minimum size=1ex, inner sep=0,fill=black] 
	
	\tikzstyle{cycle}=[draw,circle,radius=0.1,fill=black]

	\newcommand\DrawBaseGraph
	{	
		\foreach \i in {2,...,5} {
			\node[client] (l\i) at (0,\i) {} ;
			\node[server] (r\i) at (3,\i) {} ;
		}
		
		\node[client,label=left:$u$] (l1) at (0,0) {} ;
		\node[server,label=right:$v$] (r1) at (3,0) {} ;
		
		\node[server,label=left:$u'$] (label_u') at (0,2) {} ;
		\node[server,label=right:$v'$] (label_v') at (3,2) {} ;
	
		\path [dotted pattern] (l1)--(l2);
		\path [dotted pattern] (r1)--(r2);
	}

	\tikzstyle{red match} = [color=red,ultra thick,dashed]
	\newcommand\DrawFix {
	\begin{scope}[shift={(-8,0)}]

		\DrawBaseGraph
		\node at (1.5,-1) {A witness graph $W$};
		\draw[red match] (l5)--(r3);
		\draw[red match] (l4)--(r5);
		\draw[red match] (l3)--(r4);
		\draw[red match] (l2)--(r2);
		\draw[red match] (l1)--(r1);
		\draw [ultra thick,label={[above]{fix}}] (3.8,2.7)--(6.2,2.7)
			node[above, midway,align=center] {edge in \\ $\aux$};
		 \end{scope}
		 	
		\begin{scope}[shift={(-1,0)}]
		\DrawBaseGraph
		\node at (1.5,-1) {$\oW$ -- A neighbor of $W$};
		\draw[red match] (l5)--(r3);
		\draw[red match] (l4)--(r5);
		\draw[red match] (l3)--(r4);
		\draw[red match] (l2)--(r1);
		\draw[red match] (l1)--(r2);
		\end{scope}		
	}

	\newcommand\DrawSmallkFixArg {
	
		\node[client, label=left:$u$] (l1) at (0,1) {} ;
		\node[server, label=right:$v$] (r1) at (2,1) {} ;
		\node[client, label=left:$u'$] (l2) at (0,2) {} ;
		\node[server, label=right:$v'$] (r2) at (2,2) {} ;
		
		\foreach \i in {3,...,6} {
			\node[client] (l\i) at (0,\i+1) {} ;
			\node[server] (r\i) at (2,\i+1) {} ;
		}
		
		\node[client, label=left:$x$] (l5) at (0,6) {} ;
		\node[server, label=right:$y$] (r5) at (2,6) {} ;
		\node[client, label=left:$x'$] (l6) at (0,7) {} ;
		\node[server, label=right:$y'$] (r6) at (2,7) {} ;
				
		\path[dotted pattern=1.0pt] (l2) -- (l3);
		\path[dotted pattern=1.0pt] (r2) -- (r3);
		
		\foreach \i in {1,...,6} {
			\foreach \j in {1,...,6} {
				\ifthenelse {\i = \j}	{ } {
					\draw[ultra thin] (l\i) -- (r\j);
				};
			}
		}
		\draw[ultra thin] (l5) -- (r5);
		\draw[ultra thin] (l6) -- (r6);
		\foreach \i/\j in {l1/l2, l3/l4} {
			\draw[ultra thick,blue] (\i) to[bend left=40] (\j); 
		}
		\foreach \i/\j in {r1/r2, r3/r4} {
			\draw[ultra thick,blue] (\i) to[bend right=40] (\j); 
		}
		\foreach \i/\j in {l1/r1, l2/r2, l3/r3, l4/r4} {
			\draw[ultra thick,red, dashed] (\i) to (\j); 
		}
		\begin{scope}[shift={(5,0)}]
			\node[client, label=left:$u$] (l1) at (0,1) {} ;
			\node[server, label=right:$v$] (r1) at (2,1) {} ;
			\node[client, label=left:$u'$] (l2) at (0,2) {} ;
			\node[server, label=right:$v'$] (r2) at (2,2) {} ;
			
			\foreach \i in {3,...,6} {
				\node[client] (l\i) at (0,\i+1) {} ;
				\node[server] (r\i) at (2,\i+1) {} ;
			}

			\node[client, label=left:$x$] (l5) at (0,6) {} ;
			\node[server, label=right:$y$] (r5) at (2,6) {} ;
			\node[client, label=left:$x'$] (l6) at (0,7) {} ;
			\node[server, label=right:$y'$] (r6) at (2,7) {} ;
			
			\path[dotted pattern=1.0pt] (l2) -- (l3);
			\path[dotted pattern=1.0pt] (r2) -- (r3);
			
			\foreach \i in {1,...,6} {
				\foreach \j in {1,...,6} {
					\ifthenelse {\i = \j}	{ } {
						\draw[ultra thin] (l\i) -- (r\j);
					};
				}
			}
			\draw[ultra thin] (l1) -- (r1);
			\draw[ultra thin] (l2) -- (r2);
			
			\foreach \i/\j in {l5/l6, l3/l4} {
				\draw[ultra thick,blue] (\i) to[bend left=40] (\j); 
			}
			\foreach \i/\j in {r5/r6, r3/r4} {
				\draw[ultra thick,blue] (\i) to[bend right=40] (\j); 
			}
			\foreach \i/\j in {l5/r5, l6/r6, l3/r3, l4/r4} {
				\draw[ultra thick,red, dashed] (\i) to (\j); 
			}
		\end{scope}
		\draw [<->, ultra thick,label={[above]{fix}}] (2.8,4)--(4.2,4);
		\node at (3.5,4.4) {A switch};
	
	}

		\tikzstyle{client}=[draw, circle, minimum size=2ex, inner sep=0,fill=black] 
		\tikzstyle{server}=[draw, circle, minimum size=2ex, inner sep=0,fill=black]
		\tikzstyle{cycle}=[draw,circle,inner sep=0,minimum size=2ex,fill=black]
		\tikzstyle{small cycle}=[draw,circle,inner sep=0,minimum size=.5ex,fill=black]
		\newcommand\DrawDoubleBipartite{
			\foreach \i/\j/\z in {1/2/1, 5/6/2, 9/10/3} {
				\node[client, label=u\z] (l\i) at (0, \i) {} ;
				\node[server, label=v\z] (r\j) at (2, 1 + \j) {} ;
			}
			\path[dotted pattern=1.0pt] (l1) -- (l5);
			\path[dotted pattern=1.0pt] (r2) -- (r6);
			
			\foreach \i in {1,5,9} {
				\foreach \j in {2,6,10} {
					\ifthenelse {\i = \j}	{ } {
						\draw[ultra thin, gray] (l\i) -- (r\j);
					};
				}
			}
			\foreach \i/\j in {1/2, 5/6,  9/10} {
				\draw[ultra thick,black] (l\i) -- (r\j);			
			}
			\foreach \i/\j in {2/1, 6/5, 10/9} {
				\draw[white] (l\i) -- (r\j);
			}
			\ifthenelse{\matching=1}
			{ %then
				\foreach \i/\j in {5/10, 9/6} {
					\draw[ultra thick, dotted, blue] (l\i) to (r\j); 
				}				
			} %end then
			{
				\foreach \i/\j in {5/10, 9/6} {
					\draw[black] (l\i) to (r\j); 
				}
			}
		}
		
		\newcommand\DrawDoubleBipartiteRight{
			\foreach \i/\j/\z in {1/2/1, 5/6/2, 9/10/3} {
				\node[client, label=u'\z] (ll\i) at (6, \i) {} ;
				\node[server,label=v'\z] (rr\j) at (8, 1 + \j) {} ;
			}
			\path[dotted pattern=1.0pt] (ll1) -- (ll5);
			\path[dotted pattern=1.0pt] (rr2) -- (rr6);
			
			\foreach \i in {1,5,9} {
				\foreach \j in {2,6,10} {
					\ifthenelse {\i = \j}	{ } {
						\draw[ultra thin, gray] (ll\i) -- (rr\j);
					};
				}
			}
			\foreach \i/\j in {1/2,  5/6,  9/10} {
				\draw[ultra thick,black] (ll\i) -- (rr\j);			}
			\ifthenelse{\matching=1}
			{ %then
				\foreach \i/\j in {5/10, 9/6} {
					\draw[ultra thick,dotted, blue] (ll\i) to (rr\j); 
				}				
			} %end then
			{
				\foreach \i/\j in {5/10, 9/6} {
					\draw[black] (ll\i) to (rr\j); 
				}
			}
		}

 	        \tikzstyle{small node}=[draw, circle, minimum size=1ex, inner sep=0, fill=black] 
			\newcommand{\darkercolor}[3]{% Reference Color, Percentage, New Color Name
				 \colorlet{#3}{#1!#2!black}
			}
			\darkercolor{red}{50}{DarkRed}
			\darkercolor{blue}{50}{DarkBlue}

			\newcommand{\lightercolor}[3]{% Reference Color, Percentage, New Color Name
			    \colorlet{#3}{#1!#2!white}
			}
			\lightercolor{red}{50}{LightRed}
			\lightercolor{blue}{50}{LightBlue}
					
		\newcommand\DrawNodes[4]{ %{x, y, label, style}
             \foreach \i in {1,...,3} {
             	\node[#4] (#3_\i) at (#1,\i + #2 ) {} ;
             }
             \foreach \i in {1,...,3} {
               	\node[#4] (#3_\i) at (#1,\i + #2 ) {} ;
              }
		}			
		\newcommand\DrawComplete[3]{ %{label1, label2, color}
			\foreach \i in {1,...,3} {
				\foreach \j in {1,...,3} {
					\draw[#3] (#1_\i) to (#2_\j) {} ;
				}
			}
		}	
		
		\newcommand\DrawRedBlueGroup[3] { %{hight, group, style}
			\foreach \side/\x/\y in {l/0/0,r/2/0, ll/5/0, rr/7/0}	{
				\DrawNodes{\x }{\y+#1}{#2_\side}{small node}
			}
			\foreach \side/\x/\y in {minus/-1/0,plus/8/0}	{
				\DrawNodes{\x }{\y+#1}{#2_\side}{}
			}
			\DrawComplete{#2_l}{#2_r}{dashed, red, , ultra thin, #3}
			\DrawComplete{#2_ll}{#2_rr}{dashed, red,ultra thin , #3}
			\DrawComplete{#2_r}{#2_ll}{blue, #3}
			\DrawComplete{#2_minus}{#2_l}{blue, #3}
			\DrawComplete{#2_plus}{#2_rr}{blue, #3}
		}
		
		\newcommand\SendEdges[3] { %{origin, destination group, color}
		  \foreach \i in {1,...,3} {
		    \draw[ultra thick, #3] (#1) -- (#2_\i);	
		  }		
		}
					
		\tikzstyle{coiled}=[decoration={coil, aspect=0,,segment length=15pt}]
		\newcommand\DrawDoubleBipartiteMain {
	 		\DrawRedBlueGroup{-3}{g_1}{ultra thin}
	 		\DrawRedBlueGroup{0}{g_2}{ultra thin}
	 		\DrawRedBlueGroup{4.5}{g_3}{ultra thin}	
	 		\DrawRedBlueGroup{7.5}{g_4}{ultra thin}	
	 		
             \node[above] at (0,11) {A};
             \node[above] at (2,11) {B};
             \node[above] at (5,11) {C};
             \node[above] at (7,11) {D};
           
           \path [dotted pattern] (3.5,3) -- (3.5,5.5);
           \path [dotted pattern] (-1,6.5) -- (-1.5,6.5);
           \path [dotted pattern] (8,6.5) -- (8.5,6.5);
   			
   			\foreach \i in {0,2,5,7} {
	   			\node[] (e0_\i) at (\i,3.2) {};
	   			\node[] (e1_\i) at (\i,3.8) {};
	   			\node[] (e2_\i) at (\i,4.4) {};
	   			\node[] (e3_\i) at (\i,5) {};
	   		}
   			\foreach \j in {0,...,3} {
	   			\foreach \l in {0,...,3} {
			   		\ifthenelse {\j=\l} {} {
						\draw[ultra thin, gray] (e\j_0) -- (e\l_2);
						\draw[ultra thin, gray] (e\j_5) -- (e\l_7);
					}
				}
   			}

			\draw[ultra thick, green,decorate, coiled] (g_1_r_2) node[above,black,xshift=4pt,yshift=2pt] {$b^*$} -- (g_2_rr_3) node[above,black,xshift=5pt] {$d^*$};		
			\draw[ultra thick, green, decorate, coiled]  (g_3_l_2)  node[above,black,yshift=1pt] {$a^*$} -- (g_4_ll_2) node[above,black,xshift=2pt] {$c^*$};

			\draw[ultra thick, violet,dashed] (g_1_r_2) -- (g_3_l_2);		
            \draw[ultra thick, violet, dashed] (g_2_rr_3) -- (g_4_ll_2);		
			\SendEdges{g_2_rr_3}{g_1_ll}{black}
			\SendEdges{g_1_r_2}{g_1_ll}{blue}
            \SendEdges{g_3_l_2}{g_4_r}{black}
			\SendEdges{g_4_ll_2}{g_4_r}{blue}
            \SendEdges{g_3_l_2}{g_3_minus}{blue}
            \SendEdges{g_3_plus_2}{g_3_rr}{blue}
            \SendEdges{g_4_ll_2}{g_3_rr}{black}
            
	          \SendEdges{g_2_rr_3}{g_2_plus}{blue}
	          \SendEdges{g_2_minus_3}{g_2_l}{blue}
	          \SendEdges{g_1_r_2}{g_2_l}{black}
			\node[] at (0.5,-2.5) {}; %just here to create space
		}

	\newcommand{\DrawBaseFixSmallDelta}{		
	\node[small node, label=left:$a^*$] (a*) at (0,2.2) {};
		\node[small node, label=left:$b^*$] (b*) at (1,0) {};
		\node[small node, label={[label distance=0.01]175:$c^*$}] (c*) at (2,3.2) {};
		\node[small node, label=right:$d^*$] (d*) at (3,1) {};
		\node[small node, label=left:$a$] (a) at (0,5) {};
		\node[small node, label=right:$c$] (c) at (2,4.5) {};
		\draw[ultra thick, green, decorate, coiled] (b*) -- (d*);
}

		\newcommand{\DrawFixSmallDelta}{
			\DrawBaseFixSmallDelta
			\draw[ultra thick, green, decorate, coiled] (a*) -- (c*);		
			\draw[ultra thick, violet,dashed] (a*) -- (b*);
			\draw[ultra thick, violet, dashed] (c*) -- (d*);
			\draw[black] (a) to (b*);
			\draw[black] (c) to (d*);
			
			\draw [<->, ultra thick] (3.5,3)--(5.5,3) node[midway,above] {A switch};
			\begin{scope}[shift={(6.5,0)}]
			\DrawBaseFixSmallDelta
			\draw[ultra thick, green, decorate, coiled] (a) -- (c);		
			\draw[ultra thick, violet,dashed] (a) -- (b*);
			\draw[ultra thick, violet, dashed] (c) -- (d*);
			\draw[black] (a*) to (b*);
			\draw[black] (c*) to (d*);
			\end{scope}
		}

	\newcommand\DrawDoubleBipartiteRedBlue{
		\begin{tikzpicture}

		\foreach \i/\j/\z in {1/2/1, 5/6/2, 9/10/3} {
			\node[client, label=u\z] (l\i) at (0, \i) {} ;
			\node[server, label=v\z] (r\j) at (2, 1 + \j) {} ;
		}

		\path[dotted pattern=1.0pt] (l1) -- (l5);
		\path[dotted pattern=1.0pt] (r2) -- (r6);
		
		\foreach \i in {1,5,9} {
			\foreach \j in {2,6,10} {
				\ifthenelse {\i = \j}	{ } {
					\draw[ultra thin, gray] (l\i) -- (r\j);
				};
			}
		}
		\foreach \i/\j in {1/2, 5/6,  9/10} {
			\draw[ultra thick,black] (l\i) -- (r\j);			
		}
		\foreach \i/\j in {2/1, 6/5, 10/9} {
			\draw[white] (l\i) -- (r\j);
		}
		\ifthenelse{\matching=1}
		{ %then
			\foreach \i/\j in {5/10, 9/6} {
				\draw[ultra thick, dotted, blue] (l\i) to (r\j); 
			}				
		} %end then
		{
			\foreach \i/\j in {5/10, 9/6} {
				\draw[black] (l\i) to (r\j); 
			}
		}
		\end{tikzpicture}
	}

		\pgfmathdeclarefunction{gauss}{2}{%
		  \pgfmathparse{1/(#2*sqrt(2*pi))*exp(-((x-#1)^2)/(2*#2^2))}%
		}

	\tikzstyle{client}=[draw, circle, minimum size=2ex, inner sep=0,fill=black] 
	\tikzstyle{server}=[draw, circle, minimum size=2ex, inner sep=0,fill=black]

\newcommand\CompleteMatching{
	\foreach \i in {1,2} {
		\node[client] (l\i) at (0,\i) {} ;
		\node[server] (r\i) at (2,\i) {} ;
	}
	\foreach \i in {3,...,5} {
		\node[client] (l\i) at (0,\i+2) {} ;
		\node[server] (r\i) at (2,\i+2) {} ;
	}
	
	\path[dotted pattern=1.0pt] (l2) -- (l3);
	\path[dotted pattern=1.0pt] (r2) -- (r3);
	
	\foreach \i in {1,...,5} {
		\foreach \j in {1,2,3,4,5} {
			\ifthenelse {\i = \j}	{ } {
				\draw[ultra thin] (l\i) -- (r\j);
			};
		}
	}
	\foreach \i in {1,...,5} {
		\draw[ultra thin, gray] (l\i) -- (r\i);
	}	

	\node[] (l6) at (0,4) {};
	\node[] (r6) at (2,3.8) {};
	
	\foreach \i/\j in {l1/l6, l2/l4, l3/l5} {
		\draw[ultra thick,blue] (\i) to[bend left=40] (\j); 
	}
	\foreach \i/\j in {r1/r2, r6/r5, r3/r4} {
		\draw[ultra thick,blue] (\i) to[bend right=40] (\j); 
	}
	
	\foreach \i/\j in {l1/r2, l2/r1, l3/r4, l4/r5, l5/r3} {
		\draw[ultra thick, white] (\i) to (\j); 
		\draw[ultra thick,dashed, red] (\i) to (\j); 
	}
	\draw [thick, decoration={ brace, raise=0.5cm,amplitude=10pt}, decorate] (-0.3,0.8) -- (-0.3,7.2) node [midway, left=0.5cm+10pt] {$\sqrt m$} ;
			
	\begin{scope}[shift={(5.5,3)},xscale=0.6,yscale=1]
		\DrawIndependentSet
		\begin{scope}[shift={(0.8,-0.5)}]
			\DrawIndependentSet
		\end{scope}
	\end{scope}
	\draw [thick, decoration={ brace, raise=0.5cm,amplitude=10pt}, decorate] (7.1,6.5) -- (7.1,1) node [midway, right=0.5cm+10pt] {$n-2\sqrt m$} ;
}

\newcommand\DrawKleqSqrtM{
	\foreach \i in {1,2} {
		\node[client] (l\i) at (0,\i) {} ;
		\node[server] (r\i) at (2,\i) {} ;
	}
	\foreach \i in {3,...,6} {
		\node[client] (l\i) at (0,\i+1) {} ;
		\node[server] (r\i) at (2,\i+1) {} ;
	}
	
	\path[dotted pattern=1.0pt] (l2) -- (l3);
	\path[dotted pattern=1.0pt] (r2) -- (r3);
	
	\foreach \i in {1,...,6} {
		\foreach \j in {1,...,6} {
			\ifthenelse {\i = \j}	{ } {
				\draw[ultra thin] (l\i) -- (r\j);
			};
		}
	}
	\draw[gray, ultra thin] (l5) -- (r5);
	\draw[gray, ultra thin] (l6) -- (r6);
	\foreach \i/\j in {l1/l2, l3/l4} {
		\draw[ultra thick,blue] (\i) to[bend left=40] (\j); 
	}
	\foreach \i/\j in {r1/r2, r3/r4} {
		\draw[ultra thick,blue] (\i) to[bend right=40] (\j); 
	}
	\foreach \i/\j in {l1/r1, l2/r2, l3/r3, l4/r4} {
		\draw[ultra thick,red, dashed] (\i) to (\j); 
	}
	\draw [thick, decoration={ brace, raise=0.5cm,amplitude=10pt}, decorate] (0,1) -- (0,7) node [midway, left=0.5cm+10pt] {$\sqrt m$} ;
	\draw [thick, decoration={ brace, raise=0.5cm,amplitude=4pt}, decorate] (2,5.2) -- (2,3.8) node [midway, right=0.5cm+4pt, align=left] {The $i\th$ \\ matched square};
}

\newcommand{\DrawIndependentSet}{
	\node [ cycle] at (-1,-1) {};
	\node [ cycle] at (1,-1) {};
	\node [ cycle] at (2,1) {};
	\node [ cycle] at (-2,1) {};
	\node [ cycle] at (1,3) {};
	\node [ cycle] at (-1,3) {};
	\node [ cycle] at (0,1) {};	
}

%% file: Introduction.tex
\section{Introduction}

Counting the number of triangles in a graph is a fundamental algorithmic problem.
In the study of complex networks
and massive real-world graphs, triangle counting is a key operation in graph
analysis for bioinformatics, social networks, community analysis, and graph modeling~\cite{HoLe70,Co88,portes2000social,EcMo02,milo2002network,Burt04,becchetti2008efficient,foucault2010friend,BerryHLP11,SeKoPi11}.
% There has been much study on
% The problem of counting the number of triangles
% There has been quite extensive study of triangle counting by the theoretical computer science community, where
In the theoretical computer science community,
the primary tool for counting the number of triangles is
fast matrix multiplication~\cite{itai1978finding,alon1997finding,BjPa+14}. On the more applied side, there is a plethora
of provable and practical algorithms that employ clever sampling methods for
approximate triangle counting~\cite{ChNi85,ScWa05,ScWa05-2,tsourakakis2008fast,tsourakakis2009doulion,avron2010counting,kolountzakis2012efficient,chu2011triangle,SuVa11,tsourakakis2011triangle,arifuzzaman2013patric,SePiKo13,TaPaTi13}. Triangle counting has also been a popular problem
in the streaming setting~\cite{bar2002reductions,jowhari2005new,buriol2006counting,AhGuMc12,KaMeSaSu12,JhSePi13,PaTaTi+13,TaPaTi13,AhDuNe+14}.

%Yet
All these algorithms  read the entire graph, which may be time consuming when the graph is very large.
In this work, we focus on \emph{sublinear} algorithms for triangle counting.
We assume the following query access to the graph, which is standard
for sublinear algorithms that approximate graph parameters.
The algorithm can make: (1) Degree queries, in which the algorithm can query the degree $d_v$ of any vertex $v$. (2) Neighbor queries, in which the algorithm can query what vertex is the $i\th$ neighbor of a vertex $v$, for any $i \leq d_v$.
(3) Vertex-pair queries, in which
the algorithm can query for any pair of vertices $v$ and $u$ whether $(u,v)$ is an edge.
% consider a different model, in which there is query access to the graph, and the goal is to obtain an approximation of the number of triangles by performing a number of queries that is sublinear in the size of the graph.
% $n+m$, where $n$ is the number of vertices and $m$ is the number of edges.
% This is the standard model for sublinear algorithms that approximate graph parameters.

Gonen et al.~\cite{GRS11}, who studied the problem of approximating the number of stars in a graph in sublinear time,
also considered the problem of approximating the number of triangles in sublinear time.
% were the first to consider counting triangles in sublinear time.
% The main focus in \cite{GRS11} is on approximating the number of stars in a graph, and the problem of approximately
% counting the number of triangles is addressed as a natural extension of length-$2$ paths (a special case of stars).
They proved that there is no sublinear approximation algorithm for the number of triangles
when the algorithm is allowed to perform degree and neighbor queries (but not pair queries).
\footnote{To be precise,
they showed that there exist two families of graphs over $m= \Theta(n)$ edges, such that all graphs in
one family have $\Theta(n)$ triangles, all graphs in the other family have \textsf{no} triangles, but
in order to distinguish between a random graph in the first family and random graph in the second family,
it is necessary to perform $\Omega(n)$ degree and neighbor queries.}

They asked whether a sublinear algorithm exists when allowed vertex-pair queries in addition to degree and neighbor queries. % will enable for
% can lead %to better results in terms of query complexity and running time.
We show that this is indeed the case.
% Moreover, our results apply to any possible set of values $n$, $m$ and $\Delta(G)$, where $\Delta(G)$ denotes the number of triangles in a graph $G$. \par

	 \subsection{Results} \label{sect results}
Let $G$ be a graph with $n$ vertices, $m$ edges, and $t$ triangles.
	We describe an algorithm that, given an approximation parameter $0<\eps<1$ and query access to  $G$, outputs an estimate $\wht$, such that with high constant probability (over the randomness of the algorithm), $(1-\eps)\cdot\tr\leq \wht\leq (1+\eps)\cdot\tr$. The expected query complexity of the algorithm is

	\[\left(\frac{n}{\tr^{1/3}}+\min\left\{m,\frac{m^{3/2}}{\tr}\right\}\right)\cdot \polylog \;,\]
and its expected running time is
$\!\left(\frac{n}{\tr^{1/3}}+\frac{m^{3/2}}{\tr}\right)\cdot \polylog $.
We show that this result is almost optimal by % presenting lower bounds proving
proving that the number of queries performed by any multiplicative-approximation algorithm for the number of triangles in a graph is
	\[\Omega\left(\frac{n}{\tr^{1/3}}+\min\left\{m,\frac{m^{3/2}}{\tr}\right\}\right) \;.\]

\subsection{Overview of the algorithm} \label{sec:idea}

For the sake of clarity, we suppress any dependencies on the approximation parameter $\eps$
and on $\log n$ using the notation $O^*(\cdot)$.
\subsubsection{A simple oracle-based procedure for a $1/3$-estimate}
First, let us assume access to an oracle that,
% Consider first an imaginary scenario, in which there is an oracle that, i
given a vertex $v$, returns $t_v$,
the number of triangles that $v$ is incident to. Note that $t = \sum_v t_v/3$.
An unbiased estimate is obtained by sampling, uniformly at random, a (multi-)set $S$ of $s$  vertices,
and outputting
% Suppose we select $s$ vertices uniformly,
% independently and at random, and let the multiset of selected vertices be denoted by $S$.
% If we set
 $t_S =  \frac{1}{3}\cdot \frac{n}{s}\cdot \sum_{v\in S} t_v$.
%   then the expected value
% of $t_S$ equals $t$.
Yet this estimate can have extremely large variance (consider the ``wheel'' graph where there
is one vertex with $t_v = \Theta(n)$ and all other $t_v$'s are constant).
% However, since $t_v$ may vary significantly between different vertices, we might need to take a very large sample in order
% to ensure that the estimator $t_S$ is close to this expected value.
Inspired by work on estimating the average degree~\cite{feige2006sums,GR08},
we can reduce the variance by simply ``cutting off'' the contribution of
vertices $v$ for which $t_v$ is above a certain threshold. Call such vertices {\em heavy}, and denote the remaining {\em light}.
% , where we refer to such vertices as {\em heavy\/} vertices (and
% to all other vertices as {\em light\/}).
If the threshold is set to $\Theta(\tr^{2/3}/\eps^{1/3})$, then the number of heavy vertices is $O((\eps t)^{1/3})$.
This implies that the total number of triangles in which all three endpoints are heavy
is $O(\eps t)$.

% \mnote{D: connect text in this par better}
Hence, suppose we define $\widetilde{\tr}_v$ to
be $t_v$ if $t_v \leq \tr^{2/3}/\eps^{1/3}$ and $0$ otherwise, and consider
$\widetilde{\tr}_S = \frac{1}{3}\cdot \frac{n}{s}\cdot \sum_{v\in S} \widetilde{\tr}_v$.
We can argue that $\EX[\widetilde{\tr}_S] \in [(1/3 - \eps)t, t]$, since (roughly speaking) every
triangle that contains at least one light vertex is counted at least once.
% While we can no longer claim that
% the expected value of $\widetilde{\tr}_S$ is $t$, its expected value can be lower bounded
%  by $\left(\frac{1}{3}-\eps\right)t$ (and upper bounded by $t$).
% he source of the factor of $\frac{1}{3}$
%  is due to the fact that while triangles in which all three endpoints are light are counted three times (once
%  from each endpoint), triangles that have  one light endpoint are counted  only once.
 Since $\widetilde{\tr}_v$ ranges
 between $0$ and $\tr^{2/3}/\eps^{1/3}$, by applying the multiplicative Chernoff bound,
 a sample of size $s = O^*\!\left(\frac{n}{\tr^{1/3}}\right)$ % \cdot\poly(\log n,1/\eps)$
 is sufficient to ensure that with high constant probability $\widetilde{\tr}_S$ is
 in the range $\left[\left(\frac{1}{3}-2\eps\right)\cdot\tr,(1+\eps)\cdot\tr\right]$.

\subsubsection{Assigning weights to triangles so as to improve the estimate}
% Since we are interested in obtaining an estimate that is in the range $[(1-O(\eps))\cdot\tr,(1+O(\eps))\cdot\tr]$,
To improve the approximation, we assign weights to triangles inversely proportional to the number of their light endpoints
(rather than assigning a uniform weight of $\frac{1}{3}$
as is done when defining $\widetilde{\tr}_S = \frac{n}{s}\cdot \sum_{v\in S} \frac{1}{3}\cdot \widetilde{\tr}_v$).
If for each light vertex $v$ we let $\wt(v)$
be the sum over the weights of all triangles that $v$ participates in
% \mnote{D: this differs by a factor of 2 from def later, so maybe note this there} Done
and for each heavy vertex $v$ we let $\wt(v) = \widetilde{\tr}_v = 0$, then the expected value
of $\frac{n}{s}\cdot \sum_{v\in S} \wt(v)$ is in $[(1-O(\eps))\cdot\tr,(1+O(\eps))\cdot\tr]$.
% \mnote{D: refer to compensation idea in previous work? maybe when talking about previous work}

\medskip
% When turning from the above imaginary (oracle based) scenario(s) to an actual implementation, we need to address two issues.
To get rid of the fictitious oracle, we must resolve two issues.
The first issue is efficiently deciding whether a vertex is heavy or light, and the second is
approximating $\frac{n}{s}\cdot \sum_{v\in S} \wt(v)$,
 assuming we have a procedure
for deciding whether a vertex is heavy or light.
    We next discuss each of these two issues.
For convenience, we will assume that the algorithm already has constant factor estimates for $m$ and $t$.
This can be removed by approximating $m$ and performing a geometric search on $t$.

% In this discussion it seems that the algorithm works under the assumption that it is provided with
% (constant factor estimates of)  the number of edges $m$
% and the number of triangles $t$ so as to set various thresholds and sample sizes.
% As a final step in the (full) description and analysis of the algorithm we show how to remove
% this assumption.

\subsubsection{Deciding whether a vertex is heavy}
\label{intro-heavy:subsubsec}
Let $v$ be a fixed vertex with degree  $d_v$.
Consider an edge $e$ incident to $v$, and let $u$ be the other endpoint of this edge.
Let
$\tr_e$ denote the number of triangles that $e$ belongs to.
Consider the  random variable $Y$ defined by selecting, uniformly at random, a neighbor $w$ of $u$, and setting $Y=d_u$
% (the degree of $u$)
 if  $(v,w)$ is an edge (so that $(v,u,w)$ is a triangle) and $Y=0$ otherwise. Since the number of neighbors of $u$
that form a triangle with $v$ is $\tr_e$, the expected value of $Y$ is $\frac{\tr_e}{d_u}\cdot d_u = \tr_e$.
Now consider selecting (uniformly at random) several edges incident to $v$, denoted $e_1,\dots,e_r$, and for each edge $e_j$ selected, defining the corresponding random variable $Y_j$. Then the expected value of $\frac{1}{r}\sum_{j=1}^r Y_j$
is $\frac{1}{d_v}\cdot \sum_{e=(v,u)}t_{e} = \frac{2}{d_v}\cdot t_v$. If we multiply by $d_v/2$, then
we get an unbiased estimator for $t_v$, which in particular can indicate  whether $v$ is
heavy or light.

However, once again the difficulty is with the variance of this estimator and the implication on the complexity
of the resulting decision procedure.
%\dddnote{}
To %reduce the variance
\dddchange{address these difficulties} we modify the procedure described above as follows.
First, if $d_v$ is above a certain threshold, then $v$ is also considered
heavy (where this threshold is of order $O\left(m/(\eps t)^{1/3}\right)$, so that the total number of
triangles in which all three endpoints are heavy remains $O(\eps \tr)$.
Second, observe that when trying to estimate the number of triangles that an edge $e_j= (v,x_j)$ participates in, we can either select a random neighbor $w$ of $v$ and check whether $(x_j,w) \in E$, or we can select a random
neighbor $w$ of $x_j$ and check whether $(v,w) \in E$. Since it is advantageous \dddchange{for the sake of
the complexity} to consider the endpoint that has
a smaller degree, we do the following.
Each time we select an edge $e_j = (v,x_j)$ incident to $v$, we let $u_j$ be the endpoint of $e_j$ that has smaller degree.
% If we let $z_j = \{v,y_j\}\setminus u_j$ denote the other endpoint of $e_j$, then $t_{e_j}$  equals the number of neighbors $w$ of $u_j$ (among its $d_{u_j}$ neighbors) such that $(z_j,w) \in E$.
If $d_{u_j}$ is relatively large (larger than $\sqrt{m}$), then we select $k = \lceil d_{u_j}/\sqrt{m}\rceil$
neighbors of $u_j$ and let $Y_j$ equal $d_{u_j}$ times the fraction among these neighbors that close a triangle with $e_j$.
\dddchange{The setting of $k$ implies a bound on the variance of $Y_j$ (conditioned on the choice of $e_j$), which is
$\sqrt{m}$ times its expected value, $\tr_{e_j}$.
Third, in order to bound the variance due to the random choice of edges $e_j$ incident to $v$,
 we do the following. We assign each triangle that $v$ participates in to a unique edge incident to $v$ and modify the definition of $\tr_e$ to be the number of such triangles that are assigned to $e$.
The assignment is such that $\tr_e$ is always upper bounded by $O(\sqrt{m})$.
}
Finally, we perform a standard median selection
over $O(\log n)$ repetitions of the procedure.

Our analysis shows that it suffices to set $r$ (the number of random edges incident to $v$ that are selected) to be
$O^*\!\left(\frac{m^{3/2}}{\tr}\right)$ so as to ensure the correctness of the procedure
(with high probability). In the analysis of the expected query complexity and running time of the
procedure we have to take into account the number of iterations $k = \lceil d_{u_j}/\sqrt{m}\rceil$ for each
selected (lower degree endpoint) $u_j$ \dddchange{and argue that for every vertex $v$, the expected number
of these iterations is a constant.}
% This is done
% similarly to the complexity analysis of the exact triangle counter of Chiba and Nishizeki~\cite{ChNi85}.

\subsubsection{Estimating $\sum_{v\in S} \wt(v)$}
\label{intro-alg:subsubsec}
Suppose we have a (multi-)set $S$ of vertices such that $\frac{n}{s}\cdot \sum_{v\in S} \wt(v)$ is indeed
in $[(1-O(\eps))\cdot\tr,(1+O(\eps))\cdot\tr]$ (which we know occurs with high probability if we select
 $s = O^*\!\left(\frac{n}{\tr^{1/3}}\right)$ % \cdot\poly(\log n,1/\eps)$
 vertices uniformly at random).
 Consider the set of edges incident to vertices in $S$, where we view edges as directed, so that if there is an
 edge between $v$ and $v'$ that both belong to $S$, then $(v,v')$ and $(v',v)$ are considered as two different edges.
 We denote this set of edges by $E_S$, and their number by $d_S$, where
 $d_S = \sum_{v\in S} d_v$.
Suppose that for each edge $e=(v,x)$ we assign a weight $\wt(e)$, which is the sum of the weights of all triangles
 that $v$ participates in \dddchange{ and are assigned to $e$} %it participates in
  (where the weight of a triangle is as defined previously based on the number of
 light endpoints that it has). Then $\sum_{e\in E_S} \wt(e) = \sum_{v\in S} \wt(v)$.

 The next  idea is to sample {\em edges\/} in $E_S$ uniformly at random, and for each selected edge $e= (v,u)$
 to estimate $\wt(e)$. An important observation is that since we can query the degrees of all vertices in $S$,
 we can efficiently select uniform random edges in $E_S$ (as opposed to the more difficult task of
 selecting random edges from the entire graph). Similarly to what was described in the decision procedure
 for heavy vertices, given an edge $e \in E_S$ we let $u$ be its endpoint that has smaller degree. We then
 select $\lceil\sqrt{m}/d_u\rceil$ random neighbors of $u$ and for each check whether it closes a triangle with $e$. For each triangle found \dddchange{that is assigned to $e$}, we
 check how many heavy endpoints it has (using the aforementioned procedure for detecting heavy vertices) so as to compute the weight of the triangle. In this manner we can obtain
 random variables whose expected value is $\frac{1}{d_S} \sum_{v\in S} \wt(v)$, and whose variance
 is not too large (upper bounded by $\sqrt{m}$ times this expected value). We can now take an
 average over sufficiently many ($O^*\!\left(\frac{m^{3/2}}{\tr}\right)$) %\cdot \poly(\log n,1/\eps)$)
 such random variables
 and multiply by $d_S \cdot n$. By upper bounding the probability that $d_S$ is much larger than its expected
 value we can prove that the output of the algorithm is as desired. The expected query complexity and running time of the
 algorithm are shown to be $O^*\!\left(\frac{n}{\tr^{1/3}} + \frac{m^{3/2}}{\tr}\right)$.

 Finally we note that if $\tr < m^{1/2}$ so that $\frac{m^{3/2}}{\tr} > m$, then we can replace
 $\frac{m^{3/2}}{\tr}$ with $m$ in the upper bound on the query complexity
% In particular we can obtain the connected components that contains vertices of $S$, which
% enables us to compute
since we can store all queried edges so that no edge needs to be queried more than twice (once from each
endpoint).

\subsection{A high level discussion of the lower bound}

Proving that every multiplicative-approximation algorithm must perform $\Omega\!\left(\frac{n}{\tr^{1/3}}\right)$
queries is fairly straightforward, and
	our main focus is on proving that
	$\Omega\!\left(\min\left\{m,\frac{m^{3/2}}{\tr}\right\}\right)$ queries are necessary as well.
 In order to
 prove this claim we define, for every $n$, every $1\leq m\leq \binom{n}{2}$ and every $1 \leq \tr \leq \min\{\binom{n}{3},m^{3/2}\}$,
% two families of graphs $\mG_1$ and
a graph $G_1$ and a family of graphs $\mG_2$ for which the following holds:
% All the graphs in both
~(1)~The graph $G_1$ and all the graphs in $\mG_2$ have $n$ vertices and $m$ edges.
~(2)~In $G_1$ there are no triangles, while the number of triangles in each graph $G \in \mG_2$ is $\Theta(\tr)$. We prove that for values of $t$ such that $t \geq \sqrt m$, at least $\Omega\!\left(\frac{m^{3/2}}{\tr}\right)$ queries are required in order to distinguish with high constant probability between $G_1$ and a random graph in $\mG_2$. We then prove that for values of $t$ such that $t < \sqrt m$, at least $\Omega(m)$ queries are required
for this task. We give three different constructions for $G_1$ and $\mG_2$ depending on the value of $t$ as a function of $m$
(where two of the constructions are for subcases of the case that $t \geq \sqrt m$).
For further discussion of the lower bound, see Section~\ref{sec:lb}.
% \tchange{The high level idea for the constructions and the intuition as to why they give an $\Omega(m^{3/2}/t)$ lower bound are described in the beginning of Section~\ref{sec:lb}.}
\ifnum\focs=1
Due to space constraints, in this extended abstract we only described one of these constructions and provide part of the details for the corresponding lower bound proof. All missing details can be found in the full version of this paper~\cite{ELRS}. 
\fi

\subsection{Related Work} \label{related work}
\subsubsection[Approximating graph parameters in sublinear time]{Approximating graph parameters in sublinear time}
We build on previous work on approximating the average degree of a graph and the number of stars~\cite{feige2006sums,GR08,GRS11}.
% extends that of \cite{feige2006sums,GR08} on approximating the average degree of a graph (the number of edges) and the work of \cite{GRS11} on approximating the number of stars in a graph, in sublinear time.
Feige \cite{feige2006sums} investigated the problem of estimating the average degree of a graph, denoted $\overline{d}$, when given query access to the degrees of the vertices.
By performing a careful variance analysis, Feige proved that $O\!\left(\sqrt {n/\overline{d}}/\eps\right)$ queries are sufficient  in order to obtain a $(\frac{1}{2}-\eps)$-approximation of $\overline{d}$.
He also proved that a better approximation ratio cannot be achieved in sublinear time using only degree queries. The same problem was considered by Goldreich and Ron~\cite{GR08}. Goldreich and Ron proved that a $(1+\eps)$-approximation can
be achieved with
$O\!\left(\sqrt{n/\sqrt{\overline{d}}}\right)\cdot \poly(\log n,1/\eps)$ queries, if neighbor queries are also allowed.

%  are sufficient  in order to obtain
% a $(1\pm\eps)$-approximation of $\overline{d}$.
% In both results the term $\sqrt{n}$ can actually be replaced by $\sqrt{n/\overline{d}}$,
% so that the complexity of the algorithms improves as $\overline{d}$ increases.
% there exists an algorithm that outputs a $(1\pm\eps)$-approximation of the average degree $\overline{d}$ in time $O\!\left(\sqrt{n/\overline{d}(G)}\right)\cdot \polylog$.
% \mnote{D: maybe add a bit on similarities and differences}

% Gonen et al.~\cite{GRS11} considered the problem of approximating the number of $s$-stars in a graph. That is, subgraphs over $s+1$ vertices, where one vertex is connected to all others. They presented an algorithm that, given an approximation parameter $0<\eps<1$ and query access to a graph $G$, outputs an estimate $\hat{\nu}_s$ such that with high constant probability $(1-\eps)\nu_s(G) \leq \hat\nu_s \leq (1+\eps)\nu_s(G)$, where $\nu_s(G)$ denotes the number of $s$-stars in the graph. The expected query complexity and running time of their algorithm are $O\!\left(\frac{n}{\nu_s(G)^{1/(s+1)}}+ \min\left\{n^{1-1/s}, \frac{n^{s-1/s}}{\nu_s(G)^{1-1/s}}\right\}\right)\cdot \polylog$.
% We note that their algorithm uses only neighbor and degree queries, and observe that a major difference
% between counting $2$-stars and counting triangles, is that the former are non-induced subgraphs, while the latter are.
% \mnote{D: not clear that this is clear..}
%

Building on these ideas, Gonen et al.~\cite{GRS11} considered the problem of approximating the number of $s$-stars in a graph.
Their algorithm only used neighbor and degree queries. A major difference between stars and triangles is that the former are non-induced subgraphs, while the latter are.
Additional work on sublinear algorithms for estimating other graph parameters include those for approximating the size of the minimum weight spanning tree \cite{DBLP:journals/siamcomp/ChazelleRT05, DBLP:journals/siamcomp/CzumajS09, DBLP:journals/siamcomp/CzumajEFMNRS05}, maximum matching \cite{nguyen2008constant, yoshida2009improved} and of the minimum vertex cover \cite{DBLP:journals/tcs/ParnasR07,nguyen2008constant,DBLP:journals/talg/MarkoR09, yoshida2009improved, hassidim2009local, onak2012near}.

\subsubsection{Triangle counting} \label{subsec:triangles}

Triangle counting has a rich history.
A classic result of Itai and Rodeh showed that triangles can be enumerated
in $O(m^{3/2})$ time, and a more elegant algorithm was given by Chiba and Nishizeki~\cite{ChNi85}.
The connections to matrix multiplication have been exploited
for faster theoretical algorithms~\cite{itai1978finding,alon1997finding,BjPa+14}.
In practice, there is a diverse body on work on counting triangles using different techniques, for different models.
There are serial algorithms based on eigenvalue methods~\cite{tsourakakis2008fast,avron2010counting},
graph sparsification~\cite{tsourakakis2009spectral,kolountzakis2012efficient,tsourakakis2011triangle,PaTs12}, and
sampling paths~\cite{ScWa05,SePiKo13}. Triangle counters have been given
for MapReduce~\cite{Co09,SuVa11,KoPiPlSe13}; external memory models~\cite{chu2011triangle}; distributed settings~\cite{arifuzzaman2013patric};
semi-streaming models~\cite{becchetti2008efficient,kolountzakis2012efficient}; one-pass streaming~\cite{bar2002reductions,jowhari2005new,buriol2006counting,AhGuMc12,KaMeSaSu12,JhSePi13,PaTaTi+13,TaPaTi13,AhDuNe+14}.
It is worth noting that across the board, all these algorithms required reading the entire graph.

Most relevant to our work are various sampling algorithms, that set up a random
variable whose expectation is directly related to the triangle count~\cite{ScWa05,kolountzakis2012efficient,
jowhari2005new,buriol2006counting,SePiKo13,JhSePi13,PaTaTi+13,TaPaTi13,AhDuNe+14}. Typically, this involves
sampling some set of vertices or edges to get a set of three vertices. The algorithm checks whether the sampled set induces a triangle,
and uses the probability of success to estimate the triangle count.
We follow the basic same philosophy. But it is significantly more challenging to
set up the ``right" random experiment, since we cannot read the entire graph.

% 	 \subsubsection{Counting the number of triangles in the streaming model}\label{subsec:stream}
%   Bar-Yossef et al.~\cite{bar2002reductions} initiated the study of counting the number of triangles in the streaming model. Many works have been conducted since, e.g. \cite{jowhari2005new, buriol2006counting, becchetti2008efficient, tsourakakis2009doulion, tsourakakis2009spectral, tsourakakis2011triangle, yoon2011improved, kolountzakis2012efficient}, differing in the number of passes they perform, the assumptions they make on the structure of the graph, the requirements on the output and more. A work with some resemblance to ours is the work
%    of Kolountzakis et al.~\cite{kolountzakis2012efficient}. They present a streaming algorithm that makes three passes over the edge stream
%    and outputs a $(1\pm \eps)$-approximation of $\tr = \tr(G)$. The space complexity of the algorithm is $O\!\left(\sqrt{m}\cdot \log m + \frac{m^{3/2} \cdot \log n }{\\tr\cdot \eps^2}\right)$
%    (where $m,n$ and $\tr$ are as defined previously). The point of similarity between their
%    algorithm and ours is that they also classify the graph's vertices into high-degree vertices, with degree strictly greater than $\sqrt m$, and low-degree vertices, with degree at most $\sqrt m$, and apply a different approach to estimate the triangles-degree of each type of vertices.
% Since their algorithm requires several passes over the edge stream and relies heavily on direct access to uniformly selected edges, it otherwise clearly differs from our algorithm.
% 

%% file: UpperBound.tex
\section{Preliminaries} \label{sec:prel}

%We use $d_v$ for the degree of vertex $v$, $E_v$ for the set of edges incident to $v$,
%and $\Gamma_v$ for the neighborhood of $v$.
%We use $t_e$ for the number of triangles incident to edge $e$, and set $t_v = \sum\limits_{e \in E_v} t_e$.
%Note that the latter is twice the number of triangles incident to $v$.
%The set of triangles incident to an edge $e$ is denoted by $\Tr_e$. The set of triangles in the graph $G$ is denoted $\Tr$.
%We use $c, c_1, \ldots$ to denote sufficiently large constants.
%We use $\eps$ to denote the approximation parameter.
%\begin{assumption}
%Our initial description of the algorithm will assume access to estimates of $m$ and $t$, denoted $\om$ and $\ot$ respectively, such that the following holds:
%\begin{enumerate}
%	\item $t \leq \ot \leq \ct\cdot t$, for some constant $\ct$.
%	\item $m \leq \om \leq \ct\cdot m$.
%\end{enumerate}
%\end{assumption}
%
%%and will use these values to decide how much to sample.
%These (somewhat circular) assumption is easily removed by performing a geometric search on $m$ and $t$, as explained at the end of our proof.

%-----------------------------------------------------------------\\
	Let $G=(V,E)$ be a simple graph with $|V|=n$ vertices and $|E|=m$ edges. For a vertex $v\in V$, we denote by $d_v$ the degree of the vertex, by $\Gamma_v$ the set of $v$'s neighbors, and by $E_v$ the set of edges incident to $v$.  \tchange{We denote by $T_v$ the set of triangles incident to the vertex $v$, and let $t_v = |T_v|$.
	Similarly, the set of triangles in the graph $G$ is denoted by $\Tr$, and the number of triangles in the graph in denote by $t$.} We use $c, c_1, \ldots$ to denote sufficiently large constants.
	
We consider algorithms that can sample uniformly in $V$ and perform three types of queries:
	\begin{enumerate}
		\item Degree queries, in which the algorithm may query for the degree $d_v$ of any vertex $v$ of its choice.	
		\item Neighbor queries, in which the algorithm may query for the $i\th$ neighbor of any vertex $v$ of its choice. If $i>d_v$, then a special symbol (e.g. $\dagger$) is returned. No assumption is made on the order of the neighbors of any vertex.
		\item Pair queries, in which the algorithm may ask if there is an edge $(u,v)\in E$ between any pair of vertices $u$ and $v$.
	\end{enumerate}
%	We denote by $\Tr(G)$ the set of triangles in the graph $G$, and by $\tr$ the number of triangles in the graph. Each triangle, that is three vertices $u,v,w\in V$ such that $(u,v),\; (v,w)$ and $(u,w)$ are edges in $G$, is denoted by an unordered triple $(v,u,w)$. \par
	We sometimes use set notations for operations on multisets. We use the notation $O^*(\cdot)$ to
	suppress dependencies on the approximation parameter $\eps$ or on $\log n$.
	
\ifnum\focs=0
We use the following variant of the multiplicative Chernoff bound.
Let $\chi_1, \ldots,\chi_r$ be $r$ independent random variables, such that $\chi_i \in [0,B]$ for some $B>0$ and
$\EX[\chi_i]=b$ for every $1\leq i \leq r$. For every $\gamma\in(0,1]$ the following holds:
\[\Pr\left[\frac{1}{r}\sum\limits_{i=1}^{r}\chi_i > (1+\gamma)b\right] < \exp\left(-\frac{\gamma^2 \cdot b\cdot r}{3B}\right),\numberthis \label{u.b. chernoff}\]
and
	\[\Pr\left[\frac{1}{r}\sum\limits_{i=1}^{r}\chi_i < (1-\gamma)b\right] < \exp\left(-\frac{\gamma^2\cdot b \cdot r}{2B}\right).\numberthis \label{l.b. chernoff}\]

	We will also make an extensive use of Chebyshev's inequality:
	For a random variable $X$ and for $\gamma > 0$,
	$$\Pr \left[ \left| X-\EX[X]  \right| \geq \gamma \right] \leq \frac{\var[X]}{\gamma^2}\;.$$
	
	\fi
We fix a total order on vertices denoted by $\prec$ as follows: $u \prec v$ if $d_u < d_v$ or $d_u = d_v$ and $u < v$ (in terms of id number).
Given $u$ and $v$, two degree queries suffice to decide their ordering.

\begin{claim}\label{clm:prec} Fix any vertex $v$. The number of neighbors $w$ of $v$ such that $v \prec w$ is at most $\sqrt{2m}$.
\end{claim}

\begin{proof} Let $S = \{w | w \in \Gamma_v, v \prec w\}$. Naturally, $d_v \geq |S|$. By definition of $\prec$,
$\forall w \in S$, $d_w \geq d_v \geq |S|$. Thus, $\sum_{w \in S} d_w \geq |S|^2$ and $|S| \leq \sqrt{2m}$.
\end{proof}

\section{The Algorithm}\label{sec:alg}
We start by introducing the notions of heavy and light vertices and how they can be utilized in the context of
estimating the number of triangles. We then give a procedure for deciding (approximately) whether a vertex is heavy or light. Using this procedure we  give an algorithm for estimating the number of triangles based on the following assumption
(which is later removed).

\begin{assumption} \label{assump:estimates}
	Our initial algorithm takes as input estimates $\overline{\tr}$ and $\overline{m}$ on the number of edges and triangles in the graph respectively, such that
	\begin{enumerate}
		\item $t/4 \leq \ot \leq t$. \label{assume t}
		\item   $m/6 \leq \om $. \label{assume m}

\end{enumerate}
\end{assumption}

\Assum{estimates} can be easily removed by performing a geometric search on $\tr$ and using the algorithm from 	 \cite{feige2006sums} to approximate $m$, as explained precisely in the proof of \Thm{maintri}.

\tchange{	For every vertex $v$, we view the set of edges $E_v$ as directed edges {\em originating\/} from $v$. We then  \emph{associate} each triangle $(v,x,w)\in T_v$ with a unique edge $e \in E_v$, as defined next.
%	one of its edges $\ora{(v,x)}$ or $\ora{(v,w)}$ according to the total order $\prec$
	\begin{definition} \label{associate}
	We say that a triangle $(v,x,w)\in T_v$ is \textsf{associated} with the directed edge $\ora{(v,x)}$ if $x \prec w$, and to $\ora{(v,w)}$ otherwise. For a directed edge $\ora{e}= \overrightarrow{(v,x)}$ we let $T_{\ora{e}}$ denote the set of triangles
	that %$v$ associates to it. 
it is associated with, that is, the set of triangles $(v,x,w)$ such that $x \prec w$.		
	\end{definition}
}
%Since we will always refer to directed edges, 
\dddchange{Since it will always be clear from the context from which vertex an edge we consider originates,}
\tchange{for the sake of succinctness, \dddchange{we drop the directed notation} and use the notation $T_e$. We let $t_e = |T_e|$, and for a fixed vertex $v$, we get $t_v=\sum\limits_{e \in E_v} t_e$.}

In all the follows we assume that $\eps<1/2$, and otherwise we run the algorithm with $\eps = 1/2$.
\subsection{Heavy and light vertices} \label{sec:heavy}

\begin{definition} \label{def:heavy}
	We say that a vertex $v$ is \textsf{heavy} if $d_v > \frac{2 \om}{(\eps \ot )^{1/3}}$	or if $\tr_v > \frac{2\ot^{2/3}}{\eps^{1/3}}$. If $v$ is such that $d_v \leq \frac{2 \om}{(\eps \ot )^{1/3}}$ and
$\tr_v \leq \frac{\ot^{2/3}}{2 \eps^{1/3}}$, then we say that $v$ is \textsf{light}.

We shall say that
a partition $(H,L)$ of $V$ is \textsf{appropriate} (with respect to $\om$ and $\ot$) if
every heavy vertex belongs to $H$ and
every light vertex belongs to $L$.
\end{definition}

Note that for an appropriate partition $(H,L)$ both $H$ and $L$ may contain vertices that are neither heavy nor light (but no light vertex belongs to $H$ and no heavy vertex belongs to $L$).

For a fixed partition $(H,L)$ we associate with
each triangle $\Delta$ a weight depending on the number of its endpoints that belong to $L$.
\begin{definition} For a triangle $\Delta$ we define its weight $\wt_{L}{(\Delta)}$ to be
	$$\wt_{L}(\Delta)=
	\begin{cases}
	0 &\mbox{ if no endpoints of $\Delta$ belong to $L$} \\
	\tchange{1/\ell} &\mbox{ if $\Delta$ has $\ell > 0$ endpoints that belong to $L$\;.}
	\end{cases}
$$
\end{definition}
Whenever it is clear for the context, we drop the subscript $L$ and use the
notation $\wt(\cdot)$ instead of $\wt_L(\cdot)$.

\begin{claim} \label{clm:bound_heavy}
If $(H,L)$ is appropriate and \Assum{estimates} holds, then
	the number of triangles with weight $0$ is at most $\ce \cdot \eps t$ for some constant $\ce$.
\end{claim}
% \mnote{D: double check constants}\mnote{T:done}
\begin{proof}
	\tchange{\sloppy By \Assum{estimates}, the number of vertices $v$ such that $d_v$ is greater than $(2\om/(\eps \ot)^{1/3}$, is at most $2m/(2\om/(\eps \ot)^{1/3}) \leq 6(\eps t )^{1/3}$, and the number of vertices $v$ such that $t_v > 2\ot^{2/3}/\eps^{1/3}$ is at most $3t/(2\ot^{2/3}/\eps^{1/3}) \leq 6(\eps t)^{1/3}$.
	Therefore, there are at most
	${12(\eps t )^{1/3} \choose 3 }< 2000 \eps t $
	triangles with all three endpoints in $H$.
	Setting $\ce=2000 $ completes the proof.}
	\end{proof}

%	Recall that $\Tr_e$ denotes the set of triangles that contain an edge $e$.
%\begin{definition} For an edge $e$ we define its weight $\wt{(e)}$ to be $\wt(e)= \sum\limits_{\Delta \in \Tr_e}\wt(\Delta)$
%\end{definition}

\begin{definition} \label{def:weight_vertex}
	For any set $T$ of triangles
	we define $\wt(T) = \sum\limits_{\Delta \in T} \wt(\Delta)$.
	For a vertex $v\in L$ we define \tchange{$\wt(v) = \sum\limits_{\Delta \in T_v} \wt(\Delta)$},
	and $\wt(v) = 0$ for $v \in H$.
\end{definition}

\begin{lemma} \label{lem:H}
For any partition $(H,L)$, $\sum\limits_{v \in L} \wt(v)  \leq t$. If $(H,L)$ is appropriate and \Assum{estimates} holds, then
$\sum\limits_{v \in L} \wt(v) \in [t(1-\ce \cdot \eps),t]$.
\end{lemma}

\begin{proof}
\tchange{
Let $\chi(v,\Delta)$ be an indicator variable such that $\chi(v,\Delta)=1$ if $\Delta$ contains the vertex $v$, and $\chi(v,\Delta)=0$ otherwise. Consider a triangle $\Delta$ that contains $\ell > 0$ light vertices.
	Then
	$$ \sum\limits_{v \in L} \chi(v,\Delta) = \ell =1/\wt(\Delta)\;.$$
	If $\ell = \wt(\Delta) = 0$, then the above expression equals $0$. By interchanging summations,
	$$ \sum\limits_{v \in L} \wt(v) = \sum\limits_{v \in L} \sum\limits_{\Delta \in T_v}\wt(\Delta) = \sum\limits_{\Delta \in T} \wt(\Delta) \sum\limits_{v \in L}  \chi(v,\Delta)
		= t - |\{\Delta \mid \wt(\Delta) = 0\}| .$$
	}
	Clearly for any partition $(H,L)$ the above expression is at most $t$. On the other hand,
		If $(H,L)$ is appropriate and \Assum{estimates} holds, then by \Clm{bound_heavy} we have that $|\{\Delta \mid \wt(\Delta) = 0\}| \leq \ce\cdot \eps t$, and the lemma follows.
\end{proof}

\begin{theorem} \label{thm:vert} Let $s = (c \log(n/\eps)/\eps^3) n/\ot^{1/3}$ where $c$ is a constant, and let $S$ be a sample of $s$  vertices $v_1, v_2,\ldots, v_s$ that are selected uniformly, independently at random. Then
$$\EX\left[\frac{1}{s}\sum\limits_{i=1}^{s} \wt(v_i)\right] \leq \frac{t}{n}\;.$$
Furthermore, if
 $(H,L)$ is appropriate and \Assum{estimates} holds, then
	 $$\EX\left[\frac{1}{s}\sum\limits_{i=1}^{s} \wt(v_i)\right] \in [t(1-\ce\cdot \eps)/n,t/n]$$ and for a
sufficiently large constant $c$,
$$\Pr\left[\frac{1}{s}\sum\limits_{i=1}^{s} \wt(v_i) < t(1-2\ce\cdot \eps)/n\right] < \eps^2/n\;.$$
\end{theorem}

\begin{proof}
\sloppy {Let $Y$ denote the random variable $Y= \frac{1}{s}\sum\limits_{i =1}^s \wt(v_i)$.
 By the first part of \Lem{H},  $\EX\left[\frac{1}{s}\sum\limits_{i=1}^{s} \wt(v_i)\right] \leq t/n$.
Now assume that $(H,L)$ is appropriate and \Assum{estimates} holds.
The claim regarding the expected value of $Y$ follows from
the second part of \Lem{H}, so it remains to prove the claim regarding the deviation from the expected value.
	Note that  $\wt(v) \leq t_v$ for every vertex $v$,
	which for  $v \in L$ is at most $\frac{2\ot^{2/3}}{\eps^{1/3}}$.
	By the multiplicative Chernoff bound and by Item~\ref{assume t} in \Assum{estimates},

	$$\Pr\left[Y < (1-\eps)\EX[Y]\right] < \exp\left(-\frac{\eps^2\EX[Y] s}{4t^{2/3}/\eps^{1/3}} \right) <
	\exp\left( - \frac{\eps ^2\cdot c\log(n/\eps)(n/\eps \ot^{1/3})\cdot t/(2n)}{4\ot^{2/3}/\eps^{1/3}}\right)
	< \frac{\eps^2}{n}\;,$$
where the last inequality holds for a sufficiently large constant $c$.}
\end{proof}

\subsection{A procedure for deciding whether a vertex is heavy}\label{subsec:heavy-proc}
In this subsection we provide a procedure for deciding (approximately) whether a given vertex $v$ is heavy or light.
\dddchange{Recall that a high-level description of the procedure appears in Subsection~\ref{intro-heavy:subsubsec} of the introduction.}
%\mnote{D: should add some text (even if might repeat some of what is written in intro}

\bigskip
\fbox{
	\begin{minipage}{0.9\textwidth}
		{\bf \heav$(v)$}
		\smallskip
		\begin{compactenum}
			\item If $d_v > 2\om/(\eps \ot)^{1/3}$, output \textbf{heavy}.
			\item For $i = 1,2,\ldots,10\log n$:
			\begin{compactenum}
				\item For $j = 1,2,\ldots,s=\tchange{20}\om^{3/2}/\eps^2 \ot$: \label{step:a}
				\begin{compactenum}
					\item Select an edge $e \in E_v$ uniformly, independently  and at random, and let $u$ be \tchange{the smaller endpoint according to the order $\prec$.}
					\item For $k = 1,2,\ldots,r=\ceil{d_u/\sqrt{\om}}$:
					\begin{compactenum}
						\item Pick a neighbor $w$ of $u$ uniformly at random. Let $x$ denote the endpoint of $e$ that is not $v$.
						\item If $e$ and $w$ form a triangle \emph{and} $x \prec w$, set $Z_k = d_u$, else $Z_k = 0$.
					\end{compactenum}
					\item Set $Y_j = \frac{1}{r}\sum\limits_k Z_k$.
				\end{compactenum}
				\item Set $X_i = \frac{d_v}{s} \sum\limits_j Y_j$.
			\end{compactenum}
			\item If the median of the $X_i$ variables is greater than $\ot^{2/3}/\eps^{1/3}$, output \textbf{heavy}, else output \textbf{light}.
		\end{compactenum}
	\end{minipage}}

	\medskip
	We have three nested loops, with loop variables $i,j,k$ respectively. We refer to these
	as ``iteration $i$", ``iteration $j$", and ``iteration $k$".

	\begin{lemma} \label{lem:x} For any iteration $i$, $\pr[|X_i - t_v| > \eps \cdot \max(t_v, \ot d_v/\om)] < 1/4$.
% \mnote{D: should it be $t/m$?}
	\end{lemma}
	
	\begin{proof} \tchange{Recall that we associate each triangle $(v,x,w)\in T_v$ with $(v,x)$ if $x \prec w$
    and with $(v,w)$ otherwise, so that we have $t_v = \sum_{e \in E_v} t_e$.}
 %   \mnote{D: or maybe we should modify $t_v$ so that will be
 %    the number of triangles that $v$ participates in?}  
 For an edge $e = (v,x)$, $t_e$
    is upper bounded by the number of neighbors $w$ of $x$ such that $x \prec w$. By \Clm{prec}, $t_e \leq \sqrt{2m}$.

    Fix an iteration $j$ and let $e_j$ denote the edge chosen in the $j\th$ iteration and $u_j$
		denote its smaller degree endpoint. We use $\cE_j$ to denote the event of $e_j$ being chosen.
		Conditioned on the event $\cE_j$, the probability that $Z_k$ is non-zero is
		is $t_{e_j}/d_{u_j}$. Hence,
		\begin{align*}
		\EX[Z_k \mid \cE_j] = \frac{t_{e_j}}{d_{u_j}}\cdot d_{u_j} = t_{e_j},
		\end{align*}
		and
		\begin{align*}
		\var[Z_k \mid \cE_j] \leq \EX[Z^2_k \mid \cE_j] \leq d_{u_j} \cdot \EX[Z_k \mid \cE_j].
		\end{align*}
		By linearity of expectation,
		\begin{equation}
		\EX[Y_j \mid \cE_j] = \EX\left[\frac{1}{r}\sum\limits_{k=1}^{r}Z_k \mid \cE_j\right]=\frac{1}{r} \sum\limits_{k=1}^{r}\EX\left[Z_k \mid \cE_j\right]
		=\tchange{t_{e_j}}. \numberthis \label{eq:exp_conditioned_Yj}
		\end{equation}
		By the independence of the $Z_k$ variables,
		\begin{align*}
		\var[Y_j \mid \cE_j] &=\var\left[\frac{1}{r}\sum\limits_{k=1}^{r}Z_k \mid\; \cE_j \right] =\frac{1}{r^2}\sum\limits_{k=1}^{r} \var\left[Z_k \mid \cE_j\right] \leq \frac{1}{r^2}\sum\limits_{k=1}^{r}  d_{u_j}\cdot \EX\left[Z_k \mid \cE_j\right] \\
		&=\frac{d_{u_j}}{r^2}\cdot r\cdot t_{e_j} \leq \sqrt {\om} \cdot t_{e_j}. \numberthis \label{eq:var_conditioned_Yj}
		\end{align*}
		The conditioning can be removed to yield
		\begin{align*}
		\EX[Y_j] =\sum\limits_{e\in E_v}\frac{1}{d_v} \cdot \EX[Y_j \mid \cE_j] = \frac{1}{d_v}\cdot \sum\limits_{e \in E_v} \tchange{t_e = \frac{t_v}{d_v}. }\numberthis \label{eq:exp_Yj}
		\end{align*}
		By the law of total variance, the law of total expectation, the bounds $t_{e_j} \leq \sqrt{2m}$
and \tchange{$m \leq 6\om$}, \tchange{and by Equations~\eqref{eq:exp_conditioned_Yj} and \eqref{eq:var_conditioned_Yj}}:
		\begin{align*}
		\var[Y_j] &=\EX _{e_j} \left[\var \left[Y_j \mid \cE_j \right]\right] + \var _{e_j}\left[\EX \left[Y_j \mid \cE_j \right]\right] \\
&\leq \EX_{e_j}\left[\sqrt {\om} \cdot \EX \left[Y_j \mid \cE_j\right]\right] + \var_{e_j}[t_{e_j}] \\
&=\sqrt {\om} \cdot \EX\left[Y_j\right] + \EX_{e_j}[t^2_{e_j}] \\
&= \sqrt {\om} \cdot \EX\left[Y_j\right] + \frac{1}{d_v}\cdot\sum_{e_j \in E_v} t^2_{e_j}\\
&\leq \sqrt {\om} \cdot \EX\left[Y_j\right] + \sqrt{2m}\cdot \EX[Y_j] < \tchange{5}\sqrt{\om}\EX[Y_j]\;.
 \numberthis \label{eq:var_Yj}
		\end{align*}
		Let $\overline{Y} = \frac{1}{s}\sum\limits_j Y_j$. \tchange{By Equation~\eqref{eq:exp_Yj},} $\EX[\bar{Y}] = {t_v}/{d_v}$. \tchange{By Equation~\eqref{eq:var_Yj},}			

		\begin{align*}
		 \var[\bar{Y}]&=\var\left[\frac{1}{s}\sum\limits_{j=1}^{s}Y_j\right]=\frac{1}{s^2}\sum\limits_{j=1}^{s}\var[Y_j] < \frac{1}{s^2}\sum\limits_{j=1}^{s}5\sqrt {\om}\cdot \EX\left[Y_j\right] = \frac{5\sqrt {\om}}{s}\cdot \EX\left[\frac{1}{s}\sum\limits_{j=1}^{s}Y_j\right] \\
		&= \frac{5\sqrt {\om}}{s}\EX\left[\bar{Y}\right]
\;= \frac{5\sqrt{\om}\cdot(t_v/d_v)}{20 \cdot \om^{3/2}/\eps^2\ot} \;=\;
  \frac{\eps^2}{4} \cdot\frac{t_v}{d_v}\cdot\frac{\ot}{\om}\;. \numberthis \label{bound_var_y_bar}
		\end{align*}
		By Chebyshev's inequality and Equation \eqref{bound_var_y_bar},
		\begin{align*}
		\pr\left[\left\lvert \overline{Y} - \frac{t_v}{d_v} \right\rvert> \eps \cdot\max\left(\frac{t_v}{d_v}, \frac{\ot}{\om}\right)\right] & \leq
		\frac{\var[\overline{Y}]}{\eps^2 \max(t_v/d_v,\ot/\om)^2} <
		\frac{1}{4}\;.
		\end{align*}
		Since $X_i=d_v \cdot \bar{Y}$, we have that $\pr[|X_i - t_v| > \eps \cdot \max(t_v, \ot d_v/\om)] < 1/4$.
	\end{proof}
	
	%\begin{lemma} \label{lem:heavy}
	%	For every vertex $v$, if $t_v > 2t^{2/3}/\eps^{1/3}$ or $d_v > 2m/(\eps t)^{1/3}$ then $v$ is declared \textsf{heavy} with probability at least $1-1/n^2$. If $t_v \leq t^{2/3}/2\eps^{1/3}$ and $d_v \leq 2m/(\eps t)^{1/3}$ then $v$ is declared \textsf{light} with probability at least $1-1/n^2$.
	%\end{lemma}

	\begin{lemma} \label{lem:heavy}
		For every vertex $v$, if $v$ is heavy, then a call to \heav($v$) returns \textbf{heavy} with probability at least $1-1/n^2$. If $v$ is light, then a call to \heav($v$) returns  \textbf{light} with probability at least $1-1/n^2$.
	\end{lemma}
	
	\begin{proof} First consider a heavy vertex $v$. Clearly, if $d_v > 2\om/(\eps \ot)^{1/3}$, then $v$ is declared heavy. Therefore, assume that $t_v > 2\ot^{2/3}/\eps^{1/3}$ and $d_v \leq 2\om/(\eps \ot)^{1/3}$, so that $\ot d_v/\om \leq 2\ot^{2/3}/\eps^{1/3}$, \tchange{and hence $\max(\ot d_v/\om, t_v)=t_v$}. By \Lem{x}, and since $\eps<\leq 1/2$, for any iteration $i$, $\Pr\left[|X_i - t_v| > \eps t_v\right] < 1/4$. Therefore, $\Pr[X_i < \ot^{2/3}/\eps^{1/3}] < 1/4$, and by Chernoff, the probability that the median of the  $X_i$ variables (where $i = 1,\dots, 10 \log n$) will be greater than $\ot^{2/3}/\eps^{1/3}$ is at least $1-1/n^2$. Hence \heav$(v)$ outputs \textbf{heavy} with probability at least $1-1/n^2$.
		
		Now consider a light vertex $v$. Since $d_v \leq 2\om/(\eps \ot)^{1/3}$ and $t_v \leq \ot^{2/3}/2\eps^{1/3}$, it holds that $\ot d_v/\om \leq 2\ot^{2/3}/\eps^{1/3}$, implying that $\max(\ot d_v/\om, t_v)=t_v \leq 2\ot^{2/3}/\eps^{1/3}$. Therefore, by \Lem{x}, $\Pr[|X_i - t_v| > \eps (2\ot^{2/3}/\eps^{1/3})] < 1/4$, and the probability that the median will be less than $\ot^{2/3}/\eps^{1/3}$ is at least $1-1/n^2$. Hence $v$ is declared \textbf{light} with probability at least $1-1/n^2$.
	\end{proof}
	
	The following is a corollary of \Lem{heavy}.
	\begin{corollary} \label{cor:heavy}
Consider running \heav~for all the vertices in the graph. Let $H$ denote the set of vertices that are declared heavy
and let $L$ denote the set of vertices that are declared light.
Then, with probability at least
$1-1/n$, the partition $(H,L)$ is appropriate (as defined in Definition~\ref{def:heavy}).
% all the heavy vertices are declared heavy and all the light vertices are declared light.
		%	Then the following holds with probability $> 1-1/n$. If $v$ is declared light, then $_v \leq 2t^{2/3}/\eps^{1/3}$. If $v$ is declared heavy, then $d_v > 2m/(\eps t)^{1/3}$ or $t_v > t^{2/3}/2\eps^{1/3}$.
	\end{corollary}

	%Obviously, there is an upper bound on the number of vertices with either high degree or many
	%triangles.
	%\begin{corollary} \label{cor:few_heavy} The number of heavy vertices is at most $3(\eps t)^{1/3}$.
	%\end{corollary}
	
	% The following is where the degeneracy bounds come into play. We give a direct, self-contained proof, %but note the connection to Chiba-Nishizeki's bound.
	We now turn to analyze the running time of \heav. The proof will be similar to the complexity analysis of the exact triangle counter of Chiba and Nishizeki~\cite{ChNi85}.
	
	\begin{lemma} \label{lem:xtime} If Item~\ref{assume m} in \Assum{estimates} holds, then for every vertex $v$ the expected running time of \heav$(v)$ is $O^*(\om^{3/2}/\ot)$.
	\end{lemma}
	
	\begin{proof} We first argue that the expected time to generate a single sample of $Y_j$
		is $O(1)$.
		Our query model allows for selecting an edge in $E_v$ uniformly at random by a single
		query.
If $d_v \leq \sqrt{\om}$, then the degree of the smaller endpoint for any $e \in E_v$
		is at most $\sqrt{\om}$. Hence a sample is clearly generated in $O(1)$ time. Suppose that $d_v > \sqrt{\om}$.
		If an 	edge $e = (v,u)$ is sampled, then the runtime is $O(1 + \min(d_v,d_u)/\sqrt{\om})$. Hence, the
		expected runtime to generate $Y_j$ is, up to constant factors, at most:
%		\begin{eqnarray*}
%			\frac{1}{d_v} \sum\limits_{u \in \Gamma_v} \left(1+\frac{\min\left\{ d_v,d_u\right\} }{\sqrt{\om}}\right)
%			& = & 1 + \frac{1}{\sqrt{\om} \cdot d_v} \sum\limits_{\substack{u \in \Gamma_v \\ d_u \leq d_v}} d_u +\frac{1}{\sqrt{\om}\cdot d_v} \sum\limits_{\substack{u \in \Gamma_v \\ d_u > d_v}} d_v \\
%			& \leq & 1 + (1/\om )\cdot \sum\limits_u d_u + \frac{|\{u \mid d_u > d_v\}|}{\sqrt {\om}}  \\
%			& \leq & 3 + \frac{2m/d_v}{{\om}} = O(1).
%		\end{eqnarray*}
%		
\begin{align*}
			\frac{1}{d_v} \sum\limits_{u \in \Gamma_v} \left(1+\frac{\min\left\{ d_v,d_u\right\} }{\sqrt{\om}}\right) \leq   1+  \frac{1}{\sqrt{\om} \cdot d_v} \sum\limits_{u \in \Gamma_v} d_u
			\leq 1+  \frac{1}{\sqrt{\om} \cdot d_v} \sum\limits_{u \in V} d_u \leq 1+\frac{2m}{\sqrt{\om} \cdot d_v} \leq 5,
		\end{align*}
		where the last inequality follows from  Item~\ref{assume m} in \Assum{estimates}

By the above, % the inner steps of
each iteration of the `for' loop in Step~\ref{step:a} takes $O(1)$ time in expectation. Therefore, %the inner for
together, all iterations of Step~\ref{step:a}
take $O(\om^{3/2}/(\eps\ot)) $ time in expectation, and since it is repeated $O(\log n)$ times, the expected running time of the procedure is  $(\om^{3/2}/\ot) \cdot \polylog$.
	\end{proof}

	\subsection{Estimating the number of triangles given $\om$ and $\ot$} \label{sec:full}
	
	We are now ready to present an algorithm \triest{} that takes $\om$, $\ot$  as input (``advice''), and
outputs an estimate of $\tr$.  Later, we employ the the average degree approximation algorithm of Feige~\cite{feige2006sums} and a geometric search to get the bonafide algorithm that estimates $\tr$ without any initial estimates $\om$ and $\ot$.
\dddchange{Recall that a high-level description of the procedure appears in Subsection~\ref{intro-alg:subsubsec} of the introduction.}
	In what follows we rely on the following assumption.
	\begin{assumption} \label{assum:apropriate}
		We will assume that the random coins used by \heav{} are fixed in advance, and that the partition $(H,L)$ as defined in \Cor{heavy} is indeed appropriate.
	\end{assumption}
	
	By \Cor{heavy} this assumption only adds $1/n$ to the error probability in all subsequent probability
	bounds. Recall that we use $c,c_1,\dots$ to denote sufficiently large constants.
	
%	\bigskip
\begin{figure}
	\fbox{
		\begin{minipage}{0.9\textwidth}
			{\bf \triest$(\om,\ot,\eps)$}
			
			\smallskip
			\begin{compactenum}
				\item Sample $s_1 = c_1 \eps^{-3}\log (n/\eps) (n /\ot^{1/3})$  vertices, uniformly, independently and at random. Denote the chosen multiset $S$. \label{step1}
				\item Set up a data structure to enable sampling vertices in $S$ proportional to their degree.
				\item For $i = 1, 2, \ldots, s_2 = c_2 \eps^{-4}(\log^2n)  (\om^{3/2} /\ot)$:
				
				\begin{compactenum}
					\item \label{step3} Sample $v \in S$ proportional to $d_v$ and sample  $e \in E_v$ uniformly at random. Let $u$ be \tchange{the smaller endpoint according to the order $\prec$.} Let $x$ be the endpoint of $e$ that is not $v$.
					\item If $d_u \leq \sqrt{\om}$, set $r=1$ with probability $d_u/\sqrt{\om}$ and set $r=0$ otherwise.
					If $d_u > \sqrt{\om}$, set $r = \ceil{d_u/\sqrt{\om}}$.
					\item Repeat for $j = 1,2,\ldots,r$:
					\begin{compactenum}
						\item Pick a neighbor $w$ of $u$ uniformly at random.
						\item If $e$ and $w$ do not form a triangle,  set $Z_j = 0$.
                        \item If $e$ and $w$ form a triangle and $w \prec x$, set $Z_j = 0$.
						\item If $e$ and $w$ form a triangle $\Delta$ and $x \prec w$: call \heav{} for all vertices in $\Delta$, and let
						%  If $v$ is heavy, set $Z_j = 0$. Otherwise, set $Z_j = \max(d_u,\sqrt{\om})\cdot \wt(\Delta)$.
						
						$ \;\;\;\;Z_j = \begin{cases}
						0 &\mbox{if $\heav(v)$ returned \textbf{heavy} }\\
						\max(d_u,\sqrt{\om})\cdot \wt(\Delta) &\mbox{otherwise}
						\end{cases}\;.$
					\end{compactenum}
					\item Set $Y_i = \frac{1}{r}\sum\limits_{j=1}^{r} Z_j$. (If $r=0$, set $Y_i = 0$.)
				\end{compactenum}
				\item Output $X = \frac{n}{s_1 s_2}\cdot \left(\sum\limits_{v \in S} d_v\right) \cdot \left(\sum\limits_{i=1}^{s_2} Y_i\right)$ .\label{step4}
			\end{compactenum}
		\end{minipage}}
\end{figure}
%\bigskip	

Recall that $\ce$ is the constant defined in Claim~\ref{clm:bound_heavy}.	
		\begin{theorem} \label{thm:estimate}
For $X$ as defined in Step~\ref{step4} of \triest, $E[X] \leq t$.	Moreover, if $(H,L)$ is appropriate and \Assum{estimates} holds, then $\EX[X] \in [t(1- \ce \cdot \eps),t]$ and $\Pr[X < t(1-3\ce \cdot \eps)] < 3\eps/\log n$.
		\end{theorem}
		
		There are three ``levels" of randomness. First is the choice of $S$,
		second is the choice of $e$ (Step~\ref{step3}), and finally the $Z_j$'s.
		When analyzing the randomness in any level, we condition on the previous levels.
		Before proving the theorem, we present the following definition and claim.
%  and definition.

		\begin{definition} \label{def:good_great}
			Let $S$ be a multiset of $s_1$ vertices. We say that $S$ is \textsf{good} if $\sum\limits_{v \in S}\wt(v)/s_1 \geq t(1-2\ce \cdot \eps)/n$. We say that $S$ is \textsf{great}
			if, in addition to being good, $d_S = \sum\limits_{v \in S} d_v \leq s_1(2m/n)(\log n/\eps)$.
		\end{definition}

\def\tildT{\widetilde{T}}
\def\tildwt{\widetilde{\wt}}
		\begin{claim} \label{clm:Y} Fix the choice of the set $S$, and let $d_S = \sum\limits_{v \in S} d_v$. For every $i$, $ \EX[Y_i \mid S] = d^{-1}_S \sum\limits_{v \in S} \wt(v)$ and $\var[Y_i \mid S] <  5\sqrt{\om}\cdot \EX[Y_i \mid S]$.
		\end{claim}
		
		\begin{proof} This is similar to the argument in \Lem{x}. Let $v_i$ be the chosen vertex in the $i\th$ iteration of the algorithm, and let $e_i$ be the chosen edge. We refer to this event by $\cE_i$, and condition over the set $S$ being chosen and the event $\cE_i$.
Denote by $u_i$ the lower degree endpoint of $e_i$.	
 %  Fix the choice of vertex $v_i$.
			If \heav($v_i$)=\textbf{heavy}, then $\EX[Y_i \mid S, \cE_i]=0$ and $\var[Y_i \mid S, \cE_i]=0$.
			If \heav($v_i$)=\textbf{light}, then there are two possibilities. If $d_{u_i} \leq \sqrt{\om}$ then,
			$$\EX[Y_i \mid S, \cE_i] = \frac{d_{u_i}}{\sqrt {\om}} \sum\limits_{\Delta \in T_{e_i}}
			\frac{1}{d_{u_i}} \cdot \sqrt{\om}\cdot\wt(\Delta)= \wt(T_{e_i})\;.$$
% D: Moved explanation about weight being wrt appropriate (H,L) to a bit earlier
			Since the maximum value of $Y_i$ in this case is at most $\sqrt{\om}$,
			\[\var[Y_i \mid S, \cE_i] \leq \EX[Y_i^2 \mid S, \cE_i] \leq \sqrt{\om} \cdot \EX[Y_i \mid S, \cE_i]\;.\numberthis \label{eq:var_yi_case1_est}\]
			Now consider the case that $d_{u_i} > \sqrt{\om}$.
            In order to bound the variance of the $Y_i$ variables we first analyze the expectation and variance of the $Z_j$ variables.
			Note that $Z_j$ is non-zero when a triangle $\Delta \in T_{e_i}$ is found.
            It holds that
			$$\EX[Z_j \mid S, \cE_i] = \sum\limits_{\Delta \in T_{e_i}} \frac{1}{d_{u_i}} \cdot d_{u_i} \cdot  \wt(\Delta)
			=  \wt(T_{e_i}),$$
			and \[\var[Z_j \mid S, \cE_i] \leq d_{u_i} \cdot\EX[Z_j \mid S, \cE_i]. \numberthis \label{eq:var_Zj_conditional}\]
			By linearity of expectation,
			$$\EX[Y_i \mid S, \cE_i] = \wt(T_{e_i}).$$
			By independence of the $(Z_j \mid S, \cE_i)$ variables, linearity of expectation \tchange{and Equation~\eqref{eq:var_Zj_conditional},}
			\begin{align*}
			\var[Y_i \mid S, \cE_i] &= \var
			\left[\frac{1}{r}\sum\limits_{j=1}^{r}Z_j \mid S,\cE_i\right]
			= \frac{1}{r^2}\sum \limits_{j=1}^{r} \var \left[Z_j \mid S,\cE_i\right]
			\leq \frac{1}{r^2}\sum \limits_{j=1}^{r} d_{u_i} \cdot \EX\left[Z_j \mid S,\cE_i\right] \\
			& =\frac{d_{u_i}}{r} \cdot \EX \left[\frac{1}{r} \sum\limits_{j=1}^{r}Z_j \mid S,\cE_i \right]
			\leq \sqrt{\om} \cdot  \EX[Y_i \mid S, \cE_i]. \numberthis \label{eq:var_yi_case2_est}
			\end{align*}
			We remove the conditioning on $\cE_i$:
			$$ \EX[Y_i \mid S] = \sum\limits_{v \in S \cap L} \frac{d_v}{d_S} \cdot \frac{1}{d_v} \sum\limits_{e \in E_v} \wt(T_{e}) = d^{-1}_S \sum\limits_{v \in S \cap L} \sum\limits_{e \in E_v} \wt(T_{e}) =  \tchange {d^{-1}_S \sum\limits_{v \in S} \wt(v)}\;.$$
%            Note that \tchange{$\sum\limits_{e \in E_v} \wt(T_e) = \wt(v)$},
%            since every triangle incident to $v$ is counted exactly once in the left term.
%            Thus, $\EX[Y_i \mid S] = d^{-1}_S \sum\limits_{v \in S } \wt(v)$.
%%
%          \mnote{D: there seems to be an issue here with a factor of 2. Fix locally or change defs?}

			\tchange{Recall that by Claim~\ref{clm:prec}, $\wt(e) \leq \sqrt{2m}$}. Therefore, by the law of total variance, the law of total expectation, the bound $m \leq 6\om$, \tchange{and Equations~\eqref{eq:var_yi_case1_est} and~\eqref{eq:var_yi_case2_est},}
			\begin{align*}
			\var[Y_i \mid S] &=
			\EX _{e_i} \left[ \var \left[ Y_i \mid S, \cE_i \right]\right] + \var _{e_i} \left[\EX \left[Y_i \mid S,\cE_i\right]\right] \\
            & \leq \EX _{e_i} \left[\sqrt{\om} \cdot \EX \left[Y_i \mid S, \cE_i \right]\right] + \EX_{e_i}[\wt(T_{e_i})^2] \\
            & \sqrt{\om} \cdot \EX[Y_i \mid S] + \sqrt{2m} \EX_{e_i}[\wt(T_{e_i})] < 5\sqrt{\om}\cdot \EX[Y_i \mid S] \;.
			\end{align*}
			This completes the proof of \Clm{Y}.
		\end{proof}

		\begin{proofof}{\Thm{estimate}}
For a fixed set $S$, let $X_S$ denote the sum $X_S = \frac{n}{s_1 s_2}\left(\sum\limits_{v \in S} d_v\right) \cdot \left(\sum\limits_{i=1}^{s_2} Y_i\right)$
% (similarly to the definition of $X$ in
(as defined in Step~\ref{step4} of \triest), given that the set $S$ in chosen in Step~\ref{step1}.
			By the definition of $X_S$ and by \Clm{Y},
			\[ \EX[X_S] = \frac{nd_S}{s_1}\EX[Y_i \mid S] = \frac{n}{s_1}\sum\limits_{v \in S}\wt(v). \numberthis \label{eq:exp_Xs}\]
			By \Thm{vert}, $\EX_S\left[\frac{1}{s_1}\sum\limits_{v\in S}wt(v)\right] \in [t(1-\ce\cdot \eps)/n,t/n]$, implying that
			$$\EX[X_S]\in [t(1-\ce\cdot \eps),t].$$
			By \Thm{vert}, \Def{good_great} and Assumption~\ref{assum:apropriate}, $S$ is good with probability at least $1-\eps^2/n$.
			The expected value, over $S$, of $d_S$ is $\EX_s\left[d_S\right]=s_1\cdot \frac{2m}{n}$.
			By Markov's inequality,
			$$\pr _S\left[d_S>s_1\cdot \frac{2m}{n}\cdot \frac{\log n}{\eps} \right] < \frac{\eps}{\log n}.$$
			By taking a union bound, the probability that $S$ is great is at least $1-2\eps/\log n$. For a fixed choice of $S$, let $Y_S = \frac{1}{s_2}\sum\limits_{i=1}^{s_2} Y_i$. \tchange{By the independence of the $Y_i$ variables and by Claim~\ref{clm:Y}},
			\[\var \left[Y_S\right] =\frac{1}{s_2^2}\sum \limits_{i=1}^{s_2}\var \left[Y_i \mid S\right] \tchange{<} \frac{1}{s_2^2} \sum \limits_{i=1}^{s_2} \tchange{5}\sqrt{\om} \cdot \EX[Y_i \mid S] \tchange{=} \frac{\tchange{5}\sqrt{\om}}{s_2}\cdot \EX\left[Y_S\right].\numberthis \label{eq:var_Ys} \]
			By Chebyshev's inequality, \tchange{the setting of $s_2$ and Equation~\eqref{eq:var_Ys}, we get that }
			\begin{align*}
			\Pr\big[|Y_S - \EX[Y_S]| > \eps \EX[Y_S]\big] &<
			\frac{\var[Y_S]}{\eps^2 \cdot\EX[Y_S]^2} \leq \frac{5\sqrt{\om} \cdot \EX[Y_S]}{\eps^2 (c_2 \eps^{-4}\log^2n) (\om^{3/2} /\ot) \cdot \EX[Y_S]^2} \\
			&= \frac{\eps^2}{c_2(\log^2n) (\om/\ot)\cdot \EX[Y_S]} \;.
			\end{align*}
			\tchange{By Claim~\ref{clm:Y}}, $\EX[Y_S] = d^{-1}_S \sum\limits_{v \in S} \wt(v)$, which for a great $S$
			is at least \tchange {
			$$ \frac{t(1-2\ce \cdot \eps)/n)/s_1}{s_1(2m/n)(\eps /\log n)} \geq \frac{t}{4m} \cdot \frac{\eps}{\log n}.$$}
		 Therefore, by \Assum{estimates}, for a sufficiently large constant $c_2$,
			$$\Pr\left[|Y_S - \EX[Y_S]| > \eps \EX[Y_S]\right] \leq \frac{\eps}{\log n}.$$
			By the definition of $X_S$ in Step~\ref{step4} of the algorithm, $X_S$ is just a scaling of $Y_S$. Therefore,
			$$\Pr\left[|X_S - \EX[X_S]| > \eps \EX[X_S]\right] \leq \frac{\eps}{\log n}.$$
			By Equation~\eqref{eq:exp_Xs}, $\EX[X_S] = \frac{n}{s_1}\sum\limits_{v \in S}\wt(v)$, which for a great $S$
			is at least $t(1-2\ce \cdot \eps)$. Hence, for a great $S$,
			$$\Pr\left[X_S < (1-3\ce \cdot \eps)\cdot t \right] \leq \frac{\eps}{\log n}.$$ The probability of $S$ not being
			great is at most $2\eps/\log n$. We apply the union bound to remove the conditioning, so we get
			$$\Pr\left[X < (1-3\ce \cdot \eps) \cdot t \right] \leq \frac{3\eps}{\log n},$$
			which completes the proof.
		\end{proofof}
		
		\begin{theorem} \label{thm:time} If Item~\ref{assume m} in \Assum{estimates} holds then the expected running time of  \triest{} is  $O^*(n/\ot^{1/3} + \om^{3/2}/\ot)$.
		\end{theorem}
		
		\begin{proof} The sampling of $S$ is done in $O^*(n/\ot^{1/3})$ time.
% The time required to generate the
Generating the $Z_j$ variables, without the calls to \heav{}, takes time  $O^*(\om^{3/2}/\ot)$ in expectation, by an argument identical to that in the proof of \Lem{xtime}. Therefore, it remains to bound the running time resulting from calls to \heav{}.
			
			Let us compute the expected number of triangles found during the run of the algorithm. In each iteration $i$, conditioned on choosing an edge $e$, the expected number of triangles found is at most $2(d_u/\sqrt{\om})(t_e/d_u) = 2t_e/\sqrt{\om}$.
			Averaging over the edges, the expected number of triangles found in a single iteration is at most $6t/(m \cdot \sqrt{\om})$, which by Item~\ref{assume m} in \Assum{estimates} is $O(\ot/{\om}^{3/2})$.
			There are $O(\om^{3/2}/\ot)\cdot \polylog$ iterations, leading to a total of $O^*(1)$ expected triangles.
			Thus, there are $O^*(1)$ expected calls to \heav, each taking $O^*(\om^{3/2}/\ot)$ time by \Lem{xtime}.
			Together with the above, we get an expected running time of $O(n/\ot^{1/3} + \om^{3/2}/\ot) \cdot \polylog$.
		\end{proof}
		
\iffalse		
	In order to remove the need for prior knowledge on $\om$ and $\ot$, we first analyze the behavior of \triest{} when running values of $\om$ and $\ot$ for which Assumption~\ref{assumption overline delta and m} does not hold.
\fi

	\subsection{The final algorithm}\label{subsec:final-alg}
%	In order to get an estimation of the number of edges in the graph we will use a slightly modified version of the \texttt{Average-Degree-Approximation-Algorithm}$(G,\eps')$ from \cite{GR08}. The algorithm in \cite{GR08} returns an approximation $\overline{d}$ of the number of edges in $G$ such that $\overline{d} \in [(1-2\eps)d_{avg}, d_{avg}]$ with probability at least $2/3$, where $d_{avg}$ is the average degree of $G$. By making some small adaptations the success probability can be improved
	We are now ready to present an algorithm that requires no prior knowledge regarding $m$ and $t$.

\bigskip
			\fbox{
				\begin{minipage}{0.9\textwidth}
					{\bf \triestt$(\eps)$}
					
					\smallskip
					\begin{compactenum}
						\item Let $\eps'=\eps/3\ce,$ where $\ce$ is the constant defined in Claim~\ref{clm:bound_heavy}.	
						\item Invoke Feige's
% \texttt{Average-Degree-Approximation-Algorithm}$(G,\eps')$ from \cite{GR08} for $c\log n$ times.
algorithm~\cite{feige2006sums} for approximating the average degree of a graph $10\log n$ times.
Let $\overline{d}$ be the median value of all invocations.		
						\item Let $\om = n\overline{d}/2$. \label{step:3}
%						$\om \leftarrow$ Approximate-Avg-Deg$(G,\eps')$\cite{GR08}.
						\item Let $\lt = n^3$.
						\item While $\lt \geq 1$ \label{step:t tilde}
						\begin{compactenum}
							\item For $\ot=n^3,n^3/2, n^3/4, \ldots, \lt$: \label{step:t overline}
							\begin{compactenum}
								\item For $i=1,\ldots,c\eps^{-1} \log \log n $: \label{step:5a1}
								\begin{compactenum}
									\item  Let $X_i=$\triest$(\eps', \om, \ot)$.\label{step:5a1A}
								\end{compactenum}
								\item Let $X=\min_{i}\{X_i\}$.
								\item If $X \geq \ot$ \textbf{return} $X$. \label{step:5a3}
							\end{compactenum}		
							\item Let $\lt=\lt/2$.
						\end{compactenum}
					\end{compactenum}
					
				\end{minipage}}

		\bigskip
		Before analyzing the correctness and running time of the algorithm, we present the following simple proposition, whose proof we give for the sake of completeness.
			\begin{proposition} \label{prop: m = theta Delta(G) 2/3} For every graph $G$, $t \leq \frac{4}{3}m^{3/2}$.
			\end{proposition}
			\begin{proof}
				\begin{align*}
				t = \frac{1}{3} \sum\limits_{v\in V} \frac{1}{2}\tr_v \leq \frac{1}{6}\left(\sum\limits_{v:\; d_v> \sqrt m} \tr_v + \sum\limits_{v:\; d_v\leq \sqrt m}2d_v^2 \right) \leq \frac{1}{6}\left(2\sqrt m\cdot 2m + 2\sqrt m\sum\limits_{v: \;d_v\leq \sqrt m}d_v \right)
				\leq \frac{4}{3}m^{3/2}.
				\end{align*}
			\end{proof}

		\begin{theorem} \label{thm:maintri}
\sloppy
Algorithm \triestt$(\eps)$ returns
% here exists a $O^*(n/t^{1/3} + m^{3/2}/t)$ algorithm that provides
a value $X$, such that $(1-\eps)t \leq X \leq (1+\eps)t$, with probability at least $5/6$.
The expected query complexity of the algorithm is $O^*\left(n/t^{1/3} + \max\left\{m,m^{3/2}/t\right\}\right)$
and the expected running time of the algorithm is $O^*(n/t^{1/3} + m^{3/2}/t)$.
		\end{theorem}

\begin{proof}
	We first prove that the value of $X$ is as stated in the theorem.
Let $d_{avg}$ denote the average degree of vertices in $G$.
The %\texttt{Average-Degree-Approximation-Algorithm}$(G,\eps)$ by
algorithm from~\cite{feige2006sums} returns a value $\overline{d}$
such that, with probability at least $2/3$, $\overline{d} \in [d_{avg}/(2+\gamma), d_{avg}]$ for a constant $\gamma$.
%\mnote{D: doesn't a constant factor estimate suffice? (so can use Feige just as well)}
         Since we take the median value of $10\log n$ invocations, it follows from Chernoff's inequality that $\om$ is as stated in Item~\ref{assume m} of \Assum{estimates} with probability at least $1-1/\poly(n)$. %$1-1/n^4$.
         Assume that this is indeed the case.

% Assuming that $\om$	is as stated in \Assum{estimates}, we explain how to obtain a good estimate of $t$.
Before analyzing the algorithm \triestt\ as described above, first consider
executing Step~\ref{step:t overline} with $\lt=1$. That is, rather than running both an outer loop
over decreasing values of $\lt$ and an inner loop over decreasing values of $\ot$,  we only run a single loop over decreasing value of $\ot$, starting with $\ot = n^3$.
% running the algorithm \triest$(\eps, \om,\ot)$ with decreasing values $\ot=n^3,n^3/2, n^3/4, \ldots$.}
	By the first part of \Thm{estimate} and by Markov's inequality, for each value of $\ot$ and for
each $i$,
$\Pr[X_i \leq (1+\eps)t] > \eps/2$, where $X_i$ as defined in Step~\ref{step:5a1A}. Therefore, for each value of $\ot$,
the minimum estimate $X$ (as defined in Step~\ref{step:5a3}) is at most $(1+\eps)t$, with probability at least $1-1/\log^3n$.
	It follows that for each $\ot$ such that  $\ot > 2t$, we have that $X< \ot$ with probability at least $1-1/\log^3n$, and the algorithm will continue with $\ot=\ot/2$.
	Once we reach a value of $\ot$  for which $t/4 \leq \ot \leq t/2$, Item~\ref{assume t} in \Assum{estimates}, regarding $\ot$, holds. By the second part of \Thm{estimate}, $X_i \in [(1-\eps)t,(1+\eps)t]$ for every $i$ with probability at least $1-c/\log n$.
	Hence, we have that
	$$ \ot \leq \frac{1}{2}t \leq (1-\eps)t \leq X \leq (1+\eps)t,$$
	with probability at least $1-c/	\log n$. Therefore, we halt and return correct $X$.
	
	If however we do reach a value $\ot$ such that $\ot \leq t/4$, since \Assum{estimates} does not hold, we cannot lower bound $X$, implying that we can no longer bound the probability that $X < \ot$. Therefore we might continue running with decreasing values of $t$, causing the running time to exceed the desired bound of $O^*(n/t^{1/3}+m^{3/2}/t)$. In order to avoid this scenario, we run both an outer loop over $\lt$ and an inner loop over $\ot$. Specifically, starting
with $\lt = n^3$,
whenever we halve $\lt$, we run over all values of $\ot = n^3, n^3/2, \ldots$, until we reach $\lt$.
% whenever we halve the estimation of $t$,  we first run with all values $n^3, n^3/2, \ldots$ until we reach the current estimation.
This implies that for every value of $\lt > 2t$ the probability of returning an incorrect estimate, that is,
outside the range of $(1-\eps)t \leq X \leq (1+\eps)t$, is
at most $1-1/\log\tchange{^2} n$. On the other hand, for values  of $\lt$ such that $\lt \leq t/2$ the probability of returning a correct estimate (within $(1-\eps)t \leq X \leq (1+\eps)t$) is at least $1-c/\log n$.
	A union bound over all failure probabilities gives a success probability of at least $5/6$.

	\sloppy
We now turn to analyze the query complexity and running time of the algorithm.
By \cite{feige2006sums}, the expected running time of the average degree approximation algorithm
is $O^*(n/\sqrt m)$. By \Thm{time}, conditioned on $\om$ satisfying Item~\ref{assume m} in \Assum{estimates}, the expected running time of \triest{$(\eps,\om,\ot)$} is $O^*(n/\ot^{1/3}+\om^{3/2}/t)$. It follows from Proposition~\ref{prop: m = theta Delta(G) 2/3} that $n/\sqrt m = O(n/t^{1/3}$), implying that the running time is determined by the value of $\om$ and by the smallest value of $\ot$ that \triest{$(\eps, \om, \ot )$} is invoked with.
	
	Recall that whenever we halve the value of $\tilde{t}$, we  run with all values $\ot=n^3, n^3/2,\ldots$.
This, together with the fact that when running with $t/4\leq \ot \leq t/2$ we halt with probability at least $1-c/\log n$, implies that the probability of reaching a value $\tilde{t} = t/ 2^k$ is at most $(c/\log n)^k$. Therefore, the expected running time, conditioned on $\om$ satisfying Item~\ref{assume m} in \Assum{estimates}, is bounded by
	$$ \log ^2 n  \cdot O^*\left(\frac{n}{t^{1/3}} + \frac{\om^{3/2}}{t}\right) + \sum_{k=1}^{\log n}(c/\log n)^k \cdot 2^k \cdot  O^*\left(\frac{n}{t^{1/3}} + \frac{\om^{3/2}}{t}\right) = O^*\left(\frac{n}{t^{1/3}} + \frac{\om^{3/2}}{t}\right)  .$$

	Now consider the value of $\om$ computed in Step~\ref{step:3} of \triestt{$(\eps)$}.
As stated previously, with probability at least $1-1/\poly(n)$ (e.g., $1-1/n^4$), the estimate
$\om$ is within a constant factor from $m$.
Therefore the expected running time of the algorithm (without the conditioning on the value of $\om$) is bounded by

	$$\left(1-\frac{1}{n^4}\right)\cdot O^*\left(\frac{n}{t^{1/3}} + \frac{m^{3/2}}{t}\right) + \frac{1}{n^4}\cdot O(n^3) = O^*\left(\frac{n}{t^{1/3}} + \frac{m^{3/2}}{t}\right).$$
	
	Observe that we can always assume that the algorithm does not perform queries it can answer by itself. That is, we can allow the algorithm to save all the information it obtained from past queries, and assume it does not query for information it can deduce from its past queries. Further observe that any pair query is preceded by a neighbor query. Therefore, if at any point the algorithm performs more than $2\om$ queries, it can abort.
	It follows that the expected query complexity is $O^*(n/t^{1/3}+\min\{m,m^{3/2}/t\})$.
	\end{proof}

%%%%%%%%%%
\iffalse
\begin{proof}
{ \color{gray}
		First, we need to estimate $m$.
			This can be done with suitably high probability by the algorithm of Goldreich-Ron~\cite{GR08} in $O^*(n/\sqrt{m})$ time.
			(We give an independent proof in the next section of \Thm{maindeg}, which gives the desired algorithm.)

			We perform a geometric search for $t$, by guessing its value as $n^3, n^3/2, n^3/2^2, \ldots$.
			For each guessed value, we run the procedure of \Thm{estimate} independently $c\eps^{-1}\log\log n$ times, and take the minimum estimate. By Markov's inequality,  $\Pr[X \leq (1+\eps)t] > \eps/2$.
			Therefore, the minimum estimate	is at most $(1+\eps)t$, with probability at least $1-1/\log^2n$, .
			
			If the estimate is larger than the current guess of $t$, we halve the guess for $t$.
			Eventually, we reach an appropriate guess where the estimates of \Thm{estimate}
			to kick in. At the stage,
			the minimum of $c\eps^{-1}\log\log n$ estimates is at most $(1+\eps)t$
			and at least $(1-\eps)t$ with probability at least $1-1/\log n$ (we rescale $\eps$ from \Thm{estimate}). A union bound
			over all errors completes the proof.
		}
		\end{proof}
\fi
%%%%%%%%%%%%%

%% file: LowerBound.tex
\sect{A Lower Bound}\label{sec:lb}
In this section we present a lower bound on the number of queries necessary for estimating the number of triangles in a graph.
\ifnum\focs=1
Due to space constraints, in this extended abstract we provide only partial details of the proof.
All details can be found in the full version of this paper~\cite{ELRS}.
\fi
Since we sometimes refer to the number of triangles in different graphs, we use the notation $\tr(G)$ for the
number of triangles in a graph $G$.
 Our lower bound  matches our upper bound in terms of the dependence on $n$, $m$ and $\tr(G)$, up to polylogarithmic factors in $n$ and the dependence in $1/\epsilon$. In what follows, when we refer to approximation algorithms for the number of triangles in a graph, we mean multiplicative-approximation algorithms that output with high constant probability an estimation $\wht$ such that $\tr(G)/C \leq \wht\leq C\cdot \tr(G)$ for some predetermined approximation factor $C$.
% If $C$ is constant, then the algorithm is a constant-factor approximation algorithm. \par

We consider multiplicative-approximation algorithms that are allowed the following three types of queries: Degree queries, pair queries and random new-neighbor queries. Degree queries and pair queries are as defined in Section~\ref{sec:prel}. A random new-neighbor query $q_i$ is a single vertex $u$ and the corresponding answer is a vertex $v$ such that $(u,v)\in E$ and the edge $(u,v)$ is selected uniformly at random among the edges incident to $u$ that have not yet been observed by the algorithm.
\ifnum \focs =1
It is not hard to verify (as we show in the full version of this paper~\cite{ELRS})
\else
In Corollary~\ref{n over delta third cor} we show
\fi
that this implies a lower bound when the algorithm may perform (standard) neighbor queries instead of random new-neighbor queries. \par
We first give a simple lower bound that depends on $n$ and $\tr(G)$.

\begin{theorem} \label{n over delta third} Any multiplicative-approximation algorithm for the number of triangles in a graph must perform $\Omega\left(\frac{n}{\tr(G)^{1/3}}\right)$ queries, where the allowed queries are degree queries, pair queries and random new-neighbor queries.
\end{theorem}
\begin{proof}
For every $n$ and every $1 \leq \tr \leq \binom{n}{3}$ we next define a graph $G_1$ and a family of graphs $\mG_2$ for which the following holds. The graph $G_1$ is the empty graph over $n$ vertices. In
% the second family
$\mathcal{G}_2$, each graph consists of a clique of size $\lf\tr^{1/3}\rf$ and an independent set of size $n-\lf\tr^{1/3}\rf$. See Figure~\ref{Fig:lb_families} for an illustration. Within
% the second family
$\mathcal{G}_2$
the graphs differ only in the labeling of the vertices. By construction, % the graph in the first family
$G_1$ contains no triangles and each graph in % the second family
$\mathcal{G}_2$
contains $\Theta(\tr)$ triangles.
Clearly, unless the algorithm ``hits'' a vertex in the clique it cannot distinguish between
the two cases.
%In order to distinguish between a random graph in the first family and a random graph in the second family it is necessary to ``hit'' a vertex from the clique.
The probability of hitting such a vertex in a graph selected uniformly at random from
$\mG_2$ is $\lf\tr^{1/3}\rf/n$. Thus, in order for this event to occur with high constant probability, $\Omega\left(\frac{n}{\tr^{1/3}}\right)$ queries are necessary.
\end{proof}
\begin{figure}[ht!]
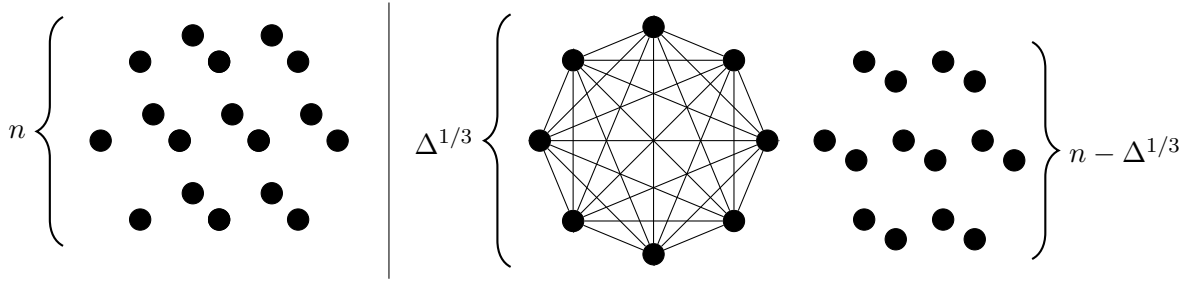

	\centering
	\ifnum\nofigures=0
	\DrawTwoFamilies
	\fi
	\caption{An illustration of the two families.}
	\label{Fig:lb_families}
\end{figure}

% \ni
We next state our main theorem. % def {Approx-Triangles-With-Advice}'{T: Should a line after a figure be indented?}

\begin{theorem} \label{m 3/2 over delta} Any multiplicative-approximation algorithm for the number of triangles in a graph must perform at least $\Omega\left( \min\left\{\frac{m^{3/2}}{\tr(G)},m\right\} \right)$ queries, where the allowed queries are degree queries, pair queries and random new-neighbor queries.
\end{theorem}
For every $n$, every $1\leq m\leq \binom{n}{2}$ and every $1 \leq \tr \leq \min\left\{\binom{n}{3},m^{3/2}\right\}$ we define a graph $G_1$ and a family of graphs $\mG_2$ for which the following holds. The graph $G_1$ and all the graphs in $\mG_2$ have $n$ vertices and $m$ edges. For the graph $G_1$, $\tr(G_1)=0$, and for every graph $G\in \mG_2$, $\tr(G)=\Theta(\tr)$. We prove it is necessary to perform $\Omega\left(\min\left\{\frac{m^{3/2}}{\tr},m\right\}\right)$ queries in order to distinguish with high constant probability between
% a random graph in the first family and
$G_1$ and a random graph in
% the second family.
$\mG_2$.
For the sake of simplicity, in everything that follows we assume that $\sqrt m$ is even.

We prove that for values of $\tr$ such that $\tr < \frac{1}{4}\sqrt m$, at least $\Omega(m)$ queries are required, and for values of $\tr$ such that $\tr \geq \sqrt m$ at least $\sloppy \Omega\left(\frac{m^{3/2}}{\tr}\right)$ queries are required.
\ifnum \focs=1
For the former case we refer the reader to the full version of this paper~\cite{ELRS}, and turn to
\else

We delay the discussion on the former case to Subsection~\ref{section Delta sqrt m}, and start with
\fi
the case that $\tr \geq \sqrt m$.
 Our construction of $\mG_2$ % the second family
depends on the value of $\tr$ as a function of $m$ where we deal separately with the following 
two ranges of $\tr$:

\begin{enumerate}
\item $\tr \in [\Omega(m),O(m^{3/2})]$. % $ \Omega(m)\leq \tr \leq O(m^{3/2})$.
\item $\tr \in [\Omega(\sqrt{m},O(m)]$. % $ \Omega(\sqrt m) \leq \tr \leq O(m)$.
\end{enumerate}
We prove that for every $\tr$ as above, % at least
$\Omega(m^{3/2}/\tr)$ queries are needed in order to distinguish between the graph $G_1$ and a random graph in $\mG_2$.
Observe that by Proposition~\ref{prop: m = theta Delta(G) 2/3}, for every graph $G$, it holds that $\tr(G) =O\left(m^{3/2}\right)$. Hence, the above ranges indeed cover all the possible values of $\tr$ as a function of $m$.

\ifnum\focs=1	
\paragraph{A high level discussion of the lower bound} \label{par: lb high level}
\else
\paragraph{A high level discussion of the lower bound.} \label{par: lb high level}
\fi
\tchange{The constructions for the different ranges of $t\geq \sqrt m$  are all based on the same basic idea, and have the following in common. In all construction for $t$ as above, $G_1$ consists of a complete bipartite graph $(L\cup R, E)$ with $|L|=|R|=\sqrt m$ and an independent set of $n-2\sqrt m$ vertices. The basic structure of the graphs in the family $\mG_2$ is the same as that of $G_1$ with the following modifications:
	\begin{itemize}
		\item For every value of $t$, we add $t/\sqrt m$ edges between vertices in $L$ (and similarly in $R$). Since each edge contributes (roughly) $\sqrt m$ triangles, this gives the desired total number of triangles in the graph. In the case that $t=m$ this is done by adding a perfect matching within $L$ and a perfect matching within $R$. In the case that $t >m$ we add several such perfect matchings, and in the case that $\sqrt m \leq t \leq m/4$ we add a (non-perfect) matching of size $t/\sqrt m$.
		\item In order to maintain the degrees of all the vertices in the bipartite component, we remove edges between vertices in $L$ and $R$.
	\end{itemize}
For an illustration of the case $t=m$, see Figure~\ref{fig:bipartite}.
In what follows we assume that the algorithm knows in advance which vertices are in $L$ and which are in $R$, and consider only the bipartite component of the graphs. In order to give the intuition for the $m^{3/2}/t$ lower bound we consider each type of query separately, starting with degree queries.

Since both in the graph $G_1$ and in all the graphs in $\mG_2$, all the vertices in $L\cup R$ have the same degree (of $\sqrt m$), degree queries do not reveal any information that is useful for  distinguishing between the two.

As for pair queries, unless the algorithm queries a pair in $L \times L $ (or $R \times R$) and receives a positive answer, or queries a pair in $L \times R$ and receives a negative answer, the algorithm cannot distinguish between the bipartite component of the graph $G_1$ and those of the graphs in $\mG_2$. We refer to these pairs as 
{\em witness pairs\/}. Roughly speaking, since there are $\Theta(t/\sqrt m)$ such pairs, and $m$ pairs in total,  it takes $\Omega(m^{3/2}/t)$ queries in order to ``catch a witness pair''.

We are left to deal with neighbor queries. Here too, distinguishing between the graph $G_1$ and the graphs in $\mG_2$ can be done by ``catching a witness''. That is, if the algorithm queries for a neighbor of a vertex in $L$ and the answer is another vertex in $L$ (analogously for a vertex in $R$). As before, the probability for hitting such a witness pair is small. However, there is another source of difference resulting from neighbor queries. When the algorithm queries a vertex $v\in L$ there is a difference in the conditional distribution on answers $v \in R$ when the answer is according to the graph $G_1$ or according to a graph in the family $\mG_2$. The reason for the difference, is that in the graph $G_1$ every vertex  has exactly $\sqrt m$ neighbors in the opposite side, while  for graphs in $\mG_2$, each vertex has $\Theta(\sqrt m - t/m)$ neighbors in the opposite side (for the range $\Omega(\sqrt m) \leq t \leq O(m)$ this is true on average). 
We prove that  this difference in sufficiently small
so as to ensure the 
% In later sections we formally prove (for the case that $t\sqrt m$) that unless $\Omega(m^{3/2}/t$ queries are performed, this difference in small. Hence we get a 
$\Omega(m^{3/2}/t)$ lower bound.

Our formal analysis is based on defining two processes that interact with an algorithm for approximating the
number of triangles, denoted ALG. The first process answer queries according to $G_1$, and the second process
answers queries while constructing a uniformly selected graph in $\mG_2$. An interaction between ALG and
each of these processes induces a distribution over sequences of queries and answers. We prove that if the
number of queries performed by ALG is smaller than $m^{3/2}/(c t)$ for a sufficiently large constant $c$,
then the statistical distance between the two distributions is a small constant. 

}

\medskip
\ifnum \focs =1
\ddchange{In this extended abstract we focus on the case that $\tr= m$ \tchange{which is a special case of the construction for the range $  t \in [\Omega(m),O(m^{3/2})]$}.
% and deal with the \tchange{other cases in the full version.}
Before doing so}
\else
We start by addressing the case that $\tr= m$ in Subsection~\ref{subsec:ksqrtm}, and deal with the
 case that $ m< \tr \leq \frac{m^{3/2}}{8}$ in Subsection~\ref{section r sqrt m}, and with the case that $ \sqrt m \leq \tr \leq \frac{m}{4}$ in Subsection~\ref{k leq sqrt m}. \par

Before embarking on the proof for $\tr=m$,
 \fi
 we introduce the notion of a {\em knowledge graph\/} (as defined previously
in e.g.,~\cite{GR-bound}),
which will be used in all lower bound proofs. Let ALG be an algorithm for approximating he number of triangles, which performs $Q$ queries. Let $q_t$ denote its $t\th$ query and let $a_t$ denote the corresponding answer. Then ALG is a (possibly probabilistic) mapping from {\em query-answer histories} $\pi\eqdef\langle(q_1,a_1),\ldots ,(q_t, a_t)\rangle$ to $q_{t+1}$, for every $t < Q$, and to $\mathbb{N}$ for $t = Q$. \par

We assume that the mapping determined by the algorithm is determined only on histories that are consistent with the graph $G_1$ or one of the graphs in $\mG_2$.
Any query-answer history $\pi$ of length $t$ can be used to define a knowledge graph $G^{kn}_\pi$ at time $t$. Namely, the vertex set of $G^{kn}_\pi$ consists of $n$ vertices. For every new-neighbor query $u_i$ answered by $v_i$ for $i\le t$, the knowledge graph contains the edge $(u_i,v_i)$, and similarly for every pair query $(u_j, v_j)$ that was answered by $1$. In addition, for every pair query $(u_i,v_i)$ that is answered by $0$, the knowledge graph maintains the information that $(u_i,v_i)$ is a non-edge. The above definition of the knowledge graph is a slight abuse of the notation of a graph since $G^{kn}_\pi$ is a subgraph of the graph tested by the algorithm, but it also contains additional information regarding queried pairs that are not edges. For a vertex $u$, we denote its set of neighbors in the knowledge graph by $\Gamma^{kn}_\pi(u)$, and let $d^{kn}_\pi(u)=\left|\Gamma^{kn}_\pi(u)\right|$. We denote by $N^{kn}_\pi(u)$ the set of vertices $v$ such that $(u,v)$ is either an edge or a non-edge in $G^{kn}_\pi$.

\ifnum\focs=0
\subsect{A lower bound for $\tr=m$}
\label{subsec:ksqrtm}
\subsubsect{The lower-bound construction} \label{Construction for k=sqrt m}
\else
\subsect{The lower-bound construction} \label{Construction for k=sqrt m}
\fi
	The graph $G_1$ has two components. The first component is a complete bipartite graph with $\sqrt m$ vertices on each side, i.e, $K_{\sqrt{m},\sqrt{m}}$, and the second component is an independent set of size $n-2\sqrt{m}$. We denote by $L$ the set of vertices $\ell_1, \ldots, \ell_{\sqrt m}$ on the left-hand side of the bipartite component and by $R$ the set of vertices $r_1, \ldots, r_{\sqrt m}$ on its right-hand side. The graphs in the family $\mathcal{G}_2$ have the same basic structure with a few modifications. We first choose for each graph a perfect matching $M^C$ between the two sides $R$ and $L$ and remove the edges in $M^C$ from the graph. We refer to the removed matching as the \textsf{``red matching"} and its pairs as \textsf{``crossing non-edges"} or \textsf{``red pairs"}. Now, we add two perfect matching from $L$ to $L$ and from $R$ to $R$, denoted $M^L$ and $M^R$ respectively. We refer to these matchings as the \textsf{blue matchings} and their edges as \textsf{``non-crossing edges''} or \textsf{``blue pairs"}. \dchange{Thus for each choice of three perfect matchings $M^C$, $M^L$ and $M^R$ as defined above, we have a corresponding graph in $\mathcal{G}_2$.}
	
	Consider a graph $G\in \mG_2$. Clearly, every blue edge participate in $\sqrt m-2$ triangles. Since, every triangle in the graph contains exactly one blue edge, there are $2\sqrt m \cdot (\sqrt m-2) = \Theta(m)$ triangles in $G$.
%	\mnote{T: we should state somewhere that all the graphs in $\mG_2$ indeed have $m$ triangles}

\begin{figure}[!htb]
	\centering
\ifnum\nofigures=0
	\begin{tikzpicture}[thick, x=2cm, scale=0.7]
	\CompleteMatching
	\end{tikzpicture}
\fi
	\caption{An illustration of the family $\mathcal{G}_2$ for $\tr=m$.}
	\label{fig:bipartite}
\end{figure}
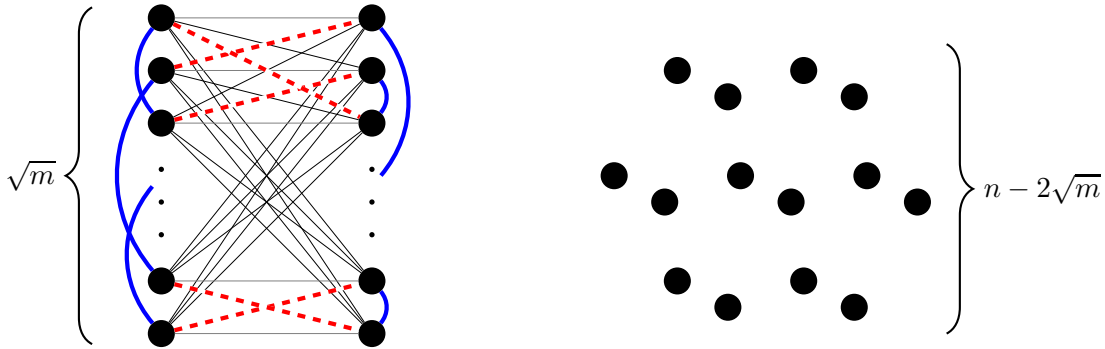

\ifnum\focs=0
\subsubsect{Definition of the processes ${P_1}$ and ${P_2}$} \label{processes}
\else
\subsect{Definition of the processes ${P_1}$ and ${P_2}$} \label{processes}
\fi
\dchange{
In what follows we describe two random processes, $P_1$ and $P_2$, which interact with an arbitrary algorithm ALG. The process $P_1$ answers ALG's queries consistently with $G_1$. The process $P_2$ answers ALG's queries while constructing a
uniformly selected random graph from $\mathcal{G}_2$. We assume without loss of generality that ALG does not ask queries whose answers can be derived from its knowledge graph, since such queries give it no new information. For example, ALG does not ask a pair query about a pair of vertices that are already known to be connected by an edge due to a neighbor query. Also, we assume ALG knows in advance which vertices belong to $L$ and which to to $R$, so that ALG need not query vertices in the independent set. Since the graphs in $\mG_2$ differ from $G_1$ only in the edges of
the subgraph induced by $L\cup R$, we think of $G_1$ and graphs in $\mG_2$ as consisting only
of this subgraph. Finally, since in our constructions all the vertices in $L\cup R$ have the same degree of $\sqrt m$, we assume that no degree queries are performed.}

For every, $Q$, every $t \leq Q$ and every query-answer history $\pi$ of length $t-1$ the process $P_1$ answers the $t\th$ query of the algorithm \dchange{consistently with $G_1$. Namely:}
\begin{itemize}
\item For a pair query $q_t=(u,v)$
% the process answers according to the single graph $G_1$ in $\mG_1$. That is,
if the pair $(u,v)$ is a crossing pair \dchange{in $G_1$}, then the process replies $1$, and otherwise it replies $0$.
\item For a random new-neighbor query $q_t=u$ the process answers with a random neighbor of $u$ that has yet been observed by the algorithm. That is, for every vertex $v$ such that $v\in\Gamma(u)\setminus\Gamma^{kn}_{\pi}(u)$ the process replies $a_t=v$ with probability $1/(\sqrt m - d^{kn}_{\pi}(u))$.
\end{itemize}

The process $P_2$ is defined as follows:
\begin{itemize}
\item For a query-answer history $\pi$ we denote by $\mathcal{G}_2(\pi)\subset \mG_2$ the subset of graphs in $\mG_2$ that are consistent with $\pi$.

\item For every $t\leq Q$ and every query-answer history $\pi$ of length $t-1$, the process $P_2$ selects a graph in $\mG_2$ uniformly at random and answers the $t\th$ query as follows.
\begin{enumerate}
\item If the $t\th$ query is a pair query $q_t = (u,v)$, then $P_2$ answers the query $q_t$ according to the selected graph.
\item  If the $t\th$ query is a random new-neighbor query $q_t = u_t$, then $P_2$'s answer is a uniform new neighbor of $u_t$ in the selected graph.
\end{enumerate}

\item After all queries are answered (i.e., after $Q$ queries), uniformly choose a random graph $G$ from $\mG_2(\pi)$.
\end{itemize}

For a query-answer history $\pi$ of length $Q$ we denote by $\pi^{\leq t}$ the length $t$ prefix of $\pi$ and by $\pi^{\geq t}$ the $Q-t+1$ suffix of $\pi$.

We note that the selected graph is only used to answer the $t\th$ query and is then
``discarded back to'' the remaining graphs that are consistent with that answer (and all previous answers in $\pi$).
%\mnote{say something about G only being used to answer the query and then it's discarded.}}

\begin{claim}
	Let $\pi$ be a query-answer history of length $t-1$. We use $\circ$ to denote concatenation. 	
	\begin{itemize}
	\item If the $t\th$ query is a pair query, then $a_t=1$ with probability $$\frac{|\mG_2(\pi\circ(q_t,1))|}{|\mG_2(\pi)|},$$
		  and $a_t = 0$ with probability $$\frac{|\mG_2(\pi\circ(q_t,0))|}{|\mG_2(\pi)|}.$$

	\item If the $t\th$ query is a random new-neighbor query $q_t=u_t$, then 	for every $v\in V \setminus \Gamma^{kn}_{\pi}(u)$ the probability that the process $P_2$ answers $a_t=v$ is
		\[\frac{|\mG_2(\pi\circ(q_t,v))|}{|\mG_2(\pi)|} \cdot \frac{1}{ \sqrt m - d^{kn}_\pi(u_t)}.\]
		 If $v \in \Gamma^{kn}_{\pi}(u)$ then the probability that $P_2$ answers $a_t=v$ is $0$.
	\end{itemize}
\end{claim}
\begin{proof}
	First consider a pair query $q_t=(u_t,v_t)$. The probability that $(u_t,v_t)$ is an edge in the graph chosen by the process $P_2$ is the fraction of graphs in $\mG_2(\pi)$ in which $(u_t,v_t)$ is an edge. This is exactly	 $\frac{|\mG_2(\pi\circ(q_t,1))|}{|\mG_2(\pi)|}.$ Similarly, the probability of choosing a graph in which $(u_t,v_t)$ is not an edge is $\frac{|\mG_2(\pi\circ(q_t,0))|}{|\mG_2(\pi)|}$.

	Now consider a random new-neighbor query $q_t=u_t$. We start with the case that $v \in V \setminus \Gamma^{kn}_{\pi}$. The probability that $v$ is chosen by $P_2$ is the probability that a graph $G$ in which $v$ is a neighbor of $u_t$ is chosen in the first step, and that $v$ is the chosen new neighbor among all of $u$'s neighbors in the second step. Since there are $|\mG_2(\pi \circ (u,v) )|$ graphs in which $v$ is a neighbor of $u_t$, and $u_t$ has $\sqrt m - d^{kn}_\pi(u_t)$ neighbors,
	this happens with probability
	$$ \frac{|\mG_2(\pi \circ (u,v) )|}{|\mG_2|}\cdot \frac{1}{ \sqrt m - d^{kn}_\pi(u_t)}.$$
	For a vertex $v$ such that $v \notin V \setminus \Gamma^{kn}_{\pi}$, in every graph $G\in\mG_2$ , $v$ is not a neighbor of $u_t$, implying that the probability that the process replies $a_t=v$ is $0$.
\end{proof}

\begin{lemma} \label{uniform graph}
For every algorithm ALG, the process $P_2$, when interacting with ALG,
% generates a uniform graphs in $\mathcal{G}_2$.
answers ALG's queries according to a uniformly generated graph $G$ in $\mG_2$.
\end{lemma}

\begin{proof}
	Consider a specific graph $G \in \mG_2$. Let $\pi$ be the query-answer history generated  by the interaction between ALG and $P_2$. Let $Q$ be the number of queries performed during the interaction. The probability that $G$ is the resulting graph from that interaction is
	\begin{align*}
		\Pr[G \in \mG_2(\pi^{\leq 1})] &\cdot  \Pr[G \in \mG_2(\pi^{\leq 2}) \mid G \in \mG_2(\pi^{\leq 1})] \cdot \; \dots \; \cdot \Pr[G \in \mG_2(\pi^{\leq Q})| G \in \mG_2(\pi^{\leq Q-1})] \cdot \frac{1}{|\mG(\pi^{\leq Q})|} \\
		&= \frac{|\mG_2(\pi^{\leq 1})|}{|\mG_2|} \cdot \frac{|\mG_2(\pi^{\leq 2})|}{|\mG_2(\pi^{\leq 1})|} \cdot \; \dots\; \cdot \frac{|\mG_2(\pi^{\leq Q})|}{|\mG_2(\pi^{\leq Q-1 })|} \cdot \frac{1}{|\mG_2(\pi^{Q})|} = \frac{1}{|\mG_2|} ,
	\end{align*}
	and the lemma follows.
\end{proof}

For a fixed algorithm ALG that performs $Q$ queries, and for $b\in\{1,2\}$, let $\mathcal{D}_{\text{ALG}}^{b}$ denote the distribution on query-answers histories of length $Q$ induced by the interaction between ALG and $P_b$. We shall show that for every algorithm ALG that performs at most $Q=\frac{m^{3/2}}{100\tr}$ queries, the statistical distance between $\mathcal{D}^{\text{ALG}}_{1}$ and $\mathcal{D}^{\text{ALG}}_{2}$, denoted $d\left(\mathcal{D}^{\text{ALG}}_{1}, \mathcal{D}^{\text{ALG}}_{2}\right)$, is at most $\frac{1}{3}$. This will imply that the lower bound stated in Theorem~\ref{m 3/2 over delta} holds for the case that $\tr(G)=m$. In order to obtain this bound we introduce the
notion of a query-answer witness pair, defined next.

\begin{definition}\label{def:qa-witness-pair}
We say that ALG has detected a \textsf{query-answer witness pair} in three cases:
\begin{enumerate}
 \item If $q_t$ is a pair query for a crossing pair $(u_t,v_t) \in L\times R$ and $a_t=0$.
 \item If $q_t$ is a pair query for a non-crossing pair $(u_t,v_t) \in (L\times L) \cup (R\times R)$ and $a_t=1$.
 \item If $q_t=u_t$ is a random new-neighbor query and $a_t=v$ for some $v$ such that $(u_t,v)$ is a non-crossing pair.
\end{enumerate}
\end{definition}

  We note that the source of the difference between $\mathcal{D}^{\text{ALG}}_{1}$ and $\mathcal{D}^{\text{ALG}}_{2}$
 is not only due to the probability that the query-answer history contains a witness pair (which is $0$
 under $\mathcal{D}^{\text{ALG}}_{1}$ and non-$0$ under $\mathcal{D}^{\text{ALG}}_{2}$). There is also
 a difference in the distribution over answers to random new neighbor queries  when the answers do not
 result in  witness pairs (in particular when we condition on the query-answer history prior to
 the $t\th$ query). However, the analysis of witness pairs serves us also in bounding the
 contribution to the distance due to random new neighbor queries that do not result in a witness pairs.

Let $w$ be a ``witness function'', such that for a pair query $q_t$ on a crossing pair, $w(q_t)=0$, and for a non-crossing pair, $w(q_t)=1$. The probability that ALG detects a witness pair when $q_t$ is a pair query $(u_t,v_t)$ and $\pi$ is a query-answer history of length \dchange{$t-1$}, is
\[ \Pr_{P_2}[w(q_t)\,\ddchange{|\,\pi}] = \frac{\left|\mathcal{G}_2\left(\pi\circ\left(q_t,w(q_t)\right)\right)\right|}{\left|\mathcal{G}_2\left(\pi\right)\right|}\le \frac{\left|\mathcal{G}_2\left(\pi\circ(q_t,w(q_t))\right)\right|}{|\mathcal{G}_2(\pi\circ(q_t,\overline{w(q_t)}))|} \ .\]

Therefore, to bound the probability that the algorithm observes a witness pair it is sufficient to bound the ratio between the number of graphs in $\mathcal{G}_2\left(\pi\circ\left(q,w(q_t) \right)\right)$ and the number of graphs in $\mathcal{G}_2(\pi\circ(q,\overline{w(q_t)}\,))$. We do this by introducing an auxiliary graph, which is defined next. \par

\ifnum\focs=0
\subsubsect{The auxiliary graph for $\tr=m$}
\else
\subsect{The auxiliary graph}
\fi
For every $t\leq Q$, every query-answer history $\pi$ of length $t-1$ for which $\pi$ is consistent with $G_1$ (that is, no witness pair has yet been detected), and every pair $(u,v)$, we consider a bipartite \textsf{auxiliary graph} $\aux$. On one side of $\aux$ we have a node for every graph in $\mathcal{G}_{2}(\pi)$ for which the pair \dchange{$(u,v)$} is a witness pair. We refer to these nodes as \textsf{witness graphs}. On the other side of the auxiliary graph, we place a node for every graph in $\mathcal{G}_{2}(\pi)$ for which the pair is not a witness. We refer to these nodes as \textsf{non-witness graphs}. We %add
\dchange{put an edge in the auxiliary graph} between a witness graph $W$ and a non-witness graph $\oW$ if the pair $(u,v)$ is a crossing (non-crossing) pair and the two graphs are identical except that their red (blue) matchings differ on exactly two pairs -- $(u,v)$ and one additional pair.
In other words, $\oW$ can be obtained from $W$ by performing a {\em switch\/}
operation, as defined next.
\begin{definition}
 \label{switch matched} We define a \textsf{switch between pairs in a matching} in the following manner. Let $(u,v)$ and $(u',v')$ be two matched pairs in a matching $M$. A \textsf{switch} between $(u,v)$ and $(u',v')$ means removing the edges $(u,v)$ and $(u',v')$ from $M$ and adding to it the edges $(u,v')$ and $(u',v)$.
\end{definition}
Note that the switch process maintains the cardinality of the matching.
We denote by $d_w(\aux)$ the \textsf{minimal} degree of any witness graph in $\aux$, and by $d_{nw}(\aux)$ the \textsf{maximal} degree of the non-witness graphs. See Figure~\ref{fig:auxiliary} for an illustration.

%
%\begin{figure}[ht!]
%	\[
%	\begin{array}{cc}
%	\begin{tikzpicture}[thick, x=2cm, scale=0.7] \DrawBadAndGoodBipartite	\label{fig:bipartiteGoodBad} \end{tikzpicture} &
%	\begin{tikzpicture}[thick, x=2cm, scale=0.7,xscale=0.5] %\path[use as bounding box] (0,-2.5)--(3,0);
%		\DrawFix \label{fig:fix_k_general} \end{tikzpicture}\\
%	\parbox{0.5\textwidth}{(a) The auxiliary graph with witness nodes on the left and non-witness nodes on the right.} & 
%	\parbox{0.5\textwidth}{(b) An illustration of two neighbors in the auxiliary graph for $\tr=m$.}
%	\end{array}
%	\]
%	\caption{}
%	\label{fig:auxiliary}
%\end{figure}

\begin{figure}[ht!]
	\ifnum\nofigures=0
	\centering
\subfloat[
%\parbox{0.45\textwidth}
{The auxiliary graph with witness nodes on the left and non-witness nodes on the right.}]{%
			\label{fig:bipartiteGoodBad}
			\begin{tikzpicture}[thick, x=2cm, scale=0.7]
			\DrawBadAndGoodBipartite		
u			\end{tikzpicture}}\hfill
\subfloat[
%\parbox{0.45\textwidth}
{An illustration of two neighbors in the auxiliary graph for $\tr=m$.}]{%
			\begin{tikzpicture}[thick, x=2cm, scale=0.7,xscale=0.5]
%			\path[use as bounding box] (0,-2.5)--(3,0);
			\DrawFix
			\end{tikzpicture}\hfill
			\label{fig:fix_k_general}
		}
		\fi
		\caption{}
		\label{fig:auxiliary}
	\end{figure}

\begin{lemma} \label{Auxilary-graph-degrees k= sqrt m}
Let $\tr=m$ and $Q = \frac{m^{3/2}}{100\tr}$. For every $t\leq Q$, every query-answer history $\pi$ of length $t-1$ such that $\pi$ is consistent with $G_1$ and every pair $(u,v)$,
\[\frac{d_{nw}(\aux)}{d_{w}(\aux)} \leq \frac{2}{\sqrt m}= \frac{2\tr}{m^{3/2}}.\]
\end{lemma}

%\medskip

%\begin{proof}OF {Lemma~\ref{Auxilary-graph-degrees k= sqrt m}}
\begin{proof}
Recall that the graphs in $\mG_2$ are as defined in \dchange{Subsection}~\ref{Construction for k=sqrt m} and illustrated in Figure~\ref{fig:bipartite}. In the following we consider crossing pairs, as the proof for non-crossing pairs is almost identical. \dchange{Recall that a} crossing pair is a pair $(u,v)$ such that $u\in L$ and $v\in R$ or vise versa. A witness graph $W$ with respect to the pair $(u,v)$ is a graph in which $(u,v)$ is a red pair, i.e., $(u,v) \in M^C$. There is an edge from $W$ to every non-witness graph $\oW \in \mG_{2}(\pi)$ such that $M^C(W)$ and $M^C(\oW)$ differ exactly on $(u,v)$ and one additional edge. \par

Every red pair 	$(u',v') \in M^C(W)$ creates a potential non-witness graph $\oW_{(u',v')}$ when switched with $(u,v)$ (as defined in Definition~\ref{switch matched}). However, not all of the these non-witness graphs are in $\mG_{2}(\pi)$. If $u'$ is a neighbor of $v$ in the knowledge graph $G^{kn}_\pi$, i.e., $u' \in \Gamma^{kn}_\pi(v)$, then $\oW_{(u',v')}$ is not consistent with the knowledge graph, and therefore $\oW_{(u',v')}\notin \mG_{2}(\pi)$. This is also the case for a pair $(u',v')$ such that $v'\in\Gamma^{kn}_\pi(u)$. Therefore, only pairs $(u',v')\in M^C$ such that $u'\notin \Gamma^{kn}_\pi(v)$ and $v'\notin \Gamma^{kn}_\pi(u)$ produce a non-witness graph $\oW_{(u',v')}\in \mG_{2}(\pi)$ when switched with $(u,v)$. We refer to these pairs as \textsf{consistent pairs}. Since $t \leq \frac{\sqrt m}{100}$, both $u$ and $v$ each have at most $\frac{m}{100}$ neighbors in the knowledge graph, implying that out of the $\sqrt m -1 $ potential pairs, the number of consistent pairs is at least
$$\sqrt m - 1 - d^{kn}_\pi(u) - d^{kn}_\pi(v) \geq \sqrt m -1 -2\cdot \frac{\sqrt m}{100} \geq \frac{1}{2}\sqrt m .$$
\sloppy
Therefore, the degree of every witness graph $W\in \mA_{\pi,(u,v)}$ is at least $\frac{1}{2}\sqrt m$, implying
that $d_{w}(\aux) \geq \frac{1}{2}\sqrt m$. \par
In order to prove that $d_{nw}(\aux) = 1$, consider a non-witness graph $\oW$. Since $\oW$ is a non-witness graph, the pair $(u,v)$ is not a red pair. This implies that $u$ is matched to some vertex $v' \in R$, and $v$ is matched to some vertex $u' \in L$. That is, $(u,v'),(v,u') \in M^C$. By the construction of the edges in the auxiliary graph, every neighbor $W$ of $\oW$ can be obtained by a single switch between two red pairs in the red matching. The only possibility to switch two pairs in $M^C(\oW)$ and obtain a matching in which $(u,v)$ \textsf{is\/} a red pair is to switch the pairs $(u,v')$ and $(v,u')$. Hence, every non-witness graph $\oW$ has at most one neighbor. \par
We showed that $d_{w}(\aux) \geq \frac{1}{2}\sqrt m$ and that $d_{nw}(\aux) \leq 1$, implying
\[\frac{d_{nw}(\aux)}{d_{w}(\aux)} \leq \frac{2}{\sqrt m}= \frac{2\tr}{m^{3/2}},\] and the proof is complete.
\end{proof}

\ifnum\focs=0
\subsubsect{Statistical distance} \label{Statistical distance}
\else
\subsect{Statistical distance} \label{Statistical distance}
\fi

For a query-answer history $\pi$ of length $t-1$ and a query $q_t$, let $Ans(\pi, q_t)$ denote the set of possible answers to the query $q_t$ that are consistent with $\pi$. Namely, if $q_t$ is a pair query (for a pair that does not
belong to the knowledge graph $G^{kn}_\pi$), then
$Ans(\pi, q_t) = \{0,1\}$, and if $q_t$ is a random new-neighbor query, then
$Ans(\pi, q_t)$ consists of all vertices except those in $N^{kn}_\pi$.
\ifnum\focs=1
\tchange{For the proofs of the next lemma we refer the reader to the full version of this paper~\cite{ELRS}.}
\fi
\begin{lemma} \label{Probability_hitting_red_blue_low}
\dchange{Let $\tr=m$ and $Q = \frac{m^{3/2}}{100\tr}$.}
For every $t\leq Q$, every query-answer history $\pi$ of length $t-1$ such that $\pi$ is consistent with $G_1$ and for every query $q_t$:
\[\sum\limits_{\tchange{a} \in Ans(\pi,q_t)} \Big|\Pr_{\tchange{P_1}}[\tchange{a} \,|\, \pi,q_t] - \Pr_{\tchange{P_2}}[\tchange{a} \,|\, \pi,q_t ]\Big| \leq \frac{\dchange{12}}{\sqrt m} = \tchange{\frac{12\tr}{m^{3/2}}} .\]
\end{lemma}
\ifnum \focs = 0
\begin{proof}
We prove the lemma separately for each type of query.
\begin{itemize} \item We start with a crossing pair query $(u_t,v_t)$. In this case the witnesses are red pairs. Namely, our witness graphs for this case are all the graphs in $\mathcal{G}_2(\pi \circ (q_t,0))$, and the non-witness graphs are all the graphs in $\mathcal{G}_2(\pi \circ (q_t,1))$. By the construction of the auxiliary graph
\[\left|\mathcal{G}_2\left(\pi \circ (q_t,0)\right)\right| \cdot d_{w}(\aux) \leq \left|\mathcal{G}_2\left(\pi \circ (q_t,1)\right)\right| \cdot d_{nw}(\aux).\]
This, together with Lemma~\ref{Auxilary-graph-degrees k= sqrt m}, implies
 \[\frac{|\mathcal{G}_2(\pi \circ (q_t,0))|}{|\mathcal{G}_2(\pi)|}\le\frac{|\mathcal{G}_2(\pi\circ (q_t,0))|}{|\mathcal{G}_2(\pi\circ (q_t,1))|} \leq \frac{d_{nw}(\aux)}{d_{w}(\aux)} =\frac{2}{\sqrt m}=\tchange{\frac{2\tr}{m^{3/2}}}\;.\]
For a pair query $q_t$, the set of possible answers $Ans(\pi, q_t)$ is $\{0,1\}$. Therefore,
\begin{align*}
\sum\limits_{a \in \{0,1\}}&\Big|\Pr_{\tchange{P_1}}[a\,|\,\pi,q_t] - \Pr_{\tchange{P_2}}[a\,|\,\pi,q_t] \Big| \\
&=\Big|\Pr_{\tchange{P_1}}[0\,|\,\pi,q_t] - \Pr_{\tchange{P_2}}[0\,|\,\pi,q_t] \Big| + \Big|\Pr_{\tchange{P_1}}[1\,|\,\pi,q_t] - \Pr_{\tchange{P_2}}[1\,|\,\pi,q_t] \Big| \\
&=\tchange{\frac{2\tr}{m^{3/2}}} + 1-\left(1-\tchange{\frac{2\tr}{m^{3/2}}}\right) =\tchange{\frac{4\tr}{m^{3/2}}} = \frac{4}{\sqrt m}\;. \numberthis \label{SD crossing pair}
\end{align*}

\item For a non-crossing pair query $q_t=(u,v)$ our witness graphs are graphs that contain $q_t$ as a blue pair, i.e., graphs from $\mathcal{G}_2(\pi,(q_t,1))$, and our non-witness graphs are graphs in which no blue pair had been queried, i.e., graphs from $\mathcal{G}_2(\pi,(q_t,0))$. From Lemma~\ref{Auxilary-graph-degrees k= sqrt m} we get that for a non-crossing pair query $q_t$: \[\frac{\left|\mathcal{G}_2\left(\pi\circ (q_t,1)\right)\right|}{\left|\mathcal{G}_2(\pi)\right|}\le\frac{\left|\mathcal{G}_2\left(\pi\circ (q_t,1)\right)\right|}{\left|\mathcal{G}_2\left(\pi\circ (q_t,0)\right)\right|} \leq \frac{d_{nw}(\aux)}{d_{w}(\aux) } =\tchange{\frac{2\tr}{m^{3/2}}} =\frac{2}{\sqrt m}\;.\]
Therefore,
\begin{align*}
\sum\limits_{a\in\{0,1\}}&\Big|\Pr_{\tchange{P_1}}[a\,|\,\pi,q_t] - \Pr_{\tchange{P_2}}[a\,|\,\pi,q_t] \Big| \\
&= \;\;\Big|\Pr_{\tchange{P_1}}[0\,|\,\pi,q_t] - \Pr_{\tchange{P_2}}[0\,|\,\pi,q_t] \Big| + \Big|\Pr_{\tchange{P_1}}[1\,|\,\pi,q_t] - \Pr_{\tchange{P_2}}[1\,|\,\pi,q_t] \Big| \\
&= \;\; 1-\left(1-\tchange{\frac{2\tr}{m^{3/2}}}\right) + \tchange{\frac{2\tr}{m^{3/2}}} =\tchange{\frac{4\tr}{m^{3/2}}} = \frac{4}{\sqrt m}\;. \numberthis \label{SD non-crossing pair}
\end{align*}

\item For a new-neighbor query $q_t=u_t$, the set of possible answers $Ans(\pi, q_t)$ is the set of all the vertices in the graph. Therefore,
\begin{align*}
&\sum\limits_{a \in Ans(\pi, q_t)} \Big|\Pr_{\tchange{P_1}}[a \,|\,\pi,q_t] - \Pr_{\tchange{P_2}}[a \,|\,\pi,q_t] \Big| \\
&\;\;=\sum\limits_{v \in R} \Big|\Pr_{\tchange{P_1}}[v \,|\,\pi,q_t] - \Pr_{\tchange{P_2}}[v \,|\,\pi,q_t] \Big| +
%\;\;\;\;\;\;\;\; +
\sum\limits_{v \in L} \Big|\Pr_{\tchange{P_1}}[v \,|\,\pi,q_t] - \Pr_{\tchange{P_2}}[v \,|\,\pi,q_t] \Big|\;.
\end{align*}
Recall that for a vertex $v \in \Gamma^{kn}_\pi(u)$, $\Pr_{\tchange{P_1}}[v \,|\,\pi,q_t] =\Pr_{\tchange{P_2}}[v \,|\,\pi,q_t] =0$. Therefore, \dchange{it suffices to} consider only vertices $v$ such that $v\notin \Gamma^{kn}_\pi(u)$. Assume without loss of generality that $u\in L$, and consider a vertex $v\in R, v\notin \Gamma^{kn}_\pi(u)$.
% For every $G$ in $\mathcal{G}_{1}(\pi)$ it holds that
\dchange{Since for every $v\in R$ we have that $(u_t,v) \in E(G_1)$, by the definition of $P_1$,}
\begin{align}
&\Pr_{\tchange{P_1}}[v \,|\,\pi,q_t] = \frac{1}{\sqrt m - d^{kn}_\pi(u_t)} \label{D1_R}\;.
\end{align}
Now consider the process $P_2$. By \dchange{its} definition,
\begin{align*}
\Pr_{\tchange{P_2}}[v \,|\,\pi,q_t] &= \frac{\mathcal{G}_2\left(\pi\circ (q_t,v)\right) }{\mathcal{G}_2(\pi)} \cdot \frac{1}{\sqrt m -d^{kn}_\pi(u)} \\
&= \frac{\mathcal{G}_2\left(\pi\circ \left((u,v),1\right)\right) }{\mathcal{G}_2(\pi)} \cdot \frac{1}{\sqrt m -d^{kn}_\pi(u)} \\
&=\left(1 - \frac{\mathcal{G}_2\left(\pi\circ \left((u,v),0\right) \right)}{\mathcal{G}_2(\pi)}\right) \cdot \frac{1}{\sqrt m -d^{kn}_\pi(u)}\;.
\end{align*}
By the first item in the proof, for any crossing pair $q_t=(u,v)$,
\begin{align*}
\frac{\mathcal{G}_2(\pi \circ (q_t,0))}{\mathcal{G}_2(\pi)} =\tchange{\frac{4\tr}{m^{3/2}}} = \frac{4}{\sqrt m} \;,
\end{align*}
and it follows that
\begin{align}
&\Pr_{\tchange{P_2}}[v \,|\,\pi,q_t] =\left(1 - \tchange{\frac{4\tr}{m^{3/2}}} \right) \cdot \frac{1}{\sqrt m -d^{kn}_\pi(u)} \label{D2_R}\;.
\end{align}
By Equations~\eqref{D1_R} and~\eqref{D2_R}, we get that for every $v \in R$ such that $v \notin \Gamma^{kn}_\pi(u)$,
\begin{align}
\Big|\Pr_{\tchange{P_1}}[v \,|\,\pi,q_t] - \Pr_{\tchange{P_2}}[v \,|\,\pi,q_t]\Big| = \frac{\tchange{4\tr/m^{3/2}}}{\sqrt m - d^{kn}_\pi(u)}. \numberthis \label{random new neighbor query single a}
\end{align}
Therefore,
\begin{align*}
\sum\limits_{v \in R } & \Big|\Pr_{\tchange{P_1}}[v \,|\,\pi,q_t] - \Pr_{\tchange{P_2}}[v \,|\,\pi,q_t] \Big|
 =\sum\limits_{v \in R, v\notin \Gamma^{kn}_\pi(u)} \Big|\Pr_{\tchange{P_1}}[v \,|\,\pi,q_t] - \Pr_{\tchange{P_2}}[v \,|\,\pi,q_t] \Big|\\
& =\left(\sqrt m - d^{kn}_\pi(u)\right) \cdot \frac{\tchange{4\tr/m^{3/2}}}{\sqrt m - d^{kn}_\pi(u)}
 =\tchange{\frac{4\tr}{m^{3/2}}} = \frac{4}{\sqrt m} \;.\numberthis \label{SD_R}
\end{align*}

Now consider a vertex $v \in L$. Observe that for every $v\in L$, it holds that $v \notin \Gamma^{kn}_\pi(u) $ since otherwise $\pi$ is not consistent with $G_1$.
% For every $G$ in $\mathcal{G}_{1}(\pi)$, $(u_t,v) \notin E(G)$, implying
\dchange{For the same reason,}
\begin{align}
&\Pr_{\tchange{P_1}}[v \,|\,\pi,q_t] = 0 \label{D1_L}\;.
\end{align}
As for $P_2$, as before,
\begin{align*}
&\Pr_{\tchange{P_2}}[v \,|\,\pi,q_t] = \frac{\mathcal{G}_2(\pi, (u_t,v)) }{\mathcal{G}_2(\pi)} \cdot \frac{1}{\sqrt m - d^{kn}_\pi(u_t)}\;.
\end{align*}	
By the second item of the claim, since for every $v \in L$, $(u_t,v)$ is a non-crossing pair, we have that
\begin{align}
\frac{|\mathcal{G}_2(\pi,(u_t,v))|}{|\mathcal{G}_2(\pi)|} = \tchange{\frac{4\tr}{m^{3/2}}} = \frac{4}{\sqrt m} \label{D2_L}\;.
\end{align}
Combining Equations \eqref{D1_L} and \eqref{D2_L} we get that for every $v \in L$
\begin{align*}
\Big|\Pr_{\tchange{P_1}}[v \,|\,\pi,q_t] - \Pr_{\tchange{P_2}}[v \,|\,\pi,q_t]\Big| = \frac{ \tchange{4\tr/m^{3/2}} }{\sqrt m - d^{kn}_\pi(u)}.
\end{align*}
Since $Q=\frac{m^{3/2}}{100\tr} = \frac{\sqrt m}{100}$, for every $t\leq Q$, $d^{kn}_\pi(u) < \frac{1}{2}\sqrt m$, and it follows that $\frac{\sqrt m-1}{\sqrt m-d^{kn}(u)}$ is bounded by $2$. Hence,
\begin{align*}
\sum\limits_{v \in L} \Big|\Pr_{\tchange{P_1}}[v \,|\,\pi,q_t] - \Pr_{\tchange{P_2}}[v \,|\,\pi,q_t] \Big| &= (\sqrt m - 1) \cdot \frac{\tchange{4\tr/m^{3/2}}}{\sqrt m - d^{kn}_\pi(u)} \\
&= \tchange{\frac{8\tr}{m^{3/2}}} = \frac{8}{\sqrt m} \;.\numberthis \label{SD_L}
\end{align*}
By Equations \eqref{SD_R} and \eqref{SD_L} we get
\begin{align*}
\sum\limits_{v \in R} &\Big|\Pr_{\tchange{P_1}}[v \,|\,\pi,q_t] - \Pr_{\tchange{P_2}}[v \,|\,\pi,q_t] \Big| + \sum\limits_{v \in L} \Big|\Pr_{\tchange{P_1}}[v \,|\,\pi,q_t] - \Pr_{\tchange{P_2}}[v \,|\,\pi,q_t] \Big| \\
& = \tchange{\frac{12\tr}{m^{3/2}}} = \frac{12}{\sqrt m} . \numberthis \label{dist neighbors}
\end{align*}
\end{itemize}
This completes the proof.
\end{proof}
\fi
Recall that $\mathcal{D}_{b}^{\text{ALG}}$, $b\in\{1,2\}$, denotes the distribution on query-answer histories of length $Q$, induced by the interaction of ALG and $P_b$.
We show that the two distributions are indistinguishable for $Q$ that is sufficiently small.

\begin{lemma} \label{dist for k=sqrt m}
Let $\tr=m$. For every algorithm ALG that asks at most $Q=\tchange{\frac{m^{3/2}}{100\tr}}$ queries, the statistical distance between $\mathcal{D}_{1}^{\text{ALG}}$ and $\mathcal{D}_{2}^{\text{ALG}}$ is at most $\frac{1}{3}$.
\end{lemma}
%\ifnum \focs =0
\begin{proof}
% D: FOR FUTURE: Distinguish between random variables and the values that they get
Consider the following hybrid distribution. Let $\mathcal{D}_{1,t}^{\text{ALG}}$ be the distribution over query-answer histories of length $Q$, where in the length $t$ prefix ALG is answered by the process $P_1$ and in the length $Q-t$ suffix ALG is answered by the process $P_2$. Observe that $\mathcal{D}_{1,Q}^\text{ALG} = \mathcal{D}_{1}^\text{ALG}$ and that $\mathcal{D}_{1,0}^\text{ALG} = \mathcal{D}_{2}^\text{ALG}$. Let $\pi = (\pi_1, \pi_2,\ldots,\pi_\ell)$ denote a query-answer history of length $\ell$. By the triangle inequality \ifnum\focs=0 $d(\mathcal{D}_{1}^{\text{ALG}}, \mathcal{D}_{2}^{\text{ALG}}) \leq \sum\limits_{t=0}^{Q-1}d(\mathcal{D}_{1,t+1}^{\text{ALG}}, \mathcal{D}_{1,t}^{\text{ALG}})$ \;. \fi

\ifnum\focs=0
\begin{align*}
d(\mathcal{D}_{1}^{\text{ALG}}, \mathcal{D}_{2}^{\text{ALG}}) \leq \sum\limits_{t=0}^{Q-1}d(\mathcal{D}_{1,t+1}^{\text{ALG}}, \mathcal{D}_{1,t}^{\text{ALG}}) \;.
\end{align*}
\fi

It thus remains to bound
$d(\mathcal{D}_{1,t+1}^{\text{ALG}}, \mathcal{D}_{1,t}^{\text{ALG}})$
\dchange{
$= \frac{1}{2}\sum\limits_{\pi }\Big|\Pr_{\mathcal{D}_{1,t+1}^{\text{ALG}}}[\pi]-\Pr_{\mathcal{D}_{1,t}^{\text{ALG}}}[\pi]\Big| $} for every $t$ such that $0\leq t\leq Q-1$. Let $\mathcal{Q}$ denote the set of all possible queries.

\begin{align*}
\sum\limits_{\pi }\Big|\Pr_{\mathcal{D}_{1,t+1}^{\text{ALG}}}[\pi]-\Pr_{\mathcal{D}_{1,t}^{\text{ALG}}}[\pi]\Big|  &= \makebox[0.42cm]{\hfill} \sum\limits_{\pi_1,\ldots,\pi_{t-1}} \Pr_{\tchange{P_1},\text{ALG}}[\pi_1,\ldots,\pi_{t-1}]\cdot \sum\limits_{q \in \mathcal{Q}} \Pr_\text{ALG}[q\,|\,\pi_1,\ldots,\pi_{t-1} ] \\
&\cdot \sum\limits_{a \in Ans((\pi_1,\ldots,\pi_{t-1}),q)} \Big|\Pr_{\tchange{P_1}}[a\,|\,\pi_1,\ldots,\pi_{t-1},q] - \Pr_{\tchange{P_2}}[a\,|\,\pi_1,\ldots,\pi_{t-1},q]\Big| \\
& \cdot \makebox[0.75cm]{\hfill}  \sum\limits_{\pi_{t+1},\ldots,\pi_{Q}} \Pr_{\tchange{P_2},\text{ALG}}[\pi_{t+1},\ldots,\pi_{Q}\,|\,\pi_1,\ldots,\pi_{t-1},\dchange{(q,a)} ] \;.\\
\end{align*}

By Lemma~\ref{Probability_hitting_red_blue_low}, \dchange{for every $1 \leq t \leq Q-1$,
and every $\pi_1,\dots,\pi_{t-1}$ and $q$},
\[\sum\limits_{a \in Ans((\pi_1,\ldots,\pi_{t-1}),q)} \Big|\Pr_{\tchange{P_1}}[a\,|\,\pi_1,\ldots,\pi_{t-1},q] - \Pr_{\tchange{P_2}}[a\,|\,\pi_1,\ldots,\pi_{t-1},q]\Big| \leq \tchange{\frac{12\tr}{m^{3/2}}} \;.\]
\dchange{We also have that for every pair $(q,a)$,
\[\sum\limits_{\pi_{t+1},\ldots,\pi_{Q}} \Pr_{\tchange{P_2},\text{ALG}}[\pi_{t+1},\ldots,\pi_{Q}\,|\,\pi_1,\ldots,\pi_{t-1},\dchange{(q,a)} ]
  = 1\;.\]}
Therefore,
\begin{align*}
% &d(\mathcal{D}_{1}^{\text{ALG}}, \mathcal{D}_{2}^{\text{ALG}}) \\
 \sum\limits_{\pi }&\Big|\Pr_{\mathcal{D}_{1,t+1}^{\text{ALG}}}[\pi]-\Pr_{\mathcal{D}_{1,t}^{\text{ALG}}}[\pi]\Big| 
\leq \sum\limits_{\pi_1,\ldots,\pi_{t-1}} \Pr_{\tchange{P_1},{\text{ALG}}}[\pi_1,\ldots,\pi_{t-1}]\sum\limits_{q \in \dchange{\mathcal{Q}}} \Pr_\text{ALG}[q\,|\,\pi_1,\ldots,\pi_{t-1} ] \cdot \tchange{\frac{12\tr}{m^{3/2}}}
% \cdot \sum\limits_{\pi_{t+1},\ldots,\pi_{Q}} % \Pr_{D_2,\text{ALG}}[\pi_{t+1},\ldots,\pi_{Q}\,|\,\pi_1,\ldots,\pi_{t-1},(q,a) ]
=\tchange{\frac{12\tr}{m^{3/2}}}\; .
\end{align*}
Hence, for $Q=\frac{\sqrt m}{100}$,
\begin{align*}
d(\mathcal{D}_{1}^{\text{ALG}}, \mathcal{D}_{2}^{\text{ALG}}) =& \frac{1}{2}\sum\limits_{\pi }\sum\limits_{t=1}^{Q-1}\Big|\Pr_{\mathcal{D}_{1,t+1}^{\text{ALG}}}[\pi]-\Pr_{\mathcal{D}_{1,t}^{\text{ALG}}}[\pi]\Big| \leq \frac{1}{2}\cdot Q \cdot \tchange{\frac{12\tr}{m^{3/2}}} \leq \frac{1}{3},
\end{align*}
and the proof is complete.
\end{proof}
%\fi
\ifnum\focs=1
\ddchange{The case of $\tr =m$ in Theorem~\ref{m 3/2 over delta} follows from
Lemma~\ref{dist for k=sqrt m}.}
\fi
%%%%%%%%%%%%%%%%%%%%%%%%%%%%%%%%%%%%%%%%%%%%%%%
\ifnum\focs=0
In the next subsection we turn to prove the theorem for the cases where $m <\tr\leq \frac{m^{3/2}}{8}$, and for the case where $\sqrt m \leq \tr \leq \frac{m}{4}$. We start with the former case. The proof will follow the building blocks of the proof for
 % $\tr=\sqrt m$,
 $\tr = m$,
where the only difference is in the description of the auxiliary graph $\aux$ and in the proof that $\frac{d_{nw}(\aux)}{d_w(\aux)} \leq \frac{2\tr}{m^{3/2}} = \frac{2r}{\sqrt m}$.

\subsect{A lower bound for $m < \tr <m^{3/2}$ } \label{section r sqrt m}
\tchange{Let $\tr=r \cdot m$ for \ddchange{an integer} $r$ such that $1< r \leq \frac{1}{8}\sqrt m$. It is sufficient for our needs to consider only values of $\tr$ for which $r$ is an integer}. \dchange{The proof of the lower bound for this case is a fairly simple extension of the proof for the case of \tchange{$\tr=m$, that is, $r= 1$}. We next describe the modifications we make in the construction of $\mG_2$. %the second family.
}

\dchange{\subsubsect{The lower-bound construction}}
Let $G_1$ be as defined in \dchange{Subsection}~\ref{Construction for k=sqrt m}. The construction of $\mG_2$ for $\tr=r \cdot m$ \dchange{can} be thought of as repeating the construction of $\mG_2$ for $\tr=m$ (as described in \dchange{Subsection}~\ref{Construction for k=sqrt m}) $r$ times. We again start with a complete bipartite graph $K_{\sqrt{m},\sqrt{m}}$ and an independent set of size $n-2\sqrt m$.
% We then choose
\dchange{For each graph $G \in \mG_2$ we select $r$} perfect matchings between the two sides $R$ and $L$ and remove these edges from the graph.
%, so that these choices will be the ``non-edges" or ``red pairs'' of the new family of graphs.
We denote the $r$ perfect matchings \dchange{by} $M^C_1, \dots, M^C_r$ and refer to them as the \textsf{red matchings}. We require that each two perfect matchings $M^C_i$ and $M^C_j$ do not have any shared edges. That is, for every $i$ and for every $j$, for every $(u,v) \in M^C_i$ it holds that $(u,v) \notin M^C_j$.
% Now, we choose independently of the previous choice, $t$
\dchange{In order to maintain the degrees of the vertices, we next select $r$} perfect matchings for each side of the bipartite graph ($L$ to $L$ and $R$ to $R$). We denote these matchings \dchange{by} $M^R_1, ..., M^R_{r}$ and $M^L_1, ..., M^L_{r}$ respectively. Again we require that no two matchings share an edge. We refer to these matchings as the \textsf{blue matchings} and their edges as \textsf{blue pairs}. Each such choice of \dchange{$3r$} matchings defines a graph in $\mG_2$.

\tchange{Let $G$ be a graph in $\mG_2$. We say that a triangle is blue if all its edges are blue. Otherwise we say the triangle is mixed. Observe that every blue edge in $G$ participates in at least $\sqrt m -2r$ mixed triangles, and at most $r$ blue triangles. Also note that every two mixed triangles are disjoint. Therefore, there are at least $\frac{1}{2}r \sqrt m \cdot(2\sqrt{m} - 2r)=\Omega(r \cdot  m)$ and at most $\frac{1}{2}r\sqrt m \cdot(2\sqrt{m} - 2r)  + r^2 \sqrt m$ triangles in $G$. Since $r < \frac{1}{8}\sqrt m$, we get that every graph in $G$ has $\Theta(r \cdot m)$ triangles.
	}
% \mnote{T: we should state somewhere that all the graphs in $\mG_2$ indeed have $m$ triangles}

\dchange{\subsubsect{The processes $P_1$ and $P_2$}
The definition of the processes $P_1$ and $P_2$ is the same as in Subsection~\ref{processes} (using the modified
definition of $\mG_2$), and Lemma~\ref{uniform graph} holds here as well.}

\subsubsect{The auxiliary graph}
As before, for every $t\leq Q$, every query-answer history $\pi$ of length $t-1$ such that $\pi$ is consistent with $G_1$ and every pair $(u,v)$, we define a bipartite auxiliary graph $\aux$, such that on one side there is a node for every witness graph $W \in \mG_2(\pi)$, and on the other side a node for every non-witness graph $\oW \in \mG_2(\pi)$. The witness graphs for this case are graphs in which $(u,v)$ is a red (blue) edge in one of the red (blue) matchings. If $(u,v)$ is a crossing pair, then for every witness graph $W$, $(u,v) \in M^C_i(W)$ for some $1\leq i\leq r$. If $(u,v)$ is a non-crossing pair, then for every witness graph $W$, $(u,v) \in M^L_i(W)$ or $(u,v) \in M^L_i(W)$. There is an edge from $W$ to every graph $\oW$ such that the matching that contains $(u,v)$ in $W$ and the corresponding matching in $\oW$ differ on exactly two pairs -- $(u,v)$ and one additional pair. For example, if $(u,v) \in M^C_i(W)$, there is an edge from $W$ to every graph $\oW$ such that
$M^C_{\dchange{i}}(W)$ and $M^C_{\dchange{i}}(\oW)$ differ on exactly $(u,v)$ and one additional pair.

\begin{lemma} \label{Auxilary-graph-degrees k=r sqrt m}
\tchange{Let $\tr=r \cdot m$ for an integer $r$ such that $1< r\leq \frac{\sqrt m}{8}$ and let $Q = \frac{m^{3/2}}{100\tr} $.}
For every $t\leq Q$, every query-answer history $\pi$ of length $t-1$ such that $\pi$ is consistent with $G_1$ and every pair $(u,v)$,
\[\frac{d_{nw}(\aux)}{d_{w}(\aux)} \leq \tchange{\frac{2\tr}{m^{3/2}}} = \frac{2r}{\sqrt m}.\]
\end{lemma}

\begin{proof}
We again \dchange{analyze the case in which} the pair is a crossing pair $(u,v)$, as the proof for a non-crossing pair is almost identical. We first consider the minimal degree of the witness graphs in $\aux$.
% Assume that $(u,v)$ is a matched pair in the $i\th$ matching $M^C_i$.
\dchange{Let $M^C_i$ be the matching to which $(u,v)$ belongs.}
As before, only pairs $(u',v') \in M^C_i$ such that $u' \notin \Gamma^{kn}_\pi(u)$, $ v' \notin \Gamma^{kn}_\pi(v)$ result in a non-witness graph $\oW\in \mG_2(\pi)$ when switched with $(u,v)$. However, we have an additional constraint. Since by our construction no two red matchings share an edge, it must be that $u'$ is not matched to $v$ in any of the other $r$ red matching, and similarly that $u$ is not matched to $v'$ in any of the other matchings.
\dchange{It follows that} of the $(\sqrt m -1 -2\cdot \tchange{\frac{m^{3/2}}{100\cdot r \cdot m}})$ potential pairs (as in the proof of Lemma~\ref{Auxilary-graph-degrees k= sqrt m}), we discard $2r$ additional pairs. \tchange{Since $1\leq r \leq \frac{\sqrt m}{8}$} we remain with $(\sqrt m -1 -\frac{\sqrt m}{50} - \frac{1}{4}\sqrt m) \geq \frac{1}{2}\sqrt m$ potential pairs. Thus, $d_{w}(\aux) \geq \frac{1}{2} \sqrt m$.

We now turn to consider the degree of the non-witness graphs and prove that $d_{nw}(\aux) \leq r$. Consider a non-witness graph $\overline{W}$. To prove that $\oW$ has at most $r$ neighbors it is easier to consider all the possible options to ``turn'' $\oW$ from a non-witness graph into a witness graph. It holds that for every $j \in [r]$, $(u,v) \notin M^C_j(\oW)$. Therefore for every matching $M^C_j$, $u$ is matched to some vertex, denoted $v'_j$ and $v$ is matched to some vertex, denoted $u'_j$. If we switch between the pairs $(u,v'_j)$ and $(v,u'_j)$, this results in a matching in which $(u,v)$ is a witness pair. We again refer the reader to Figure~\ref{fig:fix_k_general}, where the illustrated matching can be thought of as the $j\th$ matching. Denote the resulting graph by $W_{(u'_j,v'_j)}$. If the pair $(u'_j,v'_j)$ has not been observed yet by the algorithm then $W_{(u'_j,v'_j)}$ is a witness graph in $\aux$. Therefore there are at most $r$ options to turn $\oW$ into a witness graph, and $d_{nw}(\aux)\leq r$. We showed that $d_{w}(\aux)\geq \frac{1}{2}\sqrt m$ and $d_{nw}(\aux) \leq r$, implying
\[\frac{d_{nw}(\aux)}{d_{w}(\aux)} \leq \frac{2r}{\sqrt m}= \tchange{\frac{2\tr}{m^{3/2}}},\]
as required.
\end{proof}

\subsubsect{Statistical distance}
The proof of the next lemma is exactly the same as the proof of Lemma~\ref{Probability_hitting_red_blue_low}, except that occurrences of the term $(\tr/m^{3/2})$ are \tchange{replaced by $(r/\sqrt m)$ instead of $(1/\sqrt{m})$}, and we apply Lemma~\ref{Auxilary-graph-degrees k=r sqrt m} instead of Lemma~\ref{Auxilary-graph-degrees k= sqrt m}.

\begin{lemma} \label{Probability_hitting_red_blue_low_k=r sqrt m}
\tchange{Let $\tr=r \cdot m$ for an integer $r$ such that $1< r\leq \frac{\sqrt m}{8}$ and let $Q = \frac{m^{3/2}}{100\tr} $}. For every $t\leq Q$, every query-answer history $\pi$ of length $t-1$ such that $\pi$ is consistent with $G_1$ and for every query $q_t$,
\[\sum\limits_{a \in Ans(\pi,q_t)} \Big|\Pr_{\tchange{P_1}}[a \,|\,\pi,q_t] - \Pr_{\tchange{P_2}}[a \,|\,\pi,q_t ]\Big| =  \tchange{\frac{12\tr}{m^{3/2}}} = \frac{12r}{\sqrt m}\;.\]
\end{lemma}
\dchange{The proof of the next lemma is same as the proof of Lemma~\ref{dist for k=sqrt m} except that
we replace the application of Lemma~\ref{Probability_hitting_red_blue_low}, by an application of
Lemma~\ref{Probability_hitting_red_blue_low_k=r sqrt m}.}
\begin{lemma} \label{dist for k = r sqrt m} \tchange{Let $\tr=r \cdot m$ for an integer $r$ such that $1< r\leq \frac{\sqrt m}{8}$. For every algorithm ALG that performs at most $Q=\frac{m^{3/2}}{100\tr}$ }queries, the statistical distance between $\mathcal{D}_{1}^{\text{ALG}}$ and $\mathcal{D}_{2}^{\text{ALG}}$ is at most $\frac{1}{3}$.
\end{lemma}

% \par We now turn to the case that $k \leq \frac{1}{4}\sqrt m$.

\subsect{A lower bound for $\sqrt m \leq \tr \leq \frac{1}{4}m$} \label{k leq sqrt m}
\tchange{Similarly to the previous section, we let $\tr=k\sqrt m$ and assume that $k$ is an integer such that $1 \leq k\leq \frac{\sqrt m}{4}$.}
\subsubsect{The lower-bound construction}
% for the family $\mG_2$}
The construction of the graph $G_1$ is as defined in \dchange{Subsection}~\ref{Construction for k=sqrt m}, and we modify the construction of the graphs in $\mG_2$. %the second family.
As before, the basic structure of every graph is a complete bipartite graph $K_{\sqrt m,\sqrt m}$ and an independent set of size $n-2\sqrt m$ vertices. In this case,
 \dchange{for each graph in $\mG_2$,} we do not remove a perfect matching from the bipartite graph, but rather a matching $M^C$ of size $k$. In order to keep the degrees of all vertices to be $\sqrt m$, we modify the way we construct the blue matchings. Let $M^C=\{(\ell_{i_1},r_{i_1}), (\ell_{i_2},r_{i_2}), \ldots, (\ell_{i_k},r_{i_k})\}$ be the crossing matching. The blue matchings will be $M^L=\{(\ell_{i_1},\ell_{i_2}), (\ell_{i_3},\ell_{i_4}), \ldots, (\ell_{i_k-1},\ell_{i_k})\}$ and $M^R=\{(r_{i_1},r_{i_2}), (r_{i_3},r_{i_4}), \ldots, (r_{i_k-1},r_{i_k})\}$. Note that every matched pair belongs to a \textsf{four-tuple} $\langle\ell_{i_j}, \ell_{i_{j+1}}, r_{i_{j+1}},r_{i_{j}}\rangle$ such that $(\ell_{i_j},r_{i_{j}})$ and $(\ell_{i_{j+1}},r_{i_{j+1}})$ are red pairs and $(\ell_{i_j},\ell_{i_{j+1}})$ and $(r_{i_{j}},r_{i_{j+1}})$ are blue pairs. We refer to these structures as matched \textsf{squares} and to four-tuples $(\ell_x,\ell_y,r_z,r_w)$ such that no pair in the tuple is matched as \textsf{unmatched squares}. See Figure~\ref{fig:bipartite k leq sqrt m} for an illustration. Every graph in % the second family
 $\mG_2$ is defined by its set of $k$ four-tuples.

 \tchange{Similarly to previous constructions, in every graph $G\in \mG_2$, every blue edge participates in $\sqrt m -2$ triangles. Since every triangle in the $G$ contains exactly one blue edge, we have that $G$ has $k\cdot (\sqrt m-2)=\Theta(k\sqrt m)$ triangles.}
% \mnote{T: we should state somewhere that all the graphs in $\mG_2$ indeed have $?$ triangles}

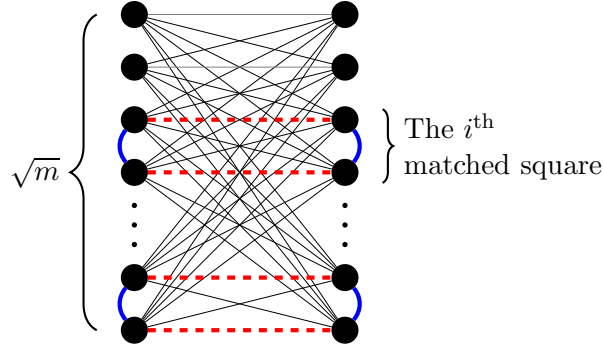
\begin{figure}[!ht]
	\centering
\ifnum\nofigures=0
	\begin{tikzpicture}[thick, x=2cm, scale=0.7]
	\DrawKleqSqrtM
	\end{tikzpicture}
\fi
	\caption{An illustration of the bipartite component in the family $\mathcal{G}_2$ for $\sqrt m \leq \tr \leq \frac{1}{4} m$.}
	\label{fig:bipartite k leq sqrt m}
\end{figure}

\subsubsect{The processes $P_1$ and $P_2$} \label{processes all neighbors query}
We introduce a small modification to the definition of the processes $P_1$ and $P_2$. Namely, we leave the answering process for pair queries as described in Subsection~\ref{processes} and modify the answering process for random new-neighbor queries as follows. Let $t\leq Q$, and $\pi$ be a query-answer history of length $t-1$ such that $\pi$ is consistent with $G_1$. \dchange{If the $t\th$ query is a new-neighbor query $q_t = u$ and $d^{kn}_\pi(u)< \frac{1}{2}\sqrt m$, then
the processes $P_1$ and $P_2$ answer as described in Subsection~\ref{processes}. However,
if the $t\th$ query is a new-neighbor query $q_t=u$ such that $d^{kn}_\pi(u)\geq \frac{1}{2}\sqrt m$, then the processes answers as follows.}
% for every vertex $v \in V$ such that the pair $(u,v)$ have not been seen yet by the algorithm, the process $P_1$ answers $a_t=v$ with probability
% $1/(\sqrt m-d^{kn}_{\pi}(u))$ and the process $P_b$ answers $a_t=v$ with probability
% \[\frac{|\mathcal{G}_b(\pi\circ(q_t,v))|}{|\mathcal{G}_b(\pi)|} \cdot \frac{1}{ \sqrt m - d^{kn}_\pi(u)}.\]
\begin{itemize}
\item The process $P_1$ answers with the set of all neighbors of $u$ in $G_1$. That is, if $u$ is in $L$, then the process replies with $a=R=\{r_1,\ldots,r_{\sqrt m}\}$, and if $u$ is in $R$, then the process replies with $a=L=\{\ell_1,\ldots,\ell_{\sqrt m}\}$.

%%%%%%%%%%%%%
\iffalse
The process $P_2$ answers with $a=\Gamma(u)$, a set of $u$'s neighbors in some graph in $\mG_2$, where the exact set $a$ is chosen as follows. Observe that in $\mG_2$, there are two types of graphs. First, there are graphs in which $u$ is not matched, that is, $(u,u') \notin M^L$ for every vertex $u' \in L$. In these graphs the set of $u$'s neighbors is all the vertices in $R$, i.e., $\Gamma(u)=\{r_1,\ldots,r_{\sqrt m}\}$. The second type of graphs are those in which $(u,u') \in M^L$ for some $u'\in L$ and $(u,v)\in M^C$ for some $v\in R$. In these graphs the set of $u$'s neighbors is $\Gamma(u)=R\setminus\{v\} \cup \{u'\}$. Therefore every possible answer $a \in Ans(\pi,q_t)$ is a set of $u'$ neighbors $\Gamma_{G_2}(u)$ for some graph $G_2 \in \mG_2$. For every $a\in Ans(\pi,q_t)$ the process replies with a specific answer $a$ with probability
\[\frac{|\mathcal{G}_2(\pi\circ(q_t,a))|}{|\mathcal{G}_2(\pi)|} \;.\]
The process $P_2$ answers as follows. It picks a random graph $G_2\in \mG_2$ and answers with all the neighbors of $u$ in that graph, that is, it replies $a=\Gamma_{G_2}(u)$. As before denote by $Ans(\pi,q_t)$ the set of all possible answers. Since a random graph $G_2\in\mG_2$ is picked, it implies that every possible set $a\in Ans(\pi,q_t)$ is chosen with probability
\[\frac{|\mathcal{G}_2(\pi\circ(q_t,a))|}{|\mathcal{G}_2(\pi)|} \;.\]	
\fi
%%%%%%%%%%%%%%
\dchange{The process $P_2$ answers with $a = \{v_1,\dots,v_{\sqrt{m}}\}$, where
$\{v_1,\dots,v_{\sqrt{m}}\}$ is the set of neighbors of $u$ in a subset of the graphs in $\mG_2$.
By the definition of $\mG_2$, if $u$ is in $L$, then this set is either $R$, or
it is $R\setminus \{r_i\}\cup \{\ell_j\}$ for some $r_i \in R$ and $\ell_j \in L$,
and if $u$ is in $R$, then this set is either $L$, or
it is $L\setminus \{\ell_i\}\cup \{r_j\}$ for some $\ell_i \in L$ and $r_j \in R$.
 For every such set $a\in Ans(\pi,q_t)$, the process returns $a$ as an answer with probability
\[\frac{|\mathcal{G}_2(\pi\circ(q_t,a))|}{|\mathcal{G}_2(\pi)|} \;.\]
}
We call this query an \textsf{all-neighbors query}.
\end{itemize}
First note that the above modification makes the algorithm ``more powerful''. That is, every algorithm that is not allowed all-neighbors query can be emulated by an algorithm that is allowed this type of query. Therefore this only strengthen our lower bound results. \par

Also note that this modification does not affect the correctness of Lemma~\ref{uniform graph}. We can redefine the function $\alpha_t(\pi)$ to be
\[\alpha_t(\pi)=
\begin{cases}
1 &\mbox{if $q_t(\pi)$ is a pair query}\\
1/\left(\sqrt m-d^{kn}_{\pi^{ \leq t-1}}(u)\right) &\mbox{if $q_t(\pi)=u$ is a random new-neighbor query}\\
1 &\mbox{if $q_t(\pi)$ is an all-neighbors query}
\end{cases}
\;,\]
and the rest of the proof follows as before. \par

% We are now ready to present the auxiliary graph for $k \leq \frac{1}{4}\sqrt m$.

\subsubsect{The auxiliary graph}
For every $t\leq Q$, every query-answer history $\pi$ of length $t-1$ such that $\pi$ is consistent with $G_1$ and every pair $(u,v)$, the witness graphs in $\aux$ are graphs in which $(u,v)$ is either a red pair or a blue pair. There is an edge between a witness graph $W$ and a non-witness graph $\oW$ if the two graphs have the same set of four-tuples except for two matched squares -- one that contains the pair $(u,v)$, $\langle u,v,u',v'\rangle$ and another one. \par

\begin{definition} \label{switch} We define a \textsf{switch between a matched square and an unmatched square} in the following manner. Let $\langle u,v,u',v'\rangle$ be a matched square and $\langle x,y,x',y' \rangle$ be an un matched squares. Informally, a \textsf{switch} between the squares is ``unmatching'' the matched square and instead ``matching'' the unmatched square. \par

Formally, a switch consists of two steps. The first step is removing the edges $(u,v)$ and $(u',v')$ from the red matching $M^C$ and the edges $(u,u')$ and $(v,v')$ from the blue matchings $M^L$ and $M^R$ respectively. The second step is adding the edges $(x,y)$ and $(x',y')$ from the red matching $M^C$ and the edges $(x,x')$ and $(y,y')$ from the blue matchings $M^L$ and $M^R$ respectively. See Figure~\ref{fig:fix_small_k} for an illustration.
\end{definition}

\begin{figure}[h]
	\centering
\ifnum\nofigures=0
	\begin{tikzpicture}[thick, x=2cm, scale=0.7]
	\DrawSmallkFixArg
	\end{tikzpicture}
\fi
	\caption{An illustration of a switch between the squares $\sq$ and $\altSq$.}
	\label{fig:fix_small_k}
\end{figure}

\begin{lemma} \label{Auxilary-graph-degrees k leq sqrt m}
\tchange{Let $\tr=k \cdot \sqrt m$ for an integer $k$ such that $1< k\leq \frac{\sqrt m}{4}$ and let $Q = \frac{m^{3/2}}{600\tr} $.} For every $t\leq Q$, every query-answer history $\pi$ of length $t-1$ such that $\pi$ is consistent with $G_1$ and every pair $(u,v)$,
\[\frac{d_{nw}(\aux)}{d_{w}(\aux)} = \frac{16k}{m} = \tchange{\frac{16\tr}{m^{3/2}}}.\]
\end{lemma}
\begin{proof}
We start with proving that $d_{w}(\aux) \geq \frac{1}{2}m$. A witness graph in $\aux$ with respect to a pair $(u,v)$ is a graph in which $(u,v)$ is part of a matched square $\sq$. Potentially, $\sq$ could be switched with every unmatched square to get a non-witness pair. There are $\sqrt m -k$ unmatched vertices on each side, \dchange{so that there are} $\binom{\sqrt m-k}{2}\cdot \binom{\sqrt m-k}{2}\geq \frac{1}{8}m^2$ potential squares. To get a graph that is in $\mG_2(\pi)$, the unmatched square $\altSq$ must be such that none of the induced pairs between the vertices $x,x',y,y'$ have been observed yet by the algorithm. When all-neighbor queries are allowed, if at most $Q$ queries has been performed, then at most $4Q$ pairs have been observed by the algorithm. Therefore, for at most $4\frac{m}{100k} \leq \frac{1}{4} m$ of the potential squares, an induced pair was queried. Hence, every witness square can be switched with at least $\frac{1}{8}m^2 - \frac{1}{4}m \geq \frac{1}{16}m^2$ consistent unmatched squares, implying that $d_{w}(\aux)\geq \frac{1}{16}m^2$. \par
To complete the proof \dchange{it remains} to show that $d_{nw}(\aux)\leq mk$. \dchange{To this end} we would like to analyze the number of witness graphs that every non-witness $\oW$ can be ``turned'' into. In every non-witness graph $\oW$ the pair $(u,v)$ is unmatched, and in order to turn $\oW$ into a witness graph, one of the $k$ matched squares should be removed and the pair $(u,v)$ with an additional pair $(u',v')$ should be ``matched''. There are $k$ options to remove an existing square, and at most $m$ options to choose a pair $u',v'$ to match $(u,v)$ with. Therefore, the number of potential neighbors of $\oW$ is at most $mk$. It follows that
\[\frac{d_{nw}(\aux)}{d_{w}(\aux)} = \frac{16mk}{m^2} = \frac{16k}{m} = \tchange{\frac{16\tr}{m^{3/2}}},\]
and the proof is complete.
\end{proof}

\dchange{\subsubsect{Statistical distance}}
For an all-neighbors query $q=u$ we say that the corresponding answer is a \textsf{witness answer} if $u\in L$ and $a \neq R$, or symmetrically if $u\in R$ and $a \neq L$. Let $E^{Q}$ be the set of all query-answer histories $\pi$ of length $Q$ such that there exists a query-answer pair $(q,a)$ in $\pi$ in which $q$ is an all-neighbors pair and $a$ is a witness answer with respect to that query, and let $\overline{E}^{Q} = \Pi^Q\setminus E^Q$. That is, $\overline{E}^{Q}$ is the set of all query-answer histories of length $Q$ such that no all-neighbors query is answered with a witness answer. \tchange{Let $\tP_1$ and $\tP_2$ by the induced distributions of the processes $P_1$ and $P_2$ conditioned on the event that the process do not reply with a witness answer.}
%on query-answer histories $\pi$ of length $Q$, conditioned on the event that $\pi \in \overline{E}^Q$. \par
Observe that for every query-answer history $\pi$ of length $t-1$, for every query $q_t$ that is either a pair query or a random new-neighbor query and for every $a\in Ans(\pi,q_t)$,
\[\Pr_{\tP_b}[a \,|\,\pi,q_t] = \Pr_{\tchange{P_b}}[a \,|\,\pi,q_t] .\]
for $b \in \{1,2\}$. Therefore, \dchange{the proof of the next lemma is exactly the
same as the proof of Lemma~\ref{Probability_hitting_red_blue_low}, except that occurrences of the term \tchange{$(\tr/m^{3/2})$ are replaced by $(k/m)$ instead of $(1/\sqrt{m})$} and we apply Lemma~\ref{Auxilary-graph-degrees k leq sqrt m} instead of Lemma~\ref{Auxilary-graph-degrees k= sqrt m}.
}

\medskip
\begin{lemma} \label{Probability_hitting_red_blue_low_k leq sqrt m}
\tchange{Let $\tr=k \cdot \sqrt m$ for an integer $k$ such that $1< k\leq \frac{\sqrt m}{4}$ and let $Q = \frac{m^{3/2}}{600\tr} $.} For every $t\leq Q$, every query-answer history $\pi$ of length $t-1$ such that $\pi$ is consistent with $G_1$ and for every pair or random new-neighbors query $q_t$,
\[\sum\limits_{a \in Ans(\pi,q_t)} \Big|\Pr_{\tP_1}[a \,|\,\pi,q_t] - \Pr_{\tP_2}[a \,|\,\pi,q_t ]\Big| = \frac{96k}{m} = \tchange{\frac{96\tr}{m^{3/2}}}\;.\]
\end{lemma}
Note that Lemma~\ref{Probability_hitting_red_blue_low_k leq sqrt m} does not cover
% all the allowed types of queries
% . The Lemma holds for pair queries and random new-neighbor queries (for vertices $u$ with $d^{kn}_\pi(u) < \frac{1}{2}\sqrt m$) as before,
% and that we need a similar claim regarding all-neighbors queries.
\dchange{all-neighbors queries, and hence we establish the next lemma.}

\begin{lemma} \label{bound dist all-neighbors query}
\tchange{Let $\tr=k \cdot \sqrt m$ for an integer $k$ such that $1< k\leq \frac{\sqrt m}{4}$ and let $Q = \frac{m^{3/2}}{600\tr} $.} For every $t\leq Q$, every query-answer history $\pi$ of length $t-1$ such that $\pi$ is consistent with $G_1$ and for every all-neighbors query $q_t$,
\[\Pr_{\tchange{P_2}}[a_t \mbox{ is a witness answer }\,|\,\pi,q_t] \leq \frac{16k}{\tchange{\sqrt m}} \;.\]
\end{lemma}

\begin{proof}
Assume without loss of generality that $u \in L$. \dchange{By the definition of the process $P_2$, it answers the query consistently with a uniformly selected random graph $G_2 \in \mG_2(\pi)$
 by returning the complete set of $u$'s neighbors in $G_2$.} In $\mG_2$, there are two types of graphs. First, there are graphs in which $u$ is not matched, that is $(u,u') \notin M^L$ for every vertex $u' \in L$. In these graphs the set of $u$'s neighbors is %also
 $\dchange{R=} \{r_1,\ldots,r_{\sqrt m}\}$. We refer to these graphs as {\em non-witness graphs\/}. The second type of graphs are those in which $(u,u') \in M^L$ for some $u'\in L$ and $(u,v)\in M^C$ for some $v\in R$. In these graphs the set of $u$'s neighbors is %$\{u'\} \cup R\setminus\{v\}$.
 \dchange{$(R\setminus\{v\})\cup \{u'\}$}.
 We refer to these graphs as {\em witness graphs\/}. As before, let $Ans(\pi, q_t)$ be the set of all possible answers for an all-neighbors query $q_t$. It holds that
\begin{align*}
\Pr_{\tchange{P_2}}[a_t \mbox{ is a witness answer }\,|\,\pi,q_t] &= \sum\limits_{\substack{a \in Ans(\pi, q_t) \\ a\neq R}} \Pr_{\tchange{P_2}}[a \;\,|\,\pi,q_t ] \\
 &= \sum\limits_{\substack{u'\in L, v\in R }} \frac{\left|\mG_2\left(\pi\circ\left((u,u'),1\right) \circ\left((u,v),0\right) \right)\right|}{|\mG_2(\pi)|} \\
 &= \sum\limits_{\substack{u'\in L }} \frac{\left|\mG_2\left(\pi\circ\left((u,u'),1\right) \right)\right|}{|\mG_2(\pi)|} \cdot \sum\limits_{\substack{v\in R }}\frac{\left|\mG_2\left(\pi\circ\left((u,u'),1\right) \circ\left((u,v),0\right) \right)\right|}{|\mG_2(\pi)|} \\
&= \sum\limits_{\substack{u' \in L}} \frac{\left|\mG_2\left(\pi\circ\left((u,u'),1\right)
 \right)\right|}{|\mG_2(\pi)|} \;.
\end{align*}
Similarly to the proof of Lemma~\ref{Probability_hitting_red_blue_low}, for every $u$ and $u'$ in $L$, $\frac{\left|\mG_2\left(\pi\circ\left((u,u'),1\right)\right)\right|}{|\mG_2(\pi)|} \leq \frac{16k}{m}$. Therefore,
\begin{align*}
\Pr_{\tchange{P_2}}[a_t \mbox{ is a witness answer }\,|\,\pi,q_t]=
\sum\limits_{\substack{u' \in L}} \frac{\left|\mG_2\left(\pi\circ\left((u,u'),1\right)
 \right)\right|}{|\mG_2(\pi)|} \leq \sqrt m \cdot \frac{16k}{m} = \frac{16k}{\sqrt m} \;,
\end{align*}
and the lemma follows.
\end{proof}

It remains to prove that a similar lemma to Lemma~\ref{dist for k=sqrt m} holds for $\sqrt m \leq \tr \leq \frac{1}{4}m$
\dchange{(and the distributions $\mathcal{D}_{1}^{\text{ALG}}$ and $\mathcal{D}_{2}^{\text{ALG}}$ as defined in
this subsection)}.
\begin{lemma} \label{dist for k leq sqrt m} \tchange{Let $\tr=k \cdot \sqrt m$ for an integer $k$ such that $1< k\leq \frac{\sqrt m}{4}$. For every algorithm ALG that performs at most $Q=\frac{m^{3/2}}{600\tr}$ }queries, the statistical distance between $\mathcal{D}_{1}^{\text{ALG}}$ and $\mathcal{D}_{2}^{\text{ALG}}$ is at most $\frac{1}{3}$.
\end{lemma}
\begin{proof}
Let the sets $E^Q$ and $\overline{E}^Q$ be as defined in the beginning of this subsection.
\dchange{By the definition of the statistical distance, and since $\Pr_{\tchange{P_1,\text{ALG}}}[E^{Q}]=0$},
\begin{align*}
d(\mathcal{D}_{1}^{\text{ALG}}, \mathcal{D}_{2}^{\text{ALG}}) &= \frac{1}{2}\left(\sum_{\pi \in E^Q}\Big|\Pr_{\tchange{P_1,\text{ALG}}}[\pi] - \Pr_{\tchange{P_2,\text{ALG}}}[\pi]\Big| +
\sum_{\pi \in \overline{E}^{Q}}\Big|\Pr_{\tchange{P_1,\text{ALG}}}[\pi]- \Pr_{\tchange{P_2,\text{ALG}}}[\pi]\Big| \right) \\
& = \frac{1}{2}\left( \dchange{\Pr_{\tchange{P_2,\text{ALG}}}[E^{Q}]}+
\sum_{\pi \in \overline{E}^{Q}}\Big|\Pr_{\tchange{P_1,\text{ALG}}}[\pi]- \Pr_{\tchange{P_2,\text{ALG}}}[\pi]\Big| \right) \numberthis \label{twice dist}.
\end{align*}
By Lemma~\ref{bound dist all-neighbors query}, the probability of detecting a witness % during
	as a result of an all-neighbors query is at most $\frac{16k}{\sqrt m}$. Since in $Q$ queries, there can be at most $\dchange{4}Q/\sqrt m$ all-neighbors queries, we have that
\[\Pr_{\mD_2^{\text{ALG}}}[E^Q] \leq \frac{1}{6}\;. \numberthis \label{eq:PrD2Eq}\]
We now turn to upper bound the second term.
\dchange{Let $\alpha = \Pr_{\tchange{P_2,\text{ALG}}}[E^Q]$.
\begin{eqnarray}
\sum\limits_{\pi \in \overline{E}^Q}
 \Big|\Pr_{\tchange{P_1,\text{ALG}}}[\pi]- \Pr_{\tchange{P_2,\text{ALG}}}[\pi]\Big|
&=& \sum\limits_{\pi \in \overline{E}^Q}\Big|\Pr_{\tchange{\tP_1,\text{ALG}}}[\pi]\cdot \Pr_{\tchange{P_1,\text{ALG}}}[\overline{E}^Q]
 - \Pr_{\tchange{\tP_2,\text{ALG}}}[\pi]\cdot \Pr_{\tchange{P_2,\text{ALG}}}[\overline{E}^Q] \Big| \nonumber \\
&=& \sum\limits_{\pi \in \overline{E}^Q}\Big|\Pr_{\tchange{\tP_1,\text{ALG}}}[\pi]
     - (1-\alpha)\cdot \Pr_{\tchange{\tP_2,\text{ALG}}}[\pi] \Big| \label{eq:D1OEQ1} \\
&\leq& \sum\limits_{\pi \in \overline{E}^Q}\Big|\Pr_{\tchange{\tP_1,\text{ALG}}}[\pi] -
      \Pr_{\tchange{\tP_2,\text{ALG}}}[\pi] \Big| +
      \alpha\cdot \Pr_{\tchange{\tP_2,\text{ALG}}}[\overline{E}^Q] \nonumber \\
&\leq& \sum\limits_{\pi \in \overline{E}^Q}\Big|\Pr_{\tchange{\tP_1,\text{ALG}}}[\pi] -
      \Pr_{\tchange{\tP_2,\text{ALG}}}[\pi] \Big| +
      \frac{1}{6} \;,\label{eq:tD2OEQ1}
\end{eqnarray}
where in Equation~\eqref{eq:D1OEQ1} we used the fact that $\Pr_{\tchange{P_1,\text{ALG}}}[\overline{E}^Q] =1$,
and in Equation~\eqref{eq:tD2OEQ1} we used the fact that $\Pr_{\tchange{\tP_2,\text{ALG}}}[\overline{E}^Q] =1$
and that $\alpha \leq 1/6$.
}

Therefore, it remains to bound
\[\sum\limits_{\pi \in \overline{E}^Q}\Big|\Pr_{\tchange{\tP_1,\text{ALG}}}[\pi] - \Pr_{\tchange{\tP_2,\text{ALG}}}[\pi]\Big|\;.\]

 Let the hybrid distributions $D_{1,t}^{\text{ALG}}$ for $t \in [Q-1]$
 be as defined in Lemma~\ref{dist for k=sqrt m} (based on the distributions $\mathcal{D}_{1}^{\text{ALG}}$ and $\mathcal{D}_{2}^{\text{ALG}}$ that are induced by the processes $P_1$ and $P_2$ that were defined in this subsection). Also, let $\tmD_{1,t}^{\text{ALG}}$ be the hybrid distribution $\mD_{1,t}^{\text{ALG}}$ conditioned on the event that no all-neighbors query is answered with a witness. That is, $\tmD_{1,t}^{\text{ALG}}$ is the distribution over query-answer histories $\pi$ of length $Q$, where in the length $t$ prefix ALG is answered by the process $P_1$, in the length $Q-t$ suffix ALG is answered by the process $P_2$, and \dchange{each all-neighbors query is answered consistently with $G_1$
 (so that no witness is observed).} % $\pi \in \overline{E}^{Q}$.
 By the above definitions and the triangle inequality,
\begin{equation}
\sum\limits_{\pi \in \overline{E}^Q}\left|\Pr_{\tchange{\tP_1,\text{ALG}}}[\pi]- \Pr_{\tchange{\tP_2,\text{ALG}}}[\pi]\right| \leq \sum\limits_{t}^{Q-1}
 \sum\limits_{\pi \in \overline{E}^Q}\Big|\Pr_{\tmD_{1,t+1}^{\text{ALG}}}[\pi]-\Pr_{\tmD_{1,t}^{\text{ALG}}}[\pi]\Big| \;. \numberthis \label{eq:oEqD1mD2}
\end{equation}
As in the proof of Lemma~\ref{dist for k=sqrt m} we have that for every $t \in [Q-1]$,
\begin{align}
 \sum\limits_{\substack{\pi \in \overline{E}^Q }} &\Big|\Pr_{\tmD_{1,t+1}^{\text{ALG}}}[\pi]-\Pr_{\tmD_{1,t}^{\text{ALG}}}[\pi]\Big|\nonumber \\
&=
\sum\limits_{\substack{\pi'=\pi_1,\ldots,\pi_{t-1},q_t :\\ \pi' \in \overline{E}^{t-1}}} \Pr_{\tP_{1},{\text{ALG}}}[\pi',q_t]
 \cdot \sum\limits_{\substack{a \in Ans(\pi', q_t) : \\ \pi'\circ(q_t,a) \in \overline{E}^t}} \Big|\Pr_{\tP_{1}}[a\,|\,\pi',q_t] - \Pr_{\tP_2}[a\,|\,\pi',q_t]\Big| \;.
 \label{eq:sum-pi-t-t1-k-leq-sqrtm}
\end{align}
By Lemma~\ref{Probability_hitting_red_blue_low_k leq sqrt m}
\dchange{(and since for an all-neighbor query $q_t$
we have that the (unique) answer according to $\tP_2$ is the same as according to $\tP_1$)},
 \[\sum\limits_{\substack{a \in Ans(\pi',q_t) : \\ \pi'\circ(q_t,a) \in \overline{E}^t}} \Big|\Pr_{\tP_{1}}[a\,|\,\pi',q_t] - \Pr_{\tP_2}[a\,|\,\pi',q_t]\Big| \leq \frac{96k}{m} = \tchange{\frac{96\tr}{m^{3/2}}},\]
and it follows that
\[ \sum\limits_{\substack{\pi \in \overline{E}^Q }} \Big|\Pr_{\tmD_{1,t+1}^{\text{ALG}}}[\pi]-\Pr_{\tmD_{1,t}^{\text{ALG}}}[\pi]\Big| \leq \frac{96k}{m} = \tchange{\frac{96\tr}{m^{3/2}}}.\]
Hence, for $Q=\frac{m^{3/2}}{600\tr}$,
\begin{equation}
\sum\limits_{t}^{Q-1} \sum\limits_{\pi \in \overline{E}^Q}\Big|\Pr_{\tmD_{1,t+1}^{\text{ALG}}}[\pi]-\Pr_{\tmD_{1,t}^{\text{ALG}}}[\pi]\Big| \leq Q \cdot \frac{48\tr}{m^{3/2}} \leq \frac{1}{6}\;. \label{last}
\end{equation}
Combining Equations~\eqref{twice dist}, \eqref{eq:PrD2Eq}, \eqref{eq:tD2OEQ1}, \eqref{eq:oEqD1mD2} and \eqref{last}, we get
\begin{align}
d(\mD_1^{\text{ALG}},\mD_2^{\text{ALG}}) \leq \frac{1}{2}\left(\frac{1}{6} + \frac{1}{6} + \frac{1}{6}\right) \leq \frac{1}{3},
\end{align}
and the proof is complete.
\end{proof}

\subsect{Lower Bound for $\tr<\frac{1}{4}\sqrt m$.} \label{section Delta sqrt m}
\subsubsect{The construction}
In this case the basic structure of $G_1$ and $\mG_2$ is a bit different. Also, for the sake of simplicity, we present \dchange{graphs} with $2m$ edges, and either $0$ or $4\tr$ triangles.
The graph $G_1$ has three components -- \textsf{two} complete bipartite graphs, each over $2\sqrt m$ vertices, and an independent set of size $n-4\sqrt m$. Let $A$ and $B$ be the left-hand side and the right-hand side sets, respectively, of the first bipartite component, and $C$ and $D$ of the second one. We refer to the edges between $A$ and $B$ and the edges between $C$ and $D$ as \textsf{black edges}. We divide each of these sets into $\frac{\sqrt m}{\tr}$ subsets of size $\tr$, denoted $\{ \Lambda_1, \ldots ,\Lambda_{\frac{\sqrt m}{\tr}} \}$ for $\Lambda \in \{A,B,C,D\}$. For every $1\leq i\leq \frac{\sqrt m}{\tr}$, we first remove a complete bipartite graph between $A_i$ and $B_i$ and between $C_i$ and $D_i$, and refer to the removed edges as \textsf{red edges}. We then add a complete bipartite graph between $B_i$ and $C_i$ and between $D_i$ and $A_i$, and refer to added edges as \textsf{blue edges}. Note that this maintains the degrees of all the vertices to be $\sqrt m$.
% \dchange{The graph obtained is $G_1$.}

In % the second family
$\mG_2$ the basic structure of all the graphs is the same as of $G_1$ with the following modifications. Each graph is defined by the choice of four ``special" vertices $a^*, b^*, c^*, d^*$ such that $a^*\in A_{i_{a^*}}, b^* \in B_{i_{b^*}},c^*\in C_{i_{c^*}}$ and $d^* \in D_{i_{d^*}}$ for some indices $i_{a^*},i_{b^*},i_{c^*}$ and $i_{d^*}$ such that no two indices are equal. We then add edges $(a^*,c^*)$ and $(b^*,d^*)$, referred to as \textsf{green} edges, and remove edges $(a^*,b^*)$ and $(c^*, d^*)$, referred to as \textsf{purple edges}. We also refer to the green and purple edges as \textsf{special edges}. Note that we add one edge and remove one edge from each special vertex, thus maintaining their initial degrees. See Figure~\ref{fig:lb-for-small-Delta}. \par
\begin{figure}[h]
	\centering
\ifnum\nofigures=0
	\begin{tikzpicture}[yscale=0.7]
	\DrawDoubleBipartiteMain
	%\item[] \DrawDoubleBipartiteTrianglesMain
	\end{tikzpicture}
\fi
	\caption{An illustration of a graph in $\mathcal{G}_2$. The broken thin (red) edges describe edges that were removed and the thin (blue) edges describe edges that were added. The broken thick (purple) edges describe the special non-edges $(a^*, b^*)$ and $(c^*, d^*)$. The curly (green) edges describe the special edges $(a^*, c^*)$ and $(b^*, d^*)$.}
	\label{fig:lb-for-small-Delta}
\end{figure}

We first prove that %for every graph in $\mG_1$, $\tr(G)=0$,
\dchange{$\tr(G_1)=0$}
and then that for every graph $G$ in $\mG_2$, $\tr(G)=4\tr.$
\begin{claim}
\tchange{The graph $G_1$ has no triangles.}
\end{claim}
\begin{proof} Consider an edge $(u,v)$ in $G_1$. First assume $u$ and $v$ are connected by a black edge, that is, they are on different sides of the same bipartite component. Hence we can assume without loss of generality that $u\in A$ and that $v \in B$. Since $u$ is in $A$ it is only connected to vertices in $B$ or vertices in $D$. Since $v$ is in $B$ it is only connected to vertices in $A$ or vertices in $C$. Thus $u$ and $v$ cannot have a common neighbor. A similar analysis can be done for a pair $(u,v)$ that is connected by a blue edge. Therefore $\tr(G)$ is indeed zero as claimed.
\end{proof}

\begin{claim} For every graph $G\in \mG_2$, $\tr(G)=4\tr$.
\end{claim}
\begin{proof} Since the only differences between %graphs in $\mG_1$ and
\dchange{$G_1$ and graphs in} $\mG_2$ are the two added green edges and the two removed red edges, any triangle in $\mG_2$ must include a green edge. Therefore we can count all the triangles that the green edges form. Consider the green edge $(a^*,c^*)$ and recall that $a^*$ is in $A_{i_{a^*}}$ and $c^*$ is in $C_{i_{c^*}}$. The only common neighbors of $(a^*,c^*)$ are all the vertices in $B_{i_{c^*}}$ and all the vertices in $D_{i_{a^*}}$. A vertex $v$ such that $v \notin B_{i_{c^*}}$ and $v \notin D_{i_{a^*}}$ is either (1) in $A$ or in $D\setminus D_{i_{a^*}}$, in which case it is not a neighbor of $a^*$, or it is (2) in $C$ or in $B\setminus B_{i_{c^*}}$, in which case it is not a neighbor of $c^*$. Since both $B_{i_{c^*}}$ and $D_{i_{a^*}}$ are of size $\tr$, the edge $(a^*,c^*)$ participates in $2\tr$ triangles. Similarly the edge $(b^*,d^*)$ participate in $2\tr$ triangles, and together we get that $\tr(G)=4\tr$, as claimed.
\end{proof}

\dchange{\subsubsect{The processes $P_1$ and $P_2$}
The definition of the processes $P_1$ and $P_2$ is the same as in Subsection~\ref{processes all neighbors query}
 (using the modified definitions of $G_1$ and $\mG_2$).
}

\subsubsect{The auxiliary graph}
We define a switch for this case as well. Informally, a switch between a matched pair $(u^*,v^*)$
and an unmatched pair $(u,v)$ is ``unmatching'' $(u^*,v^*)$ and ``matching'' $(u,v)$ instead. Formally stating we define a switch as follows.

\begin{definition} A \textsf{switch} between a green pair $(a^*,c^*)$ and a pair $(a,c)$ such that $a\in A_i$, $c\in C_j$ and none of the indices $i,j,i_{b^*},i_{d^*}$ are equal, is the following two steps process. In the first step we ``unmatch'' $(a^*,c^*)$ by removing the green edge $(a^*,c^*)$ and adding the edges $(a^*,b^*)$ and $(c^*,d^*)$. In the second step we ``match'' $(a,c)$ by adding the green edge $(a,c)$ and removing the edges $(a,b^*)$ and $(c,d^*)$. A switch with the pair $(b^*,d^*)$ can be defined in a similar manner. \par
\end{definition}

\begin{figure}[h]
	\centering
\ifnum\nofigures=0
	\begin{tikzpicture}[yscale=0.7]
	\DrawFixSmallDelta
	%\item[] \DrawDoubleBipartiteTrianglesMain
	\end{tikzpicture}
\fi
	\caption{An illustration of a switch between the pairs $(a^*,c^*)$ and $(a,c)$.}
	\label{fig:fix-for-small-Delta}
\end{figure}
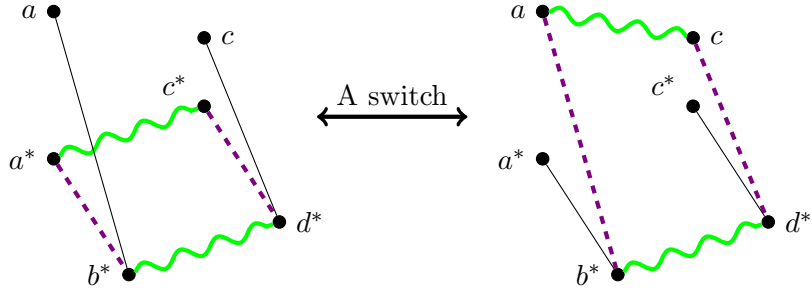

Let $\tr < \sqrt m$ and let $Q=\frac{m}{600}$. For every $t\leq Q$, every query-answer history $\pi$ of length $t-1$ and every pair $(u,v)$ we define the following auxiliary graph. The witness nodes are graphs in which $(u,v)$ is one of the four special pairs. If the pair is a green matched pair then there is an edge in the auxiliary graph between a witness graph $W$ and a non-witness graph $\oW$, if $\oW$ can be obtained from $W$ by a single switch between $(u,v)$ and another unmatched pair.

\begin{lemma} \label{Auxilary-graph-degrees Delta leq sqrt m}
For $\tr < \frac{1}{4} \sqrt m$ let $Q=\frac{m}{600}$. For every $t\leq Q$, every query-answer history $\pi$ of length $t-1$ such that $\pi$ is consistent with $G_1$ and every pair $(u,v)$,
\[\frac{d_{nw}(\aux)}{d_{w}(\aux)} = \frac{8}{m}.\]
\end{lemma}
\begin{proof}
We analyze the case where the pair $(u,v)$ is such that $u\in A$ and $v \in C$, as the proof for the other cases is almost identical. We first prove that $d_{w}(\aux) \geq \frac{1}{8}m$. A witness graph $W$ is a graph in which $(a,c)$ is a special pair. That is $(u,v) = (a^*,c^*)$. Potentially, for every pair $(a',c')$ such that $a'\in A_i$, $c'\in C_j$ and none of the indices $i,j,i_{b^*},i_{d^*}$ are equal, the graph resulting from a switch between $(a^*,c^*)$ and $(a',c')$ is a non-witness graph. There are $\sqrt m-2\tr$ vertices $a'$ in $A\setminus (A_{i_{b^*}}\cup A_{i_{d^*}})$ and for each such $a'$ there are $\sqrt m-3\tr$ vertices $c'$ in $C\setminus (C_{i_{b^*}}\cup C_{i_{d^*}} \cup C_{i_{a'}}))$. Since $\tr< \frac{1}{4}\sqrt m$, there are at least $(\sqrt m-2\tr)\cdot(\sqrt m-3\tr) = m-6\tr^2 \geq \frac{1}{4}m$ potential pairs $(a',c')$ that $(a^*,c^*)$ could be switched with. For the resulting graph to be consistent, that is, to be in $\mG_2(\pi)$, the pair $(a',c')$ must be such that the pairs $(a',c')$, $(a^*,b^*)$ and $(c^*,d^*)$ have not been observed yet by the algorithm. Since the number of queries is at most $\frac{1}{600}m$, at least $\frac{1}{4}m - \frac{1}{125}m \geq \frac{1}{8}m$ of the potential pairs $(a',c')$ can be switched with $(a^*,c^*)$ such that the resulting graph is consistent with $\mG_2(\pi)$. Therefore, $d_{w}(\aux)\geq \frac{1}{8}m$. \par
Now consider a non-witness graph $\oW$. There is only one possibility to turn $\oW$ into a witness graph, which is to switch the pair $(u,v)$ with the green pair $(a^*,c^*)$. Therefore, the maximal degree of every non-witness graph, $d_{nw}(\aux)$, is $1$. \par
Together we get that
\[\frac{d_{nw}(\aux)}{d_{w}(\aux)} \leq \frac{8}{m},\]
and the proof is complete.
\end{proof}
\dchange{\subsubsect{Statistical distance}}
A similar proof to the ones of Lemma~\ref{Probability_hitting_red_blue_low_k leq sqrt m} and Lemma~\ref{bound dist all-neighbors query} using Lemma~\ref{Auxilary-graph-degrees Delta leq sqrt m} gives the following lemmas for the case that $1 \leq \tr < \frac{1}{4}\sqrt m$.

\begin{lemma} \label{bound dist all-neighbors query delta leq sqrt m}
Let $1 \leq \tr < \frac{1}{4}\sqrt m$ and $Q=\frac{m}{600}$. For every $t\leq Q$, every query-answer history $\pi$ of length $t-1$ such that $\pi$ is consistent with $G_1$ and for every all-neighbors query $q_t$,
\[\Pr_{\tchange{P_2}}[a_t \mbox{ is a witness answer }\,|\,\pi,q_t] \leq \frac{16}{m}\;.\]
\end{lemma}
\begin{lemma} \label{Probability_hitting_red_blue_low delta leq sqrt m}
Let $1 \leq \tr < \frac{1}{4}\sqrt m$ and $Q = \frac{m}{600}$. For every $t\leq Q$, every query-answer history $\pi$ of length $t-1$ such that $\pi$ is consistent with $G_1$ and for every pair or random new-neighbors query $q_t$,
\[\sum\limits_{a \in Ans(\pi,q_t)} \Big|\Pr_{\tP_1}[a \,|\,\pi,q_t] - \Pr_{\tP_2}[a \,|\,\pi,q_t ]\Big| = \frac{96}{m}\;.\]
\end{lemma}
The next lemma is proven in a similar way to 1.3.4 based on the above two lemma.
\begin{lemma} \label{dist for Delta leq sqrt m} Let $1 \leq \tr < \frac{1}{4}\sqrt m$. For every algorithm ALG that asks at most $Q=\frac{m}{600}$, the statistical distance between $\mathcal{D}_{1}^{\text{ALG}}$ and $\mathcal{D}_{2}^{\text{ALG}}$ is at most $\frac{1}{3}$.
\end{lemma}
%%%%
\iffalse
This implies that for a graph $G$, if $\tr(G) <\frac{1}{4}\sqrt m$, then $\Omega(m)$ queries are required to approximate $\tr(G)$ up to a multiplicative factor. From Lemmas~\ref{dist for k=sqrt m},~\ref{dist for k = r sqrt m} and~\ref{dist for k leq sqrt m}, if $\tr(G) \geq \sqrt m$, then the number of required queries is $\Omega\left(\frac{m^{3/2}}{\tr(G)}\right)$. Hence, any triangles approximation algorithm that succeeds with probability at least $2/3$ is required to perform $\Omega\left(\min\left\{\frac{m^{3/2}}{\tr(G)},m\right\}\right)$ queries as stated in Theorem~\ref{m 3/2 over delta}.
\fi
%%%

\subsect{Wrapping things up}
Theorem \ref{m 3/2 over delta} follows from Lemmas~\ref{dist for k=sqrt m}, \ref{dist for k = r sqrt m}, \ref{dist for k leq sqrt m} and \ref{dist for Delta leq sqrt m}, and the next corollary is proved using Theorems \ref{m 3/2 over delta} and \ref{n over delta third}.

\begin{corollary}
\label{n over delta third cor} Any multiplicative-approximation algorithm for the number of triangles in a graph must perform $\Omega\left(\frac{n}{\tr(G)^{1/3}} + \min\left\{m, \frac{m^{3/2}}{\tr(G)}\right\}\right)$ queries, where the allowed queries are degree queries, pair queries and neighbor queries.
\end{corollary}

\begin{proof}
Assume towards a contradiction that there exists an algorithm ALG' for which the following holds:
\begin{enumerate}
\item ALG' is allowed to ask neighbor queries as well as degree queries and pair queries.
\item ALG' asks $Q'$ queries.
\item ALG' outputs a $(1\pm\eps)$-approximation to the number of triangles of any graph $G$ with probability greater than $2/3$.
\end{enumerate}
Using ALG' we can define an algorithm ALG that is allowed random new-neighbor queries, performs at most $Q=3Q'$ queries and answers correctly with the same probability as ALG' does. ALG runs ALG' and whenever ALG' performs a query $q'_t$, ALG does as follows:
\begin{itemize}
\item If $q'_t$ is a degree query, ALG performs the same query and sets $a'_t=a_t$.
\item If $q'_t$ is a pair query $(u,v)$, then ALG performs the same query $q=q'$. Let $a_t$ be the corresponding answer.
\begin{itemize}
\item If $a_t=0$, then ALG sets $a'_t=a_t$.
\item If $a_t=1$, then ALG sets $a'_t=(a_t,i,j)$, such that $i$ and $j$ are randomly chosen labels that have not been previously used for neighbors of $u$ and $v$, and are within the ranges $\left[1..d(u)\right]$ and $\left[1..d_v\right]$ respectively.
\end{itemize}
\item If $q'_t$ is a neighbor query $(u,i)$, ALG performs a random new-neighbor query $q_t=u$, and returns the same answer $a'_t=a_t$.
\end{itemize}
We note that the above requires the algorithm ALG to store for every vertex $v$, all the labels used for its neighbors in the previous steps. Once ALG' outputs an answer, ALG outputs the same answer. It follows that ALG performs at most $3Q$ queries to the graph $G$. By the third assumption above, ALG outputs a $(1\pm\eps)$-approximation to the number of triangles of any graph $G$ with probability greater than $2/3$. If $Q' \notin \Omega\left(\frac{n}{\tr(G)^{1/3}} + \min\left\{m, \frac{m^{3/2}}{\tr(G)}\right\}\right)$ then $Q \notin \Omega\left(\frac{n}{\tr(G)^{1/3}} + \min\left\{m, \frac{m^{3/2}}{\tr(G)}\right\}\right)$ which is a contradiction to Theorem~\ref{n over delta third} and Theorem~\ref{m 3/2 over delta}.
\end{proof}

\fi